\documentclass[12pt]{article}
\usepackage{newinutile}
\usepackage{amsmath,amssymb,amsthm,amsfonts}
\usepackage{epsfig,graphics,color,calc,graphicx}
\usepackage{subfigure}
\usepackage{comment}
\usepackage{pgf,tikz}
\usepackage{mathrsfs}
\usepackage{color}
\usetikzlibrary{arrows}
\usepackage{mathrsfs}
\usepackage{hyperref}
\hypersetup{colorlinks=true,linkcolor=blue}
\usepackage{cite}
\usepackage{mathrsfs}

\newtheorem{theorem}{Theorem}

\newtheorem{lemma}{Lemma}

\newtheorem{proposition}{Proposition}

\newcommand{\rr}[1]{{\normalfont\textrm{#1}}}
\newcommand{\ud}{\mathrm{d}}

\newcommand{\cc}[1]{{\mathcal{#1}}}
\newcommand{\bb}[1]{{\mathbb{#1}}}

\newlength{\pecettawidth}
\begin{document}
\title{Linear Boltzmann dynamics in a strip with large reflective
obstacles: stationary state and residence time}

\author{Alessandro Ciallella and Emilio N.M.\ Cirillo}
\email{alessandro.ciallella@sbai.uniroma1.it, emilio.cirillo@uniroma1.it}
\affiliation{Dipartimento di Scienze di Base e Applicate per l'Ingegneria, 
             Sapienza Universit\`a di Roma, 
             via A.\ Scarpa 16, I--00161, Roma, Italy.}



\begin{abstract}
The presence of obstacles modifies the way in which 
particles diffuse. In cells, for instance, it is observed that, due to the 
presence of macromolecules playing the role of obstacles, 
the mean-square displacement of biomolecules scales as a power law 
with exponent smaller than one. On the other hand, different 
situations in grain and pedestrian dynamics in which the presence of
an obstacle accelerate the dynamics are known.
We focus on the time, called the residence time, needed by particles to 
cross a strip assuming that the dynamics inside the strip follows the 
linear Boltzmann dynamics. 
We find that the residence time is not monotonic with respect 
to the size and the location of the obstacles, since the obstacle 
can force those particles that eventually cross the strip to 
spend a smaller time in the strip itself.
We focus on the case of a rectangular strip with two open sides and two 
reflective sides and we consider reflective obstacles into the strip.
We prove that the stationary state of the linear Boltzmann dynamics,
in the diffusive regime,
converges to the solution of the Laplace equation with Dirichlet boundary 
conditions on the open sides and homogeneous Neumann boundary conditions on the other sides and on the obstacle boundaries.
\end{abstract}




\maketitle

\section{Introduction} \label{ENMC:sec:introduction}
\par\noindent
Particle based studies of agent behavior can reproduce 
realistic scenarios. Many different real situations, 
such as grains discharging from a silo, 
traffic flow, and pedestrian dynamics, 
can be studied via particle based modelling. 
The focus, in this paper, is on the effect of obstacles 
on particle flows \cite{CPamm2017}. 

It is well known that the mean square distance 
traveled by particles undergoing Brownian motion 
is proportional to time. However, many experimental 
measures of molecular diffusion in cells show a sublinear 
behavior. This phenomenon, called \emph{anomalous diffusion},
is in some cases explained as an effect of the presence of macromolecules 
playing the role of obstacles for diffusing smaller molecules
\cite{Sbj1994,HFrpp2013,MHSbj2017,ESMCBjcp2014}

While in the case of anomalous diffusion the obstacles induce
a sort of slowing down of the dynamics, on the contrary
in many other different contexts it has been noticed 
and exploited the fact that the presence of an obstacle 
can surprisingly accelerate the dynamics 
improving, in particular, 
the flux rate of particles passing through a bottleneck.

In granular flows, for instance when grains discharge from a silo, 
it happens that the out-coming flow can be dramatically reduced
due to clogging at the exit. In \cite{TLPprl2001}
it was proposed that this phenomenon is caused by the formation of arches.
In the case of spherical grains it has been demonstrated that 
the presence of clogging in a three--dimensional silo 
critically depends on the ratio between 
the outlet size and the diameter of the particles \cite{ZGMPPpre2005}.  A solution that is implemented to improve the granular flow is 
to place an obstacle above the silo exit
\cite{AMAGTOKpre2012,ZJGLAMprl2011} 
which prevents arches to be formed or to become stable. 

A similar phenomenon is observed in pedestrian flows
(see, e.g., \cite{Hrmp2001,BDsiamr2011,HMFBepBpd2001,HFMV2011}
for reviews of models and related problems).
In the case of pedestrian exiting a room under panic
clogging at the door can be reduced by means of suitably 
positioned obstacles, see \cite[Section~6.3]{ABCKsjap2016}
and \cite{HFVn2000,HBJWts2005}. It has also been noticed that 
a correct positioning can reduce injuries under 
panicked escape from a room thanks to the so called ``waiting--room''
effect \cite{EDLR2003}: pedestrians slow down and accumulate close to 
the obstacle so that the exit is decongested. 
We mention, here, that also the possibility of clustering far from 
the exit due to individual cooperation has been object of study 
in \cite{CMpA2013}.

The a priori unexpected phenomena discussed above are  
a sort of inversion of the Braess' paradox \cite{BNWts2005,Harfm2003}
stating that adding a road link to a road network can cause 
cars to take longer to cross the network. Indeed, in the examples 
discussed above it seems that adding barriers results in a decrease 
of the time that particles need to cross a region of space. 

This is precisely the issue we address in this paper.
Inspired by \cite{CKMSpre2016, CKMSSpA2016}, 
we consider particles entering an horizontal strip through the 
left boundary and eventually exiting through the right 
one \cite{FPvSpre2017}. 
We compute the typical time needed to cross the 
strip, called the \emph{residence time}, and 
analyze its dependence on the size and position of a reflecting 
 obstacle positioned inside the strip. 
Surprisingly, we find not monotonic behaviors of the residence time 
as a function of the side lengths of the obstacle and the 
coordinate of its center. In particular, 
for suitably choices of the obstacles the residence time in presence 
of the barrier is smaller than the one measured for the empty 
strip. In other words, our results show that placing a suitable obstacle
in the strip allows to select those particles that cross the strip 
in a shorter time.  We also observe that the same obstacle, placed in different position, can either increase or decrease the residence time.

Inside the strip we consider particle moving according to the Markov process solving the linear Boltzmann equation.
We stress that this is the first 
study of  residence times by means of the methods of Kinetic Theory.
In this case we calculate the residence time by directly simulating the motion of single particles by a Monte Carlo method.
This dynamics should be consistent with the Lorentz process 
at appropriate regimes.  
The Lorentz gas model is a system of 
non--interacting particles moving in a region where static small disks are 
distributed according to a Poisson probability measure. 
This is a classical model for finite velocity random motions. 
Particles perform uniform linear motion up to the contact with a disk 
where they are elastically 
reflected. 

We study the system in a 
diffusive regime and we show that there exists a unique stationary solution 
which converges to the solution of a Laplace problem with mixed boundary 
conditions: Dirichlet boundary conditions on the vertical sides and homogeneous Neumann on the rest. 
We prove the convergence and check it numerically.
Moreover, by constructing numerically the stationary profiles we qualitatively verify that the overall flux in presence of obstacles is decreasing, as expected by physical intuition. This holds even in those cases in which the residence time is smaller with respect to the empty strip case.
 
In Section \ref{s:modello} we introduce the model under investigation 
and  we state our main theoretical results. 
In Section \ref{s:simul} we propose a Monte Carlo algorithm that we use to construct 
the stationary state of the linear Boltzmann model and we discuss 
the stationary profile in presence of large fixed obstacles.
In Section \ref{s:residence} we discuss our results on residence time. 
Finally, Section \ref{s:dimo} contains the proofs of the results we state 
in Section \ref{s:modello}.

\section{Model and results}
\label{s:modello}
We consider a system of light particles moving in the two--dimensional space. We choose as the domain  a subset $\Omega$ of the finite strip 
$(0,L_1)\times(0,L_2)\subset\bb{R}^2$. 
This strip has two open boundaries, that we think as the left side 
$\partial\Omega_L=\{0\}\times(0,L_2)$ and 
the right side $\partial\Omega_R=\{L_1\}\times(0,L_2)$. 
The strip is in contact on the left side and on the right side with two 
mass reservoirs at equilibrium with particle mass densities $\rho_L$ and 
$\rho_R$, respectively.
Particles traveling into $\Omega$ are instead  specularly reflected upon 
colliding with  the upper side $(0,L_1)\times\{L_2\}$ and lower 
side $(0,L_1)\times\{0\}$ of the strip.

We consider the case in which large fixed obstacles are placed in the strip so that the domain $\Omega$ is a connected set. These obstacles are convex
 sets with  smooth reflective
  boundaries.
We consider a generic   configuration of a finite number of obstacles with positive mutual distance and positive distance from the sides of the strip.
 In the sequel we will call $\partial\Omega_E$ the union of the obstacle boundaries and the upper and lower sides of the strip (see Figure \ref{fig:2.1_domain}).
Therefore, when a particle reaches $\partial\Omega_E$ it experiences a specular reflection.

\begin{figure}[ht!]
\centering
\definecolor{wqwqwq}{rgb}{0.3764705882352941,0.3764705882352941,0.3764705882352941}
\definecolor{qqqqff}{rgb}{0.,0.,1.}
\definecolor{ffqqqq}{rgb}{1.,0.,0.}
\definecolor{eqeqeq}{rgb}{0.8784313725490196,0.8784313725490196,0.8784313725490196}
\begin{tikzpicture}[line cap=round,line join=round,>=triangle 45,x=1.cm,y=1.cm]
\clip(-1.2,-2.5) rectangle (10.2,3.75);
\fill[line width=1.2pt,color=eqeqeq,fill=eqeqeq,fill opacity=0.10000000149011612] (0.,0.) -- (0.,3.) -- (9.,3.) -- (9.,0.) -- cycle;
\fill[fill=black,fill opacity=0.10000000149011612] (9.,3.) -- (10.,3.) -- (10.,0.) -- (9.,0.) -- cycle;
\fill[fill=black,fill opacity=0.10000000149011612] (0.,3.) -- (-1.,3.) -- (-1.,0.) -- (0.,0.) -- cycle;
\draw [line width=1.pt,dashed,color=gray] (0.,0.)-- (0.,3.);
\draw [line width=1.6pt,color=black] (0.,3.)-- (9.,3.);
\draw [line width=1.pt,dashed,color=gray] (9.,3.)-- (9.,0.);
\draw [line width=1.6pt,color=black] (9.,0.)-- (0.,0.);
\draw (9.,3.)-- (10.,3.);
\draw (10.,3.)-- (10.,0.);
\draw (10.,0.)-- (9.,0.);
\draw (0.,3.)-- (-1.,3.);
\draw (-1.,3.)-- (-1.,0.);
\draw (-1.,0.)-- (0.,0.);
\draw [rotate around={73.90918365114764:(7.541710094234184,1.8497761157344577)},line width=1.6pt,color=black,fill=qqqqff,fill opacity=0.10000000149011612] (7.541710094234184,1.8497761157344577) ellipse (0.5cm and 0.4cm);

\draw [rotate around={15:(3.3,1.2)},line width=1.6pt,color=black,fill=qqqqff,fill opacity=0.10000000149011612] (3.3,1.2) ellipse (0.3cm and 0.5cm);



\begin{scriptsize}
\draw[color=black] (0.40012375895827857,1.100793147323) node {$\partial \Omega_L$};
\draw[color=black] (8.6000375895827857,1.100698793147323) node {$\partial \Omega_R$};
\draw[color=black] (4.003540388846925,.358072978331086) node {$\partial \Omega_E$};
\draw[color=black] (5.0517845286029,2.07804558281278) node {$\Omega$};
\draw[color=black] (9.715404397846223,1.50089858490220897) node {$\rho_R$};
\draw[color=black] (-0.7009588291303488,1.5006705642026902) node {$\rho_L$};
\end{scriptsize}
\end{tikzpicture}
\vspace{-2.3cm}
\caption{Domain $\Omega$: strip with large fixed obstacles, where $\partial\Omega_L$ and $\partial\Omega_R$ are the vertical open boundaries and $\partial\Omega_E$ are reflective boundaries.}\label{fig:2.1_domain}
	\end{figure}
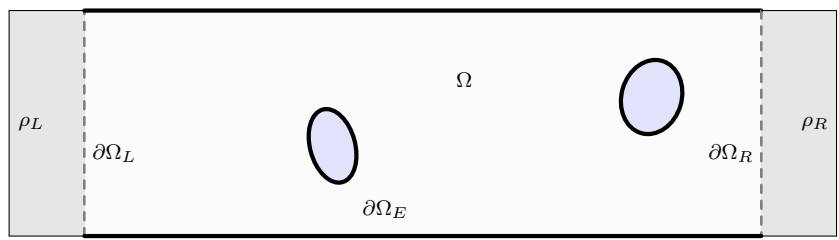

The linear
Boltzmann equation is a kinetic linear equation, 
 combining free transport and scattering off of a medium. 
This equation consists of two terms: a free transport term and a collision operator $\cc{L}$.

Let us consider  the phase space $\Omega\times S^1$, where  $S^1:=\{ v\in \bb{R}^2 : |v|=1 \}$. We will consider the  operator $\cc{L}$ with elastic collision kernel.
 The equation for  $(x,v)\in{\Omega}\times S^1$ and positive times $t$ reads
\begin{equation}\label{eq:boltz_lin}
(\partial_t + v\cdot \nabla_x) g(x,v,t) 
= \eta_\varepsilon \cc{L} g(x,v,t), \qquad x\in\Omega,\, v\in S^1,\,t\geq 0
\end{equation}
where, by the elastic collision rule $v'=v- 2(n\cdot v ) n$, the operator $\cc{L}$ is defined  for any $ f\in L^1(S^1)$ as
\begin{equation} \label{eq:op_L}
\cc{L} f(v) 
=\lambda \int_{-1}^1 [f(v') - f(v)] \,\mathrm{d} \delta.
\end{equation}
Here $n=n(\delta)$ is the outward pointing normal to a circular scatterer of radius $1$ at the point of collision among the particle with velocity $v$ and the scatterer. So $\delta=\sin \alpha$ if $\alpha$ is the angle of incidence between $v$ and $n$  that has $\delta$ as impact parameter (see Figure \ref{fig:bol_op}); $\lambda>0$ is a fixed parameter.
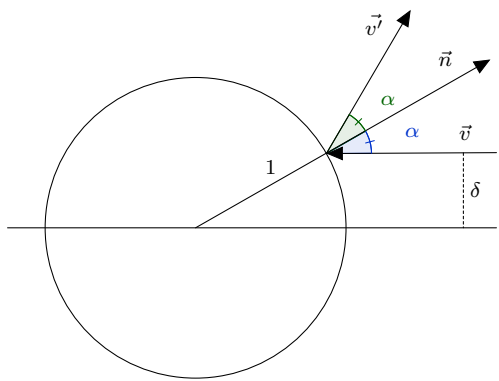
\begin{figure}[ht!]
\centering
\definecolor{xdxdff}{rgb}{0.49019607843137253,0.49019607843137253,1.}
\definecolor{qqwuqq}{rgb}{0.,0.39215686274509803,0.}
\definecolor{qqttcc}{rgb}{0.,0.2,0.8}
\begin{tikzpicture}[line cap=round,line join=round,>=triangle 45,x=1.cm,y=1.cm]
\clip(4.5,2.5) rectangle (11.,8.);
\draw [shift={(8.738913114785536,5.988018815219054)},color=qqttcc,fill=qqttcc,fill opacity=0.10000000149011612] (0,0) -- (0.30359952163475573:0.6) arc (0.30359952163475573:29.856561195427133:0.6) -- cycle;
\draw [shift={(8.738913114785536,5.988018815219054)},color=qqwuqq,fill=qqwuqq,fill opacity=0.10000000149011612] (0,0) -- (29.856561195427123:0.6) arc (29.856561195427123:59.6148453145519:0.6) -- cycle;
\draw(7.,5.) circle (2.cm);
\draw [->] (7.,5.) -- (10.92,7.24);
\draw [->] (11.,6.) -- (8.738913114785536,5.988018815219054);
\draw [shift={(8.738913114785536,5.988018815219054)},color=qqttcc] (0.30359952163475573:0.6) arc (0.30359952163475573:29.856561195427133:0.6);
\draw[color=qqttcc] (9.260317210598283,6.128509985296365) -- (9.376184787445561,6.159730245313545);
\draw [->] (8.738913114785536,5.988018815219054) -- (9.86,7.9);
\draw [shift={(8.738913114785536,5.988018815219054)},color=qqwuqq] (29.856561195427123:0.6) arc (29.856561195427123:59.6148453145519:0.6);
\draw[color=qqwuqq] (9.122508067193936,6.368091061623338) -- (9.207751389951358,6.452551560824289);
\draw [domain=4.5:11.] plot(\x,{(--25.-0.*\x)/5.});
\draw [dash pattern=on 1pt off 1pt] (10.560012353940715,5.997668566686561)-- (10.56,5.);
\begin{scriptsize}
\draw[color=black] (10.32,7.25) node {$\vec{n}$};
\draw[color=black] (10.58,6.27) node {$\vec{v}$};
\draw[color=qqttcc] (9.88,6.27) node {$\alpha$};
\draw[color=black] (9.38,7.71) node {$\vec{v'}$};
\draw[color=qqwuqq] (9.56,6.69) node {$\alpha$};
\draw[color=black] (8.,5.8) node {$1$};
\draw[color=black] (10.72,5.51) node {$\delta$};
\end{scriptsize}
\end{tikzpicture}
\caption{Elastic collision with a scatterers: impact parameter $\delta$ and angle of incidence $\alpha$.}
\label{fig:bol_op}
\end{figure}

We denote by $g_\varepsilon$ the solution of the equation   corresponding to the value of $\eta_\varepsilon$, that is a positive parameter that we let go to $+ \infty$ as $\varepsilon$ goes to $ 0^+$. 
The choice of the kernel and the related parameters 
will be discussed at the end of this Section.

 The equation describes the evolution of the density of particles
, moving of linear motion and having random collisions, against a circular scatterer, that preserve the energy. 
	The  time between two consecutive jumps in the velocities is distributed with exponential law with mean value 
$(\lambda \eta_\varepsilon )^{-1}$.
Since both random collisions and hits against the elastic boundaries preserve the energy, the modulus of the velocity of a particle moving in $\Omega$ is constant, so we consider it to  be equal to one. 

On the elastic boundary  $\partial \Omega_E$ 
we impose reflective boundary condition and on 
the open boundary $\partial\Omega_L \cup\partial\Omega_R $ 
we set  Dirichlet condition:
\begin{equation}
\label{eq:bound_cond_h_eps}
\begin{cases}
g(x,v',t)=g(x,v,t) \qquad\qquad& x\in \partial\Omega_E,\, v\cdot n<0,\,t\geq 0
\\ 
g(x,v,t)= f_B(x,v)	& x\in \partial \Omega_L \cup \partial \Omega_R,\, v\cdot n >0,\, t\geq 0,
\end{cases}
\end{equation}
where $f_B$ is defined in (\ref{eq:def_f_B}) below and $v'$ is given by the elastic collision rule $v'=v- 2(n\cdot v ) n$. Here we denote by $n=n(x)$  the inward pointing normal on the boundary $\partial \Omega$ of the domain.
We consider as
initial datum the function 
$f_0(x,v)\in L^\infty(\Omega\times S^1)$ and we define $f_B$
(not depending on $t$) as
\begin{equation}\label{eq:def_f_B}
f_B(x,v):=\begin{cases}
\rho_L /{2\pi}\,\qquad\quad x \in \partial\Omega_L , \, v\cdot n > 0\\
\rho_R/{2\pi}\,  \qquad\quad x \in \partial\Omega_R,\, v\cdot n > 0 ,
\end{cases}
\end{equation}
where $ {1}/{2\pi}$ is the  density of the uniform distribution on $S^1$. 

We are interested in the study of the stationary problem associated to (\ref{eq:boltz_lin})-(\ref{eq:bound_cond_h_eps}):
\begin{equation}
\label{eq:stationary problem}
\begin{cases}
v\cdot \nabla_x g^S(x,v)= \eta_\varepsilon \cc{L} g^S\qquad\quad& x\in \Omega,\,v\in S^1\\
g^S(x,v')=g^S(x,v)   &  x\in \partial\Omega_E,\, v\cdot n<0
\\ g^S(x,v)= f_B(x,v) & x\in\partial\Omega_L\cup \partial\Omega_R, \, v\cdot n>0.
\end{cases}
\end{equation}

We want to investigate the behavior of the solution $g^S_\varepsilon$ 
of \eqref{eq:stationary problem}
and 
prove its convergence to the stationary solution of the diffusion problem in ${\Omega}$ with mixed boundary conditions given 
by
\begin{equation}
\label{eq:Lapl_mixed_omega1}
\begin{cases}
 \Delta \rho(x) = 0 \qquad&\, x\in\Omega
\\ \rho(x)=\rho_L\qquad \quad& x\in \partial \Omega_L  
\\ \rho(x)=\rho_R\qquad& x\in \partial\Omega_R
\\ \partial_n \rho (x) =0 \qquad & x \in \partial {\Omega}_E.
\end{cases}
\end{equation}

\begin{theorem}\label{th:ex!g_stat}
If $\varepsilon>0$ is sufficiently small
there exists a unique stationary solution $g^S_\varepsilon \in L^\infty ({\Omega} \times S^1)$ of (\ref{eq:stationary problem}).
\end{theorem}
\begin{theorem}
\label{th:g_to_rho}
The stationary solution $g_\varepsilon^S$ of (\ref{eq:stationary problem})
verifies
\begin{equation}
g^S_\varepsilon \to \rho
\end{equation}
as $\varepsilon \to 0$, where $\rho(x)$ is the  solution to the problem \eqref{eq:Lapl_mixed_omega1}. 
The convergence is in $L^\infty(\Omega\times S^1)$.
\end{theorem}

The choice of the elastic collision kernel for the operator $\cc{L}$ defined in \eqref{eq:op_L} is due to the physical model we have in mind. 
We are considering a particle moving with initial velocity $v\in S^1$ and hitting an hard circular scatterer whose position is random.
 The random impact parameter $\delta$ chosen uniformly in $[-1,1]$ allows to individuate this collision.
In a similar way  we could let the particle move following the Lorentz process, that is moving freely  in a region where static small disks of radius $\varepsilon$ are 
distributed according to a Poisson probability measure and elastically colliding with those disks. 
In this case for suitable choices of the mean value of the Poisson distribution in terms of $\varepsilon$ and $\eta_\varepsilon$ it could be possible to prove that the diffusive limit for the linear Boltzmann equation and the Lorentz process in a (small disks) low density limit are asymptotically equivalent in the limit $\varepsilon \to 0$ (see \cite{BNPP} 
for the case of an 
infinite $2D$ slab with open boundary). 

\section{Numerical convergence of stationary linear Boltzmann}
\label{s:simul} 
\par\noindent
We investigate, here, 
the stationary solution of the linear Boltzmann equation from a numerical point of view.
Our algorithm directly simulates the motion of single particles following the linear Boltzmann equation.
In the simulations we exploit the interpretation of the linear Boltzmann equation as the equation describing a stochastic jump process in the velocities 
 and we directly simulate the motion of single particles.

We will show that the numerical stationary solution that we  construct is close to the solution of the associated Laplace problem \eqref{eq:Lapl_mixed_omega1} if the scale parameter $\varepsilon$ is small enough, that is to say if the average time $t_m$
between two consecutive hits is sufficiently small.
This time will be called in the sequel \emph{mean flight time}.

 We will construct the solution of the Laplace problem \eqref{eq:Lapl_mixed_omega1} in our geometry by using the \emph{COMSOL Multiphysics} software.

We proceed in the following way:
we consider particles entering in $\Omega$ from the reservoirs. A particle starts its trajectory from the left boundary $\partial{\Omega}_L$  or from the right boundary $\partial{\Omega}_R$, where the mass density is $\rho_L$ and $\rho_R$ respectively. 
Therefore the number of particles we let enter  from each side is chosen to be proportional to $\rho_L$ and $\rho_R$. 
In other words, we select the starting side of the particle with probability ${\rho_\rr{L}}/({\rho_\rr{L}+\rho_\rr{R}})$ (left side) and ${\rho_\rr{R}}/({\rho_\rr{L}+\rho_\rr{R}})$ (right side) respectively.
Then we draw uniformly the position $x$ in $\partial\Omega_\rr{L}$ or $\partial\Omega_\rr{R}$ and the velocity $v$ in $S^1$ with $v \cdot n(x)>0$, $n(x)$ inward-pointing normal.

Once the particle started, it moves with uniform linear motion until it hits a scatterers or the elastic boundary $\partial{\Omega}_E$. 
We pick $t_h$, the time until the hit with a scatterers, following the exponential law of mean $t_m$, with $t_m$ a fixed parameter. The particle travels with velocity $v$ for a time $t_h$. If during this time it hits the elastic boundary $\partial{\Omega}_E$, its velocity changes performing an elastic collision.
At time $t_h$ we simulate an hit with a scatterers by picking an impact parameter $\delta$ uniformly in $[-1,1]$ and changing the velocity from $v$ to $v'=v- 2 (v\cdot n) n $, where $\delta=\sin \alpha$ and $n=n(\delta)$  is the  outward pointing normal to the scatterers such that the angle of incidence between $n$ and $v$ is $\alpha$ (see Figure \ref{fig:bol_op}).

We proceed as before by letting the particle move 
 until it leaves ${\Omega}$ by reaching again the open boundary $\partial\Omega_L\cup\partial{\Omega}_R$. Then the particle exits from the system and we  are ready to simulate another particle.
We simulate a number $N$ of particles.

The random number generator we use in our simulation is the Mersenne Twister \cite{MN-Mer, MN-Test}. 

We want to construct the stationary solution of equation \eqref{eq:boltz_lin}-\eqref{eq:bound_cond_h_eps}.
Note that we can simulate particles one by one since in the considered model particles are not interacting.
Moreover, being the stationary state not dependent on the initial datum $f_0$, we consider in this algorithm only particles starting from the reservoirs.
We assume that in the stationary state the density of particles in a region is proportional to the total time spent by all particles in that region. Moreover, due to isotropy, there is no preferential direction for the velocity, so the stationary $g^S_\varepsilon$ is not dependent on $v$.

We divide the space ${\Omega}$ in equal small square cells.
 For every particle  we calculate the time  that it spends in every cell.
Then we calculate the total time that particles spend in each cell. 
For $t_m$ going to zero and $N$ very big, in every infinitesimal region of $\Omega$
 the stationary
density has to be proportional to the total time spent by particles in
that region.
Considering  sufficiently small cells, for $t_m$ small enough and a number of particles simulated $N$ big enough,
 the total time spent in the cell we calculate is proportional to our numerical stationary solution. 

We construct with our algorithm a grid of sojourn times in the cells. 
The last step we have to do is to normalize it. 
It is sufficient to multiply by a constant, obtained by imposing the correct value of $g_\varepsilon^S$ in a point (e.g., the boundary datum). 
We call our numerical solution $h_{t_m}(x)$.
So we fix a cell in contact with the reservoir where we calculated a sojourn time $t_c$ and consider the value $f_B$ of the stationary solution.
We choose as multiplication constant $c={ f_B(x) }/{ t_c }$. So  the simulated solution $h_{t_m}(x)$ is constructed by multiplying the sojourn time in each cell for this constant $c$. 

In the sequel we will show that for $t_m$ sufficiently small and $N$ big enough, the simulated solution $h_{t_m}(x)$ well approximates the solution $\rho(x)$ of the associated Laplace problem.

All the simulations we are going to discuss in this Section are performed with $N=5\cdot10^7$ particles.

Let us preliminary consider the case of $\Omega=(0,4)\times(0,1)$ in absence of obstacles. We fix mass densities at the reservoirs $\rho_L=1$ and $\rho_R=0.5$. In this first case there is no dependence on the vertical coordinate in the solution of the Laplace problem \eqref{eq:Lapl_mixed_omega1}. Indeed we know that the problem has analytic solution $\rho(x_1,x_2)=1-{x_1}/{8}$, where we are denoting by $(x_1,x_2)\in\bb{R}\times\bb{R}$ the spatial coordinates $x$ in $\Omega$.
We divide the domain in $200\times 50$ equal square cells and we consider simulations with different values of $t_m$, to understand which values of the mean time $t_m$  provide a good approximation of the solution we are looking for.

In  {Figures} \ref{f:st0_1} and \ref{f:st0_2} we show that the choice of $t_m$ of the order of $10^{-2}$ is suitable for our purpose. Indeed in {Figures} \ref{f:st0_1} we compare the simulated solution $h_{t_m}$ with the analytic solution for the values $t_m=2\cdot 10^{-1}$,  $t_m=10^{-1}$, $t_m=2\cdot10^{-2}$ . We see in a $3D$ plot and in a $2D$ plot, obtained from the previous one by averaging on the $x_2$ variable, that $h_{t_m}$ becomes closer to $\rho(x_1)$ when  $t_m$ decreases.
So in {Figures} \ref{f:st0_2} we fix the parameter $t_m=10^{-2}$ and we verify that $h_{t_m}$ is close to the function $\rho(x_1)$ showing the relative error  $|h_{t_m}-\rho|/{\rho}$.

\begin{figure}[!ht]
\begin{picture}(200,140)(0,0)
\put(-26,-30){
  \includegraphics[width=0.72\textwidth]{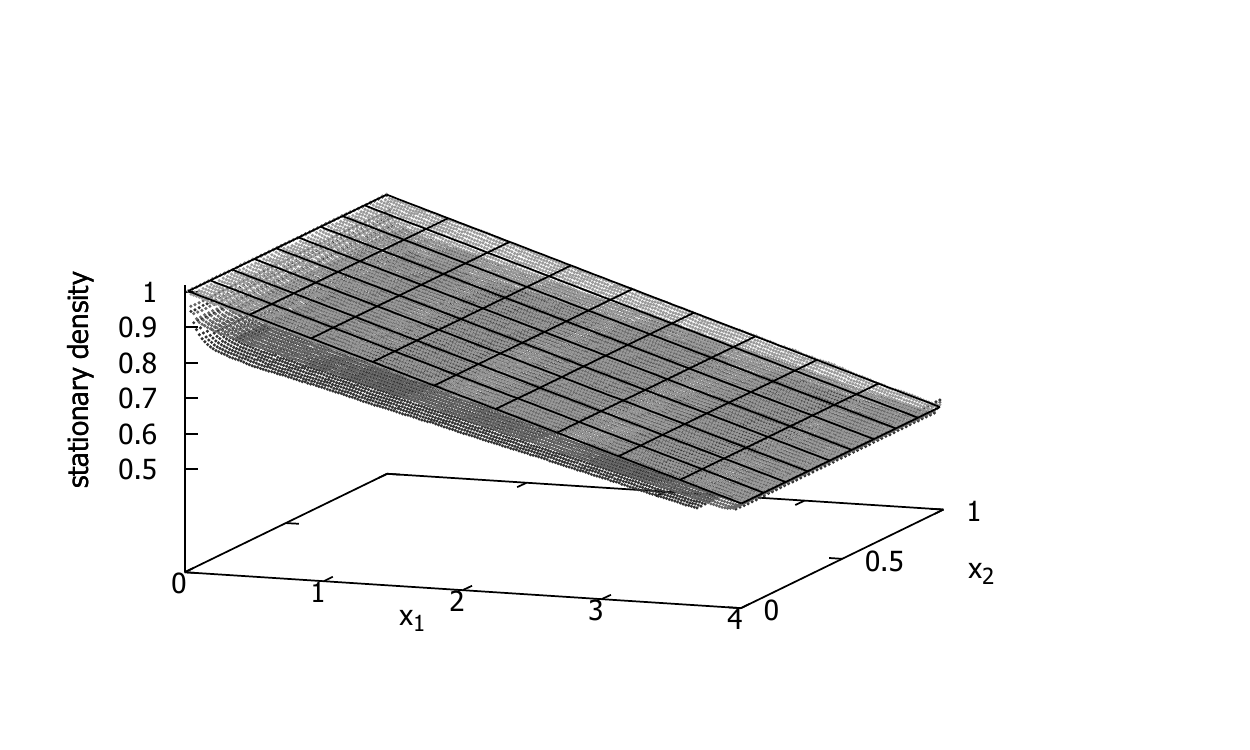}
}
\put(240,-17){
  \includegraphics[width=0.63\textwidth]{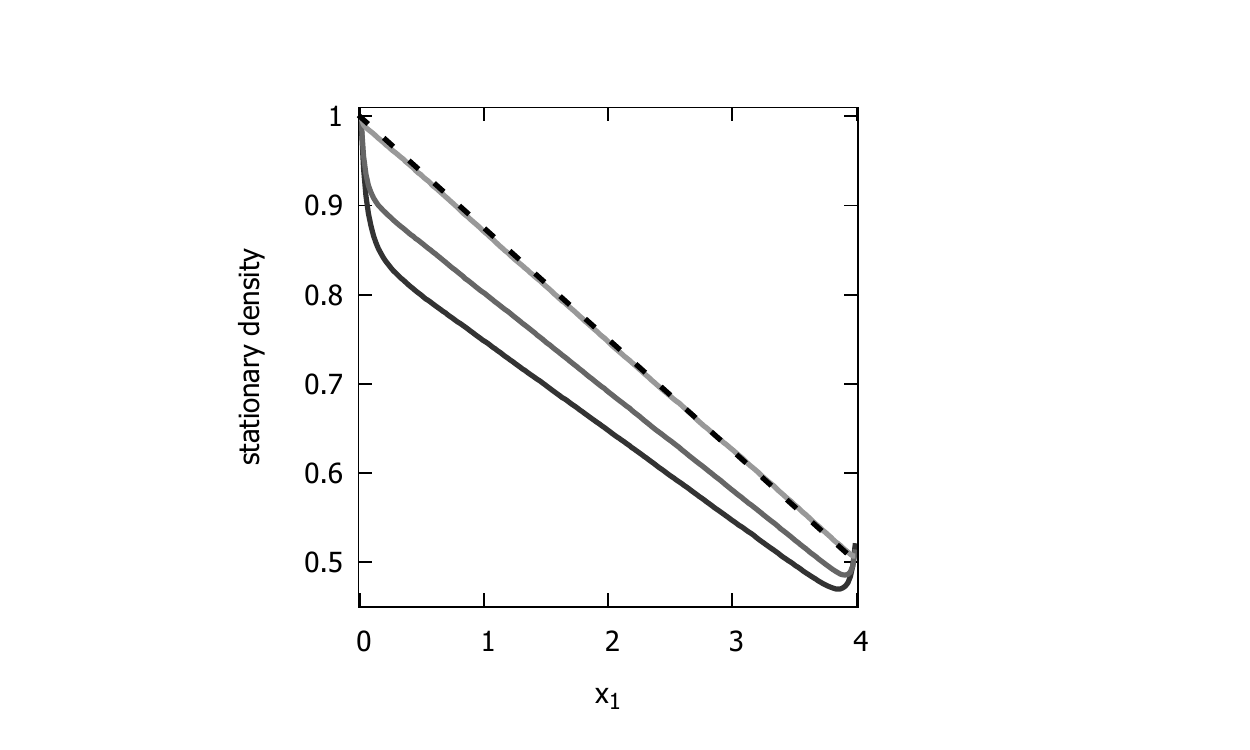}
}
\end{picture}
\caption{Plot of the simulated solutions $h_{t_m}$ in a $3D$ plot and in a $2D$ plot constructed by averaging on the $x_2$ variable: in dark gray $t_m=2\cdot 10^{-1}$, in gray $t_m=10^{-1}$, in light gray $t_m=2\cdot10^{-2}$. In black (grid and dashed line) the analytic solution $\rho$ of the associated Laplace
problem.}
\label{f:st0_1}
\end{figure}

\begin{figure}[!ht]
\begin{picture}(200,110)(0,0)
\centering
\put(100,-18){
  \includegraphics[width=0.63\textwidth]{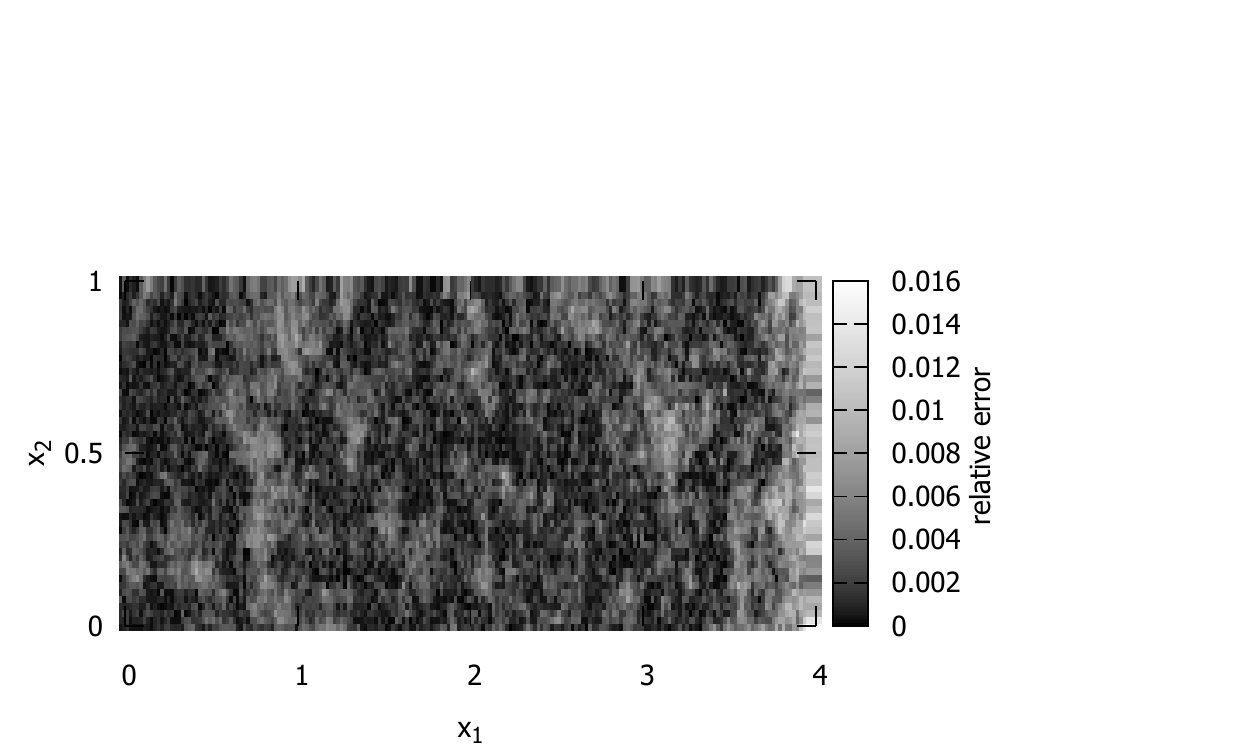}
}
\end{picture}
\caption{Simulation parameter $t_m=10^{-2}$: 
  relative error  ${|h_{t_m}-\rho|}/{\rho}$.}
\label{f:st0_2}
\end{figure}


We consider now the more interesting case with presence of obstacles in the strip. Our domain $\Omega$ is the strip $(0,4)\times(0,1)$ minus the obstacles. We fix again mass densities at the reservoirs $\rho_L=1$ and $\rho_R=0.5$. 
We fix again $t_m=10^{-2}$, since 
we have shown that in the empty case this choice for the exponential clock allows to construct a numerical solution $h_{t_m}$ that is close to  the analytical solution $\rho(x_1,x_2)$ of the associated
Laplace problem \eqref{eq:Lapl_mixed_omega1}.
We propose different situations for  the domain $\Omega$ and we show in {Figures} \ref{f:st1} - \ref{f:st_d} that in each case the simulation algorithm works correctly.
We compare our numerical solution with the solution $\rho(x_1,x_2)$ of the associated Laplace problem \eqref{eq:Lapl_mixed_omega1}. We show the plots of the $h_{t_m}$ and $\rho$ and the map of the relative error ${|h_{t_m}-\rho|}/{\rho}$  as in {Figure} \ref{f:st0_2} .

The first case we consider is the presence of a big squared obstacle with side $8\cdot 10^{-1}$, in different positions into the strip. In Figure \ref{f:st1} the results on two different positions are presented.

\begin{figure}[!htp]
\begin{picture}(200,270)(0,0)
\put(-26,115){
  \includegraphics[width=0.71\textwidth]{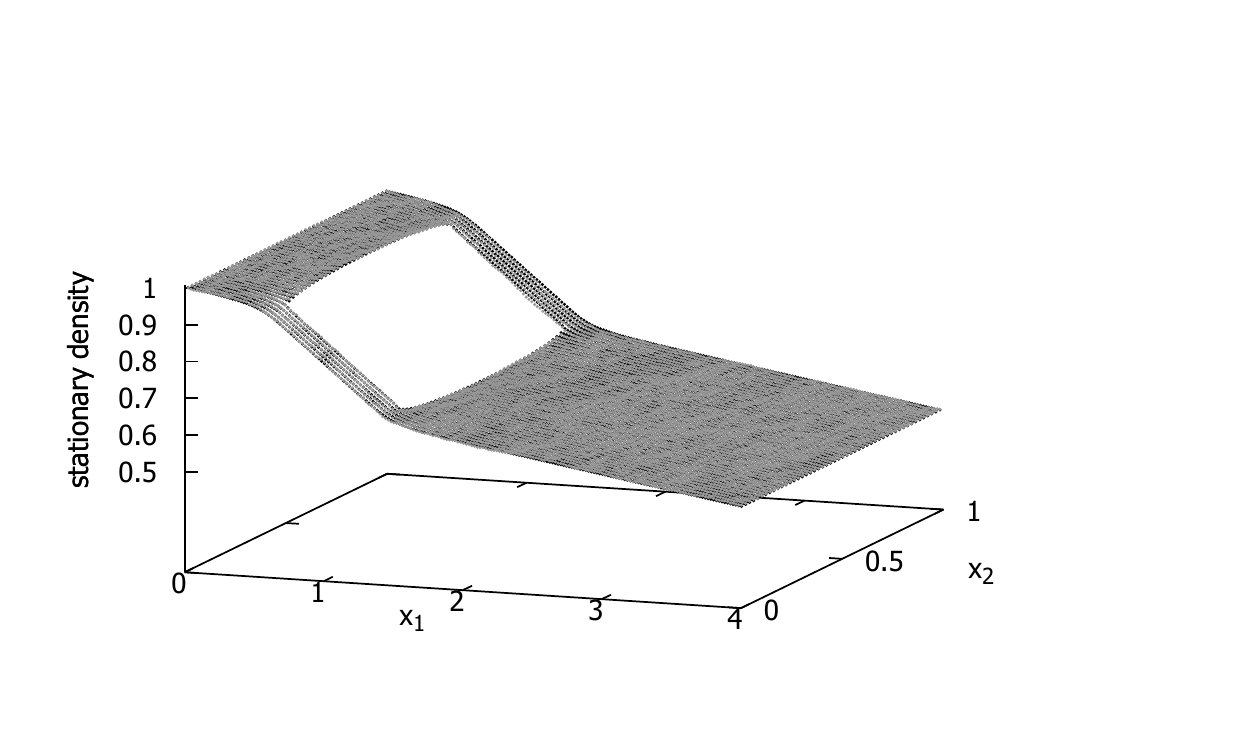}
}
\put(239,130){
  \includegraphics[width=0.63\textwidth]{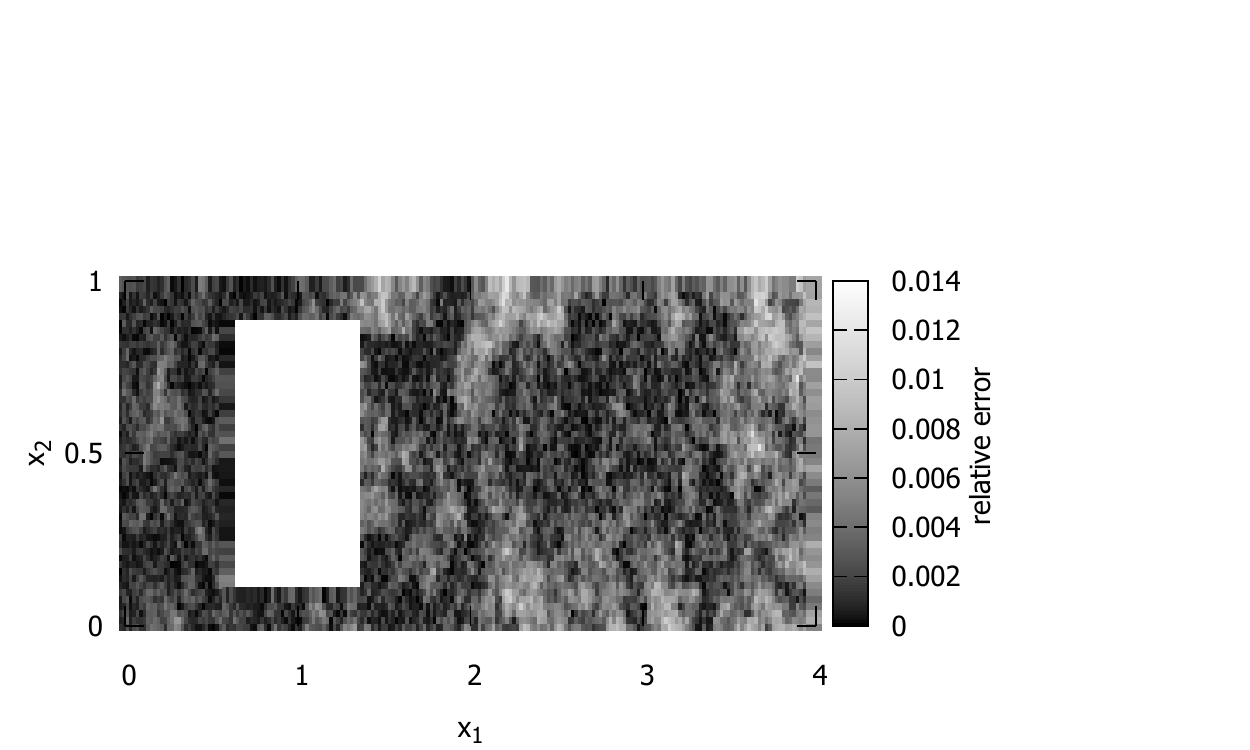}
}
\put(-26,-25){
  \includegraphics[width=0.71\textwidth]{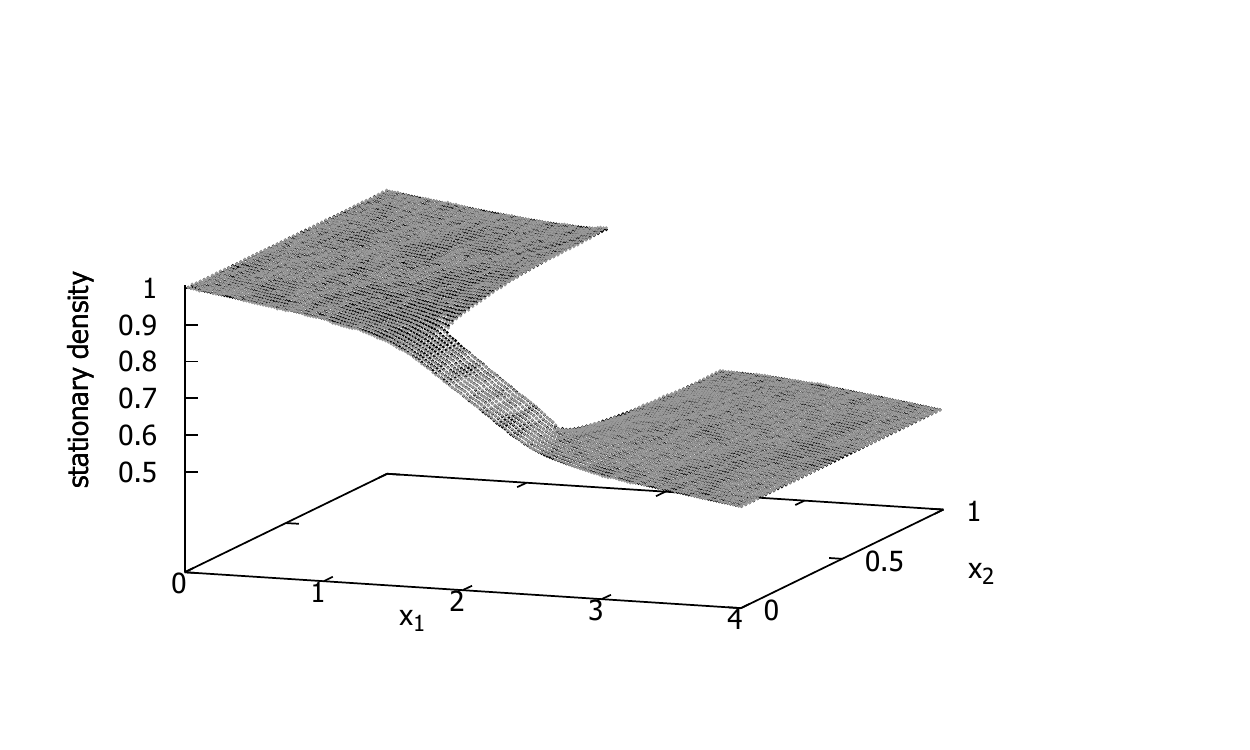}
}
\put(239,-8){
  \includegraphics[width=0.63\textwidth]{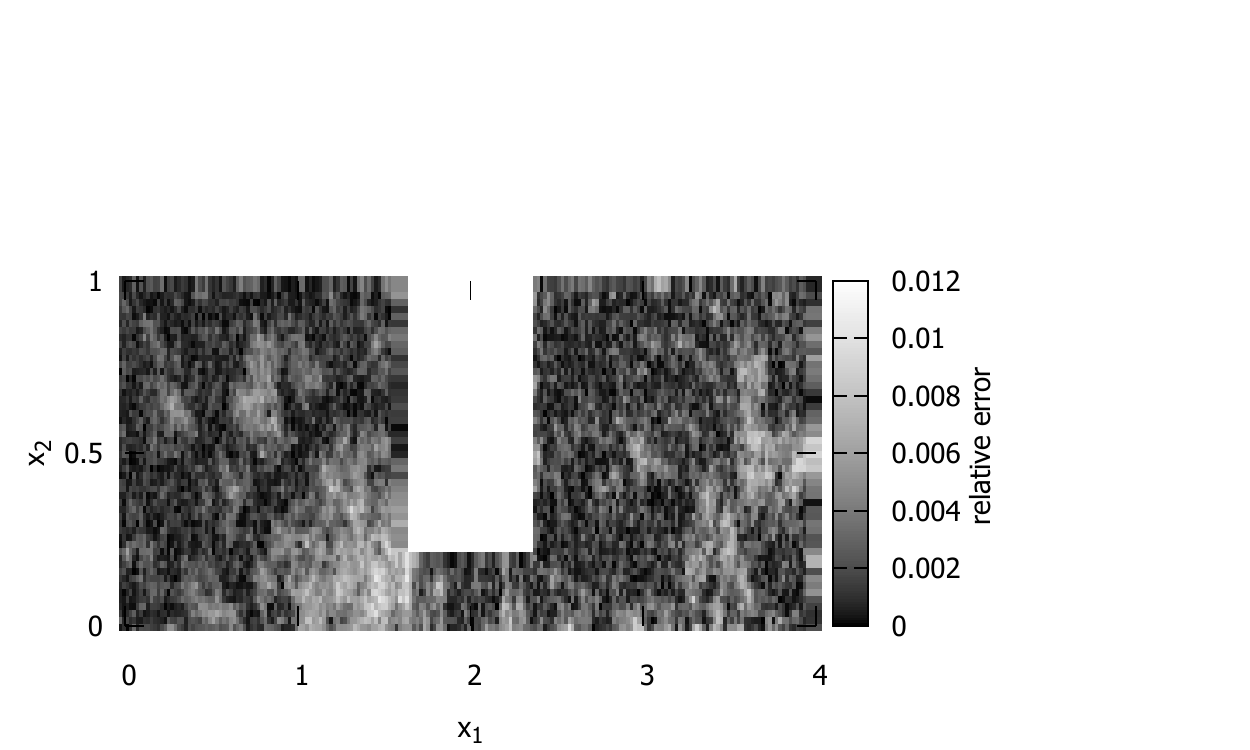}
}
\end{picture}
\caption{Simulation parameter $t_m=10^{-2}$: on the left in gray the
numerical solution $h_{t_m}$ and in black  the solution $\rho$ of the associated Laplace problem; on the right the relative error  ${|h_{t_m}-\rho|}/{\rho}$. Into the strip there is a square obstacle with side $8\cdot10^{-1}$.}
\label{f:st1}
\end{figure}

Another interesting case is the presence of a very thin and tall obstacle 
placed vertically inside the strip. We show it in Figure \ref{f:st5b}, by picking a thin obstacle of height of $8\cdot10^{-1}$.

\begin{figure}[htp!]
\begin{picture}(200,120)(0,0)
\put(-26,-33){
  \includegraphics[width=0.71\textwidth]{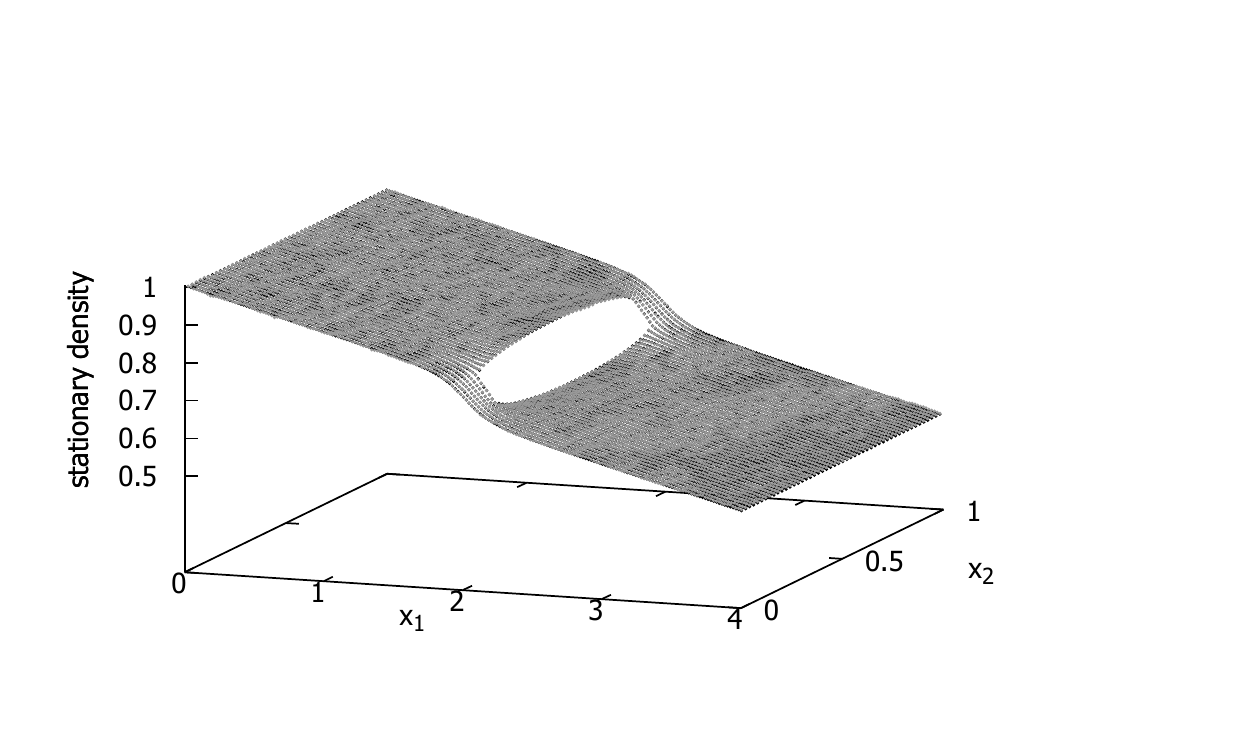}
}
\put(239,-15){
  \includegraphics[width=0.63\textwidth]{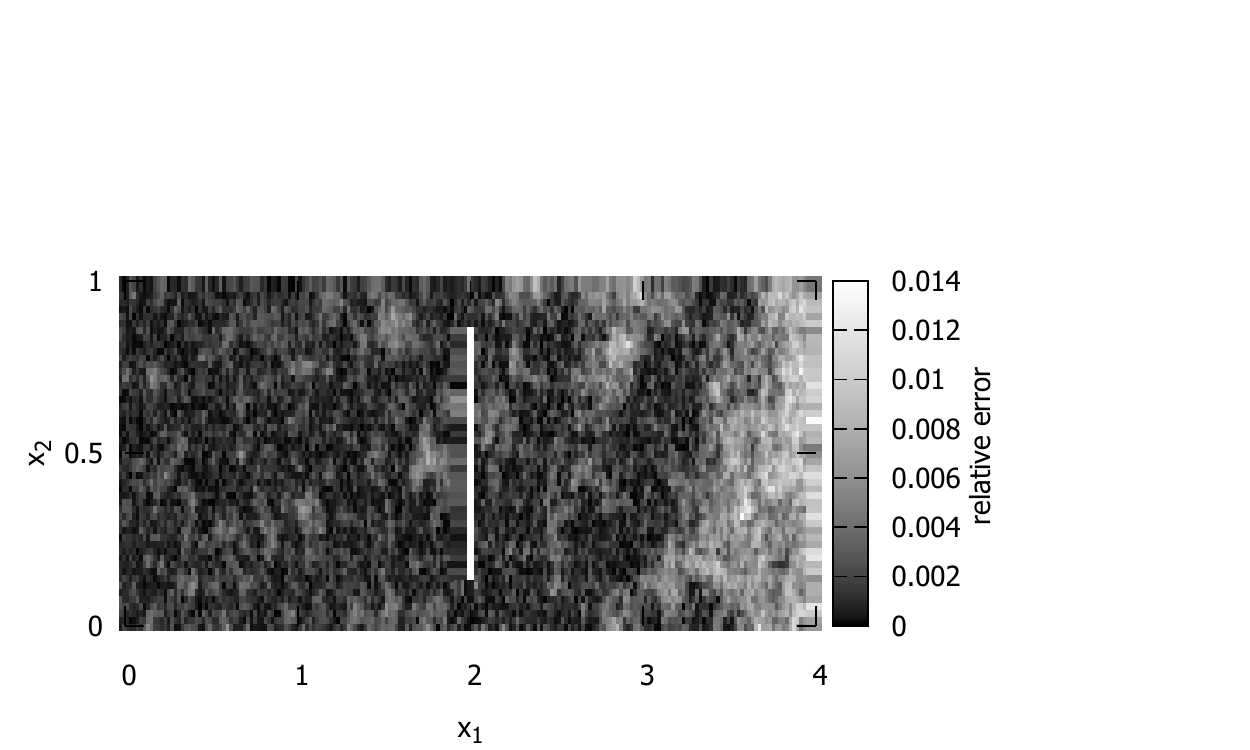}
}
\end{picture}
\caption{Simulation parameter $t_m=10^{-2}$: on the left in gray the
numerical solution $h_{t_m}$ and in black  the solution $\rho$ of the associated Laplace problem; on the right the relative error  ${|h_{t_m}-\rho|}/{\rho}$. In the strip is placed a very thin obstacle with height of $0.8$.}
\label{f:st5b}
\end{figure}

The last case we want to present is the presence in the strip of two obstacles. %
In Figure \ref{f:st_d} we consider two different situations: in the first we set in the strip two squared obstacles with sides $6\cdot 10^{-1}$ long, in the second we place into the strip two rectangular obstacles of sides $4\cdot 10^{-1}$ and $7 \cdot 10 ^{-1}$.

\begin{figure}[ht!]
\begin{picture}(200,270)(0,0)
\put(-26,115){
  \includegraphics[width=0.71\textwidth]{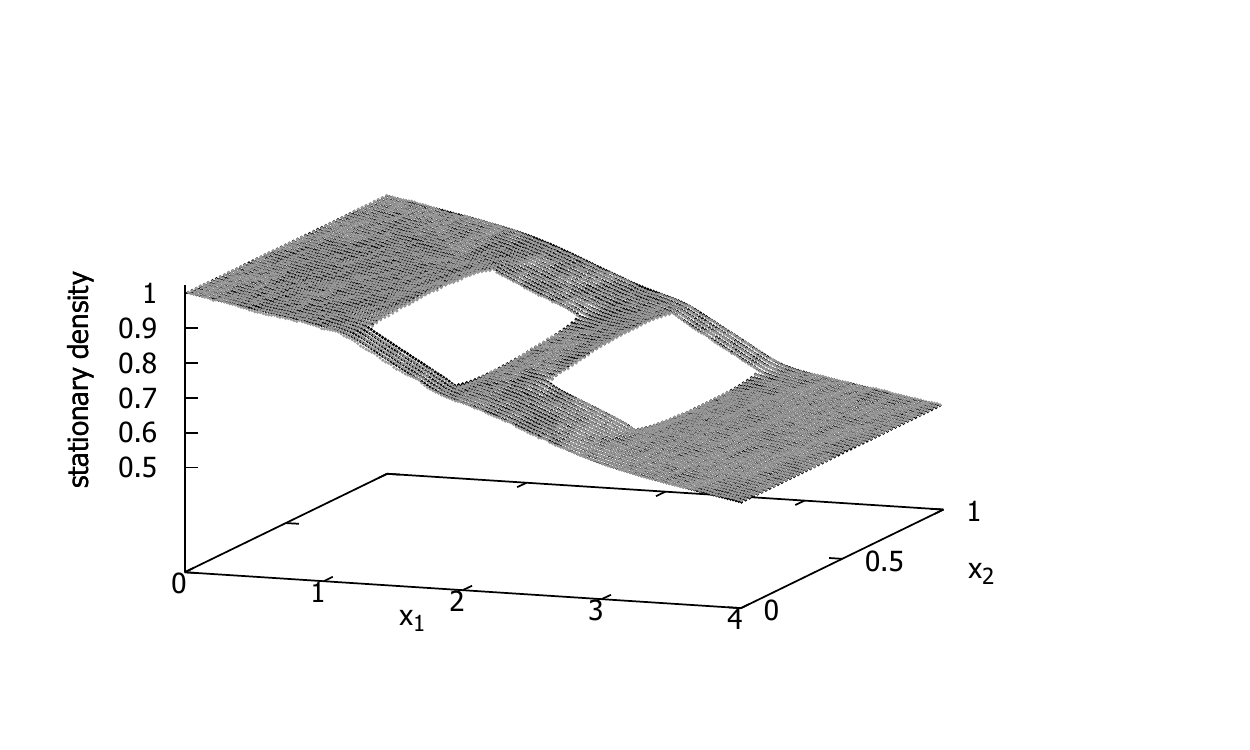}
}
\put(239,130){
  \includegraphics[width=0.63\textwidth]{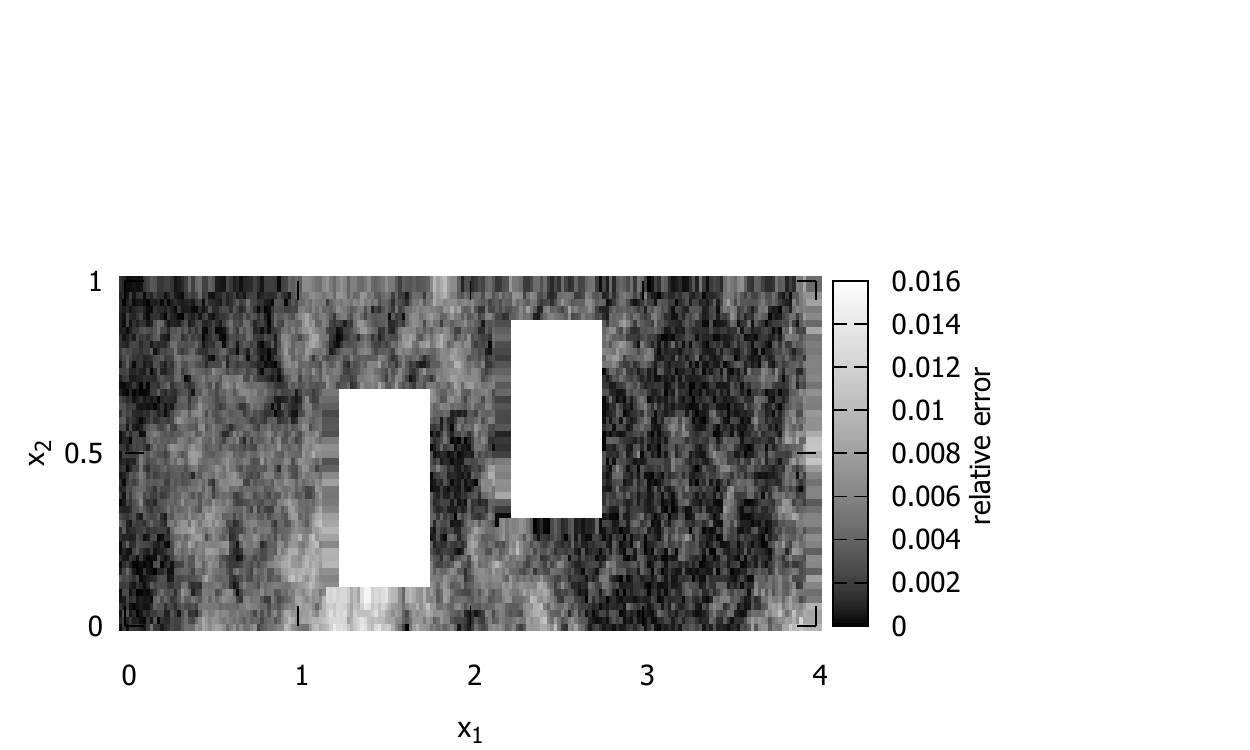}
}
\put(-26,-25){
  \includegraphics[width=0.71\textwidth]{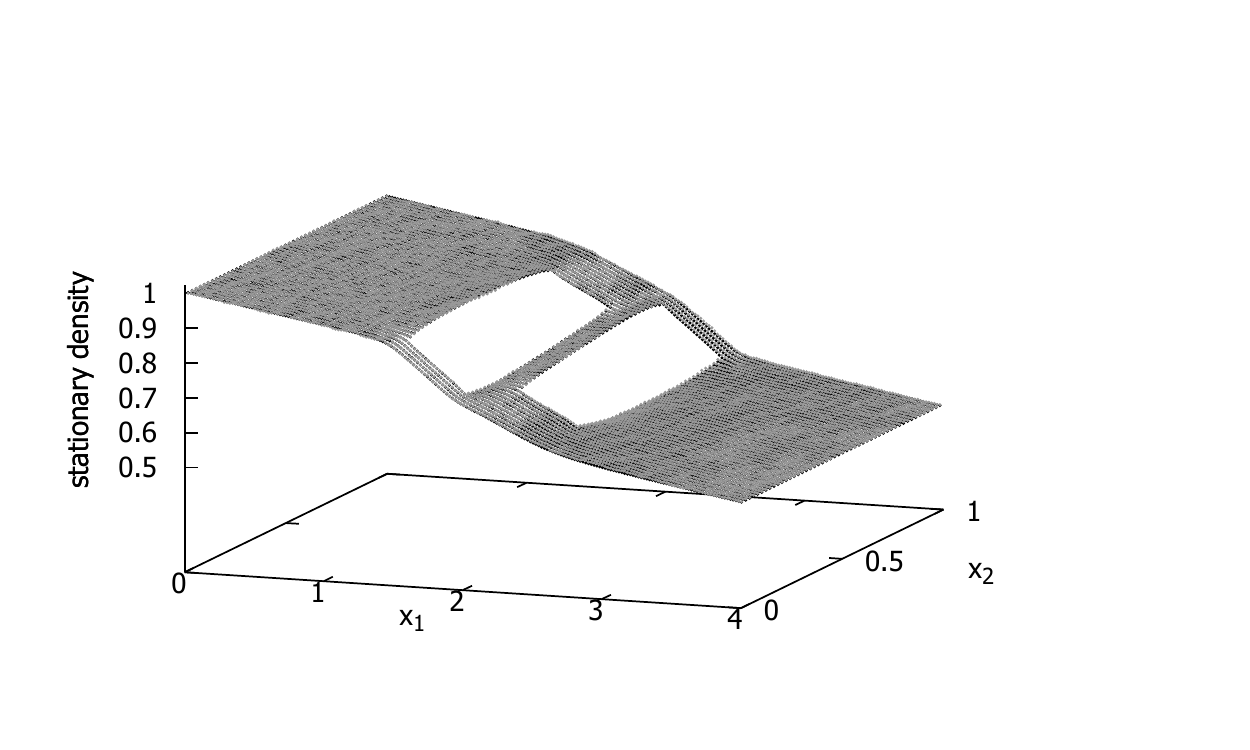}
}
\put(239,-8){
  \includegraphics[width=0.63\textwidth]{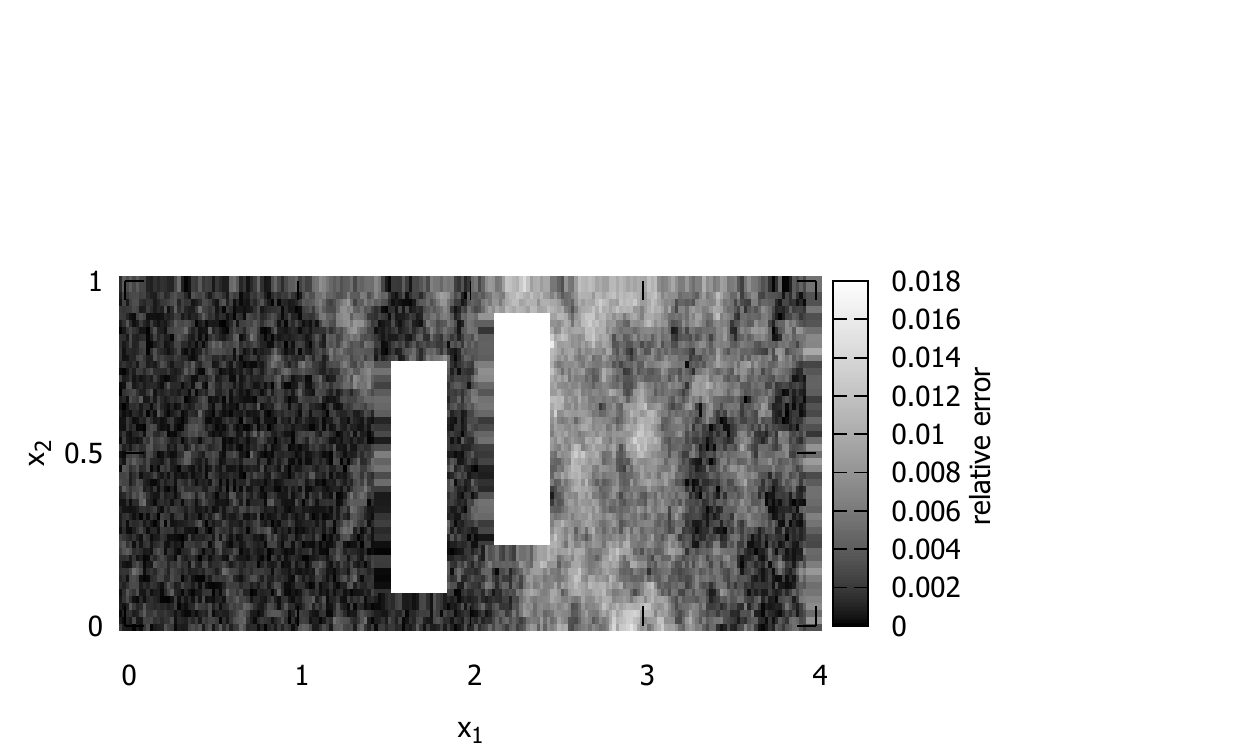}
}
\end{picture}
\caption{Simulation parameter $t_m=10^{-2}$: on the left in gray the
numerical solution $h_{t_m}$ and in black  the solution $\rho$ of the associated Laplace problem; on the right the relative error  ${|h_{t_m}-\rho|}/{\rho}$.
In the first line we show the case of two squared obstacles with side $6\cdot 10^{-1}$, in the second one a couple of rectangular obstacles, taller and thinner than the squares.}
\label{f:st_d}
\end{figure}

Note  that due to the presence of obstacles the solutions are not independent of the vertical coordinate $x_2$ anymore as it was in the empty strip case. However, we can notice that before and beyond the obstacles in the $x_1$ direction  the stationary states are closer to a flat state than in the empty case, with a steeper slope in the tight channels at sides of the obstacles. 
The total stationary mass flux through any vertical line $\{x_1\}\times(0,1) \cap \Omega$ does not depend on $x_1$. 
Indeed, this should follow from the Fick's law, that we expect to be 
valid also in presence of obstacles (in absence of obstacle, being the limiting problem one--dimensional, the Fick's law holds as shown in \cite{BNPP}), together with the 
divergence theorem and the fact that 
the boundary conditions are homogeneous on $\partial \Omega_E$.
The Fick's law would tell us that, in presence of obstacles, 
the total flux on the lines $\{x_1=a\}\cap\Omega$ is smaller than in the 
empty case, as it is possible to see focusing on the vertical lines before 
the obstacles. In this sense, and opposite to what happen in the case of the 
study of the residence time, we find on the flux the intuitive result we 
expected.

\section{Residence time}
\label{s:residence} 
\par\noindent
We consider the domain $(0,L_1)\times(0,L_2)$ 
with the same boundary conditions as 
in Section~\ref{s:modello}, namely, reflecting horizontal boundaries and 
open vertical boundaries. 
As before $\Omega$ denotes a subset of this domain obtained by placing 
large fixed reflecting obstacles.
Particles in $\Omega$ move according to the Markov process 
solving the linear Boltzmann equation and described in detail in 
Section~\ref{s:simul}.

In Section~\ref{s:simul} we investigated the stationary state of the system 
and we demonstrated that, provided the mean flight time 
$t_m$ is sufficiently small, the stationary state is very well 
approximated by the solution of the Laplace problem 
\eqref{eq:Lapl_mixed_omega1} even in presence of obstacles.
We have also noted that, due to the presence of obstacles, the 
total flux crossing the strip is smaller 
with respect to the one measured in absence of obstacles.  
This implies that if we consider a fixed number of particles 
entering the strip throught the left boundary, 
the number of them exiting through the right boundary 
decreses when an obstacle is inserted. 
In our simulations we remark that the ratio 
between the number of particles exiting through the left boundary 
in presence of an obstacle and in the empty strip case does not 
depend very much on the geometry of the obstacle and, in the worst 
case we considered, it is approximatively equal to $1/5$.
Detailed data for the different cases 
we studied are reported in the figure captions of this section. 

In this section, on the other hand,  
we focus on those particles that do the entire 
trip, that is to say they enter through the left boundary and
eventually exit the strip 
through the right one. Limiting our numerical computation to these 
particles, we measure the average time needed to cross the strip, 
also called the residence time and discuss its dependence on the 
size and on the position of a large fixed obstacle placed in the 
strip.  
The surprising result is that the residence time is not monotonic 
with respect to the obstacle parameters, such as position and 
size. 
More precisely, 
we show that obstacles can increase or decrease the residence time with 
respect to the empty strip case depending on their side lengths 
and on their position. Moreover, in some cases, 
by varying only one of these parameters
a transition from the increasing effect to the decreasing effect is 
observed.

In some cases we observe that the residence time measured in presence 
of an obstacle is smaller than the one measured for the empty strip. 
In other words, we find that the obstacle is able to select those 
particles that cross the strip in a smaller time. 
More precisely, particles that succeed to cross the strip do it faster 
than they would in absence of obstacles.

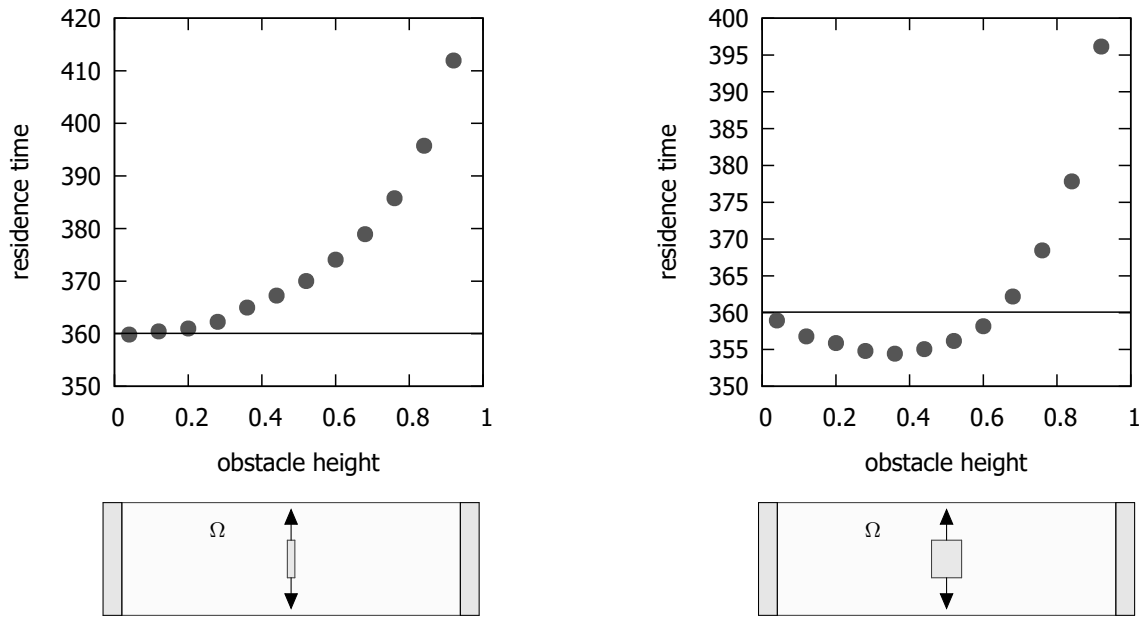
\begin{figure}[ht!]
\begin{picture}(200,237)(0,0)
\put(-55,60){
  \includegraphics[width=0.88\textwidth]{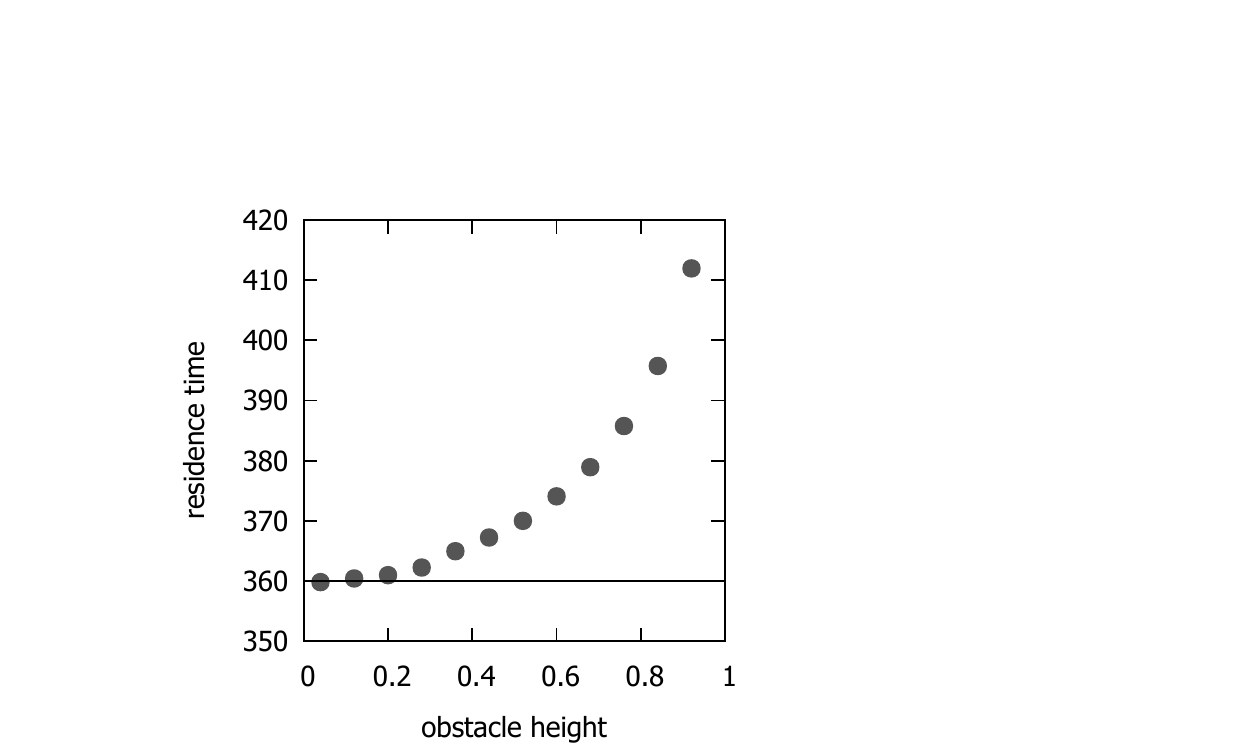}
}
\put(190,60){
  \includegraphics[width=0.88\textwidth]{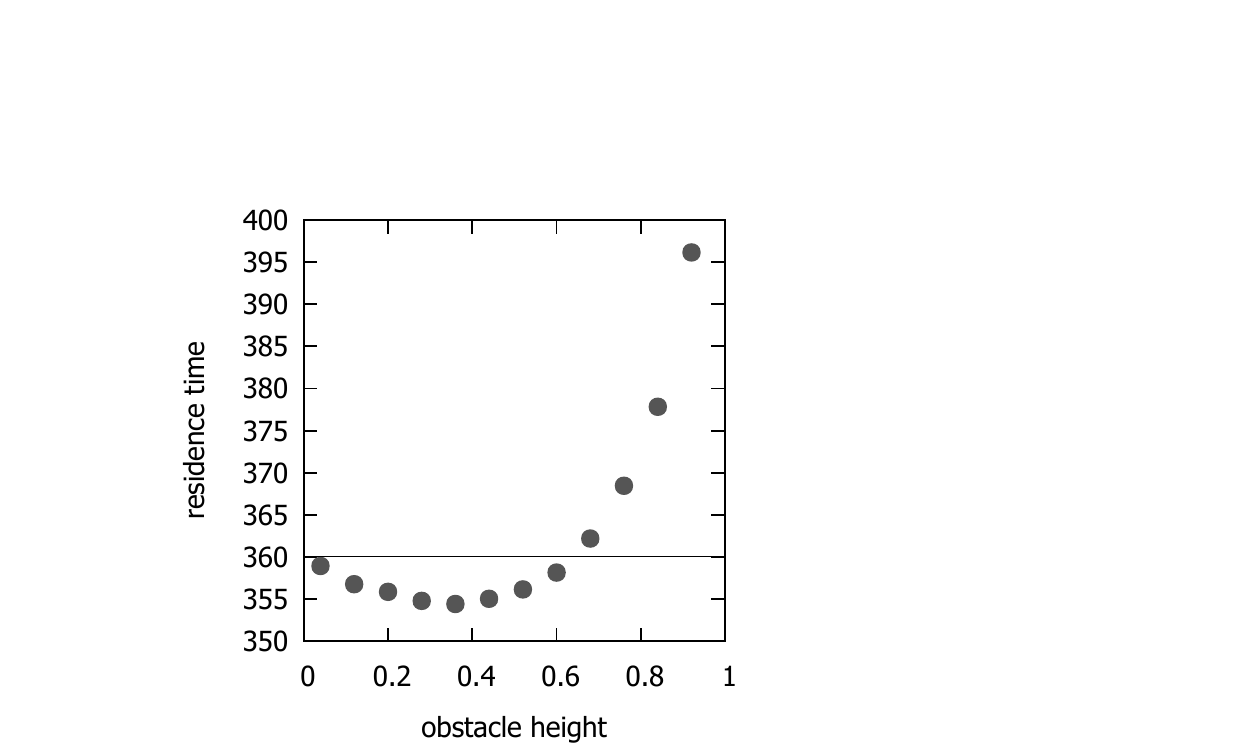}
}
\put(25,-5){
\definecolor{qqqqff}{rgb}{0.,0.,1.}
\definecolor{xdxdff}{rgb}{0.49019607843137253,0.49019607843137253,1.}
\definecolor{uququq}{rgb}{0.25098039215686274,0.25098039215686274,0.25098039215686274}
\definecolor{eqeqeq}{rgb}{0.8784313725490196,0.8784313725490196,0.8784313725490196}
\begin{tikzpicture}[line cap=round,line join=round,>=triangle 45,x=0.5cm,y=0.5cm]
\clip(-0.53,-1.) rectangle (9.53,4.);
\fill[color=eqeqeq,fill=eqeqeq,fill opacity=0.10000000149011612] (0.,0.) -- (0.,3.) -- (9.,3.) -- (9.,0.) -- cycle;
\fill[fill=black,fill opacity=0.10000000149011612] (9.,3.) -- (9.5,3.) -- (9.5,0.) -- (9.,0.) -- cycle;
\fill[fill=black,fill opacity=0.10000000149011612] (0.,3.) -- (-0.5,3.) -- (-0.5,0.) -- (0.,0.) -- cycle;
\fill[color=uququq,fill=uququq,fill opacity=0.10000000149011612] (4.4,2.) -- (4.6,2.) -- (4.6,1.) -- (4.4,1.) -- cycle;
\draw [color=uququq] (0.,0.)-- (0.,3.);
\draw [color=uququq] (0.,3.)-- (9.,3.);
\draw [color=uququq] (9.,3.)-- (9.,0.);
\draw [color=uququq] (9.,0.)-- (0.,0.);
\draw (9.,3.)-- (9.5,3.);
\draw (9.5,3.)-- (9.5,0.);
\draw (9.5,0.)-- (9.,0.);
\draw (9.,0.)-- (9.,3.);
\draw (0.,3.)-- (-0.5,3.);
\draw (-0.5,3.)-- (-0.5,0.);
\draw (-0.5,0.)-- (0.,0.);
\draw (0.,0.)-- (0.,3.);
\draw [color=uququq] (4.4,2.)-- (4.6,2.);
\draw [color=uququq] (4.6,2.)-- (4.6,1.);
\draw [color=uququq] (4.6,1.)-- (4.4,1.);
\draw [color=uququq] (4.4,1.)-- (4.4,2.);
\draw [->] (4.5,2.) -- (4.5,2.85);
\draw [->] (4.5,1.) -- (4.5,0.15);
\begin{scriptsize}
\draw[color=black] (2.546686424710366,2.288597693839765) node {$\Omega$};
\end{scriptsize}
\end{tikzpicture}
}
\put(273,-5){
\definecolor{qqqqff}{rgb}{0.,0.,1.}
\definecolor{xdxdff}{rgb}{0.49019607843137253,0.49019607843137253,1.}
\definecolor{uququq}{rgb}{0.25098039215686274,0.25098039215686274,0.25098039215686274}
\definecolor{eqeqeq}{rgb}{0.8784313725490196,0.8784313725490196,0.8784313725490196}
\begin{tikzpicture}[line cap=round,line join=round,>=triangle 45,x=0.5cm,y=0.5cm]
\clip(-0.5257992564018515,-1.0002443737933087) rectangle (9.513789868443846,3.990645421475449);
\fill[color=eqeqeq,fill=eqeqeq,fill opacity=0.10000000149011612] (0.,0.) -- (0.,3.) -- (9.,3.) -- (9.,0.) -- cycle;
\fill[fill=black,fill opacity=0.10000000149011612] (9.,3.) -- (9.5,3.) -- (9.5,0.) -- (9.,0.) -- cycle;
\fill[fill=black,fill opacity=0.10000000149011612] (0.,3.) -- (-0.5,3.) -- (-0.5,0.) -- (0.,0.) -- cycle;
\fill[color=uququq,fill=uququq,fill opacity=0.10000000149011612] (4.1,2.) -- (4.9,2.) -- (4.9,1.) -- (4.1,1.) -- cycle;
\draw [color=uququq] (0.,0.)-- (0.,3.);
\draw [color=uququq] (0.,3.)-- (9.,3.);
\draw [color=uququq] (9.,3.)-- (9.,0.);
\draw [color=uququq] (9.,0.)-- (0.,0.);
\draw (9.,3.)-- (9.5,3.);
\draw (9.5,3.)-- (9.5,0.);
\draw (9.5,0.)-- (9.,0.);
\draw (9.,0.)-- (9.,3.);
\draw (0.,3.)-- (-0.5,3.);
\draw (-0.5,3.)-- (-0.5,0.);
\draw (-0.5,0.)-- (0.,0.);
\draw (0.,0.)-- (0.,3.);
\draw [color=uququq] (4.1,2.)-- (4.9,2.);
\draw [color=uququq] (4.9,2.)-- (4.9,1.);
\draw [color=uququq] (4.9,1.)-- (4.1,1.);
\draw [color=uququq] (4.1,1.)-- (4.1,2.);
\draw [->] (4.5,2.) -- (4.5,2.85);
\draw [->] (4.5,1.) -- (4.5,0.15);
\begin{scriptsize}
\draw[color=black] (2.546686424710366,2.288597693839765) node {$\Omega$};
\end{scriptsize}
\end{tikzpicture}
}
\end{picture}
\caption{Residence time 
vs.\ height of a centered rectangular obstacle with 
fixed width $4\cdot 10^{-2}$ (on the left) and 
$4 \cdot 10^{-1}$ (on the right). 
Simulation parameters: 
$L_1=4$, 
$L_2=1$,
$t_m=2\cdot 10^{-2}$,
total number of inserted particles $10^8$,
the total number of particles exiting through the right boundary varies 
from $5.3\cdot 10^{5}$ to $3.6\cdot 10^{5}$ (on the left) 
and
from $5.3\cdot 10^{5}$ to $2.1\cdot 10^{5}$  (on the right)
depending on the obstacle height.
The solid lines represent the value of the residence time measured 
for the empty strip (no obstacle).}
\label{f:ct1}
\end{figure}

We now discuss the different cases we considered. All numerical details 
are in the figure captions. 
The statistical error is not represented in the pictures 
since it is negligible and it could not
be appreciated in the graphs.
In each figure we draw a graph reporting the numerical data and 
a schematic picture illustrating the performed experiment. 
We first describe our result and at the end of this section we propose 
a possible interpretation.

In Figure~\ref{f:ct1} we report the residence time as a function of the 
obstacle height. The obstacle is placed at the center of the strip
and its width is very small on the left and larger on the right. 
We notice that in the case of a thin barrier, the residence 
time increases with the height of the obstacle. 
On the other hand, for a wider obstacle, we do not find this 
a priori intuitive result, but we observe a not monotonic dependence 
of the residence time on the obstacle height. In particular, 
it is interesting to remark that if the obstacle height 
is chosen smaller that $0.65$ the residence time 
is smaller than the one measured for the empty strip. 
This effect is even stronger if the width 
of the obstacle is increased (Figure~\ref{f:ct2}).

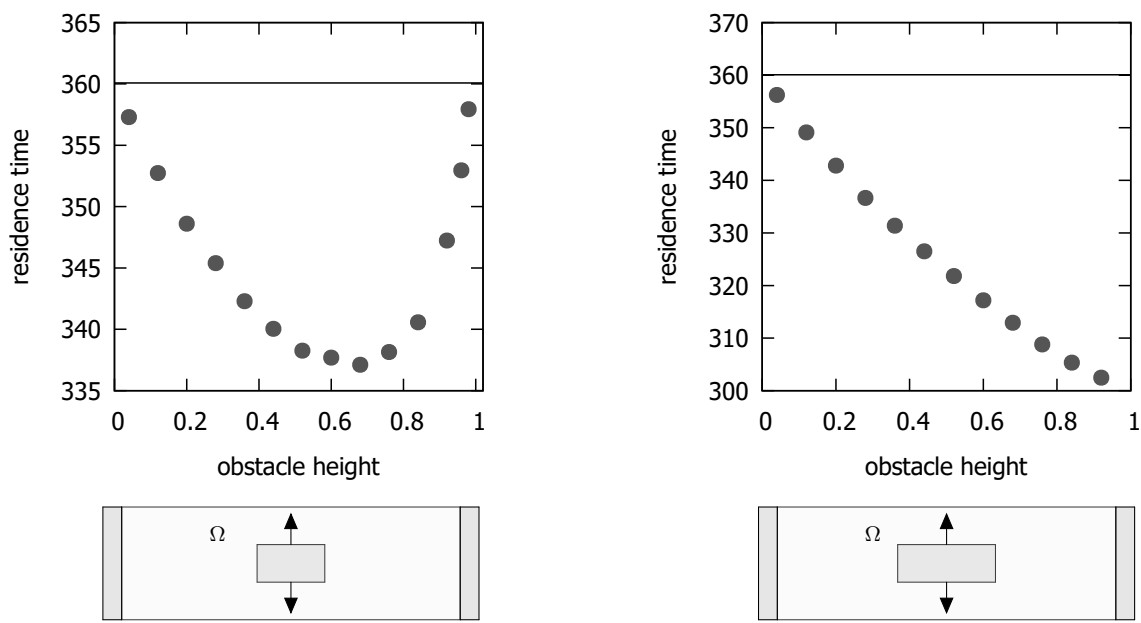
\begin{figure}[ht!]
\begin{picture}(200,237)(0,0)
\put(-55,60){
  \includegraphics[width=0.88\textwidth]{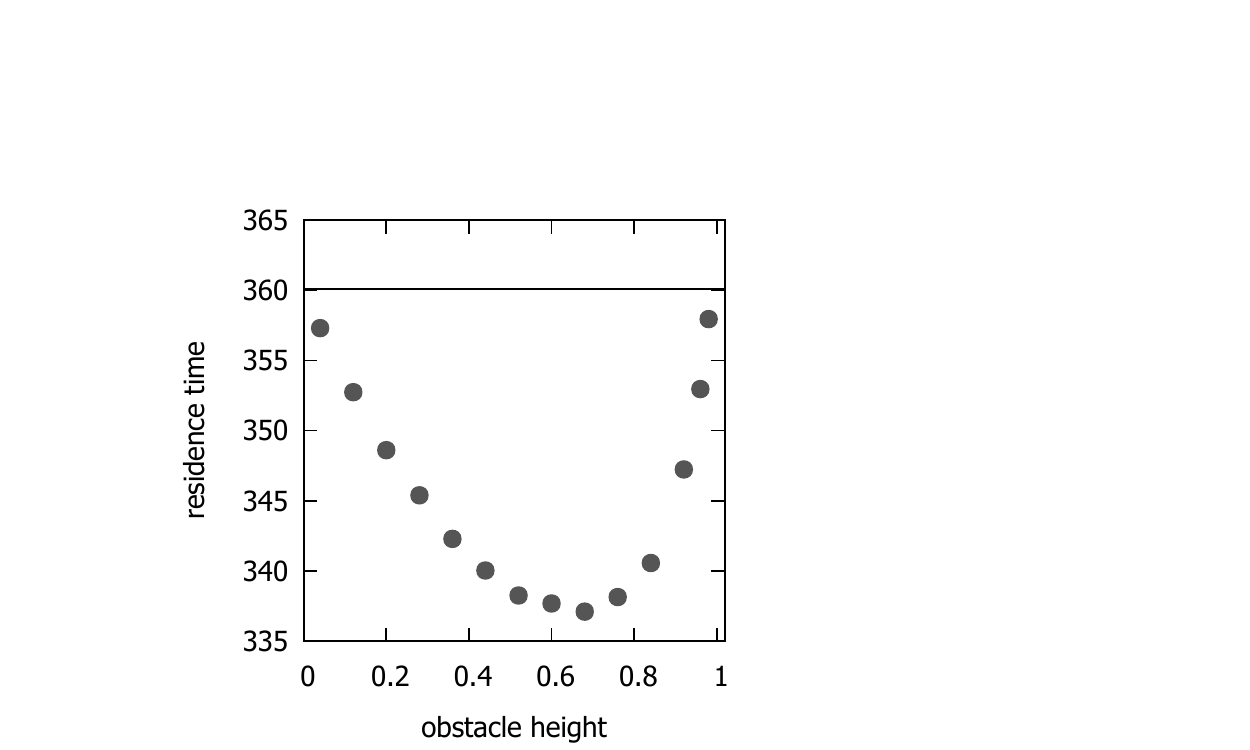}
}
\put(190,60){
  \includegraphics[width=0.88\textwidth]{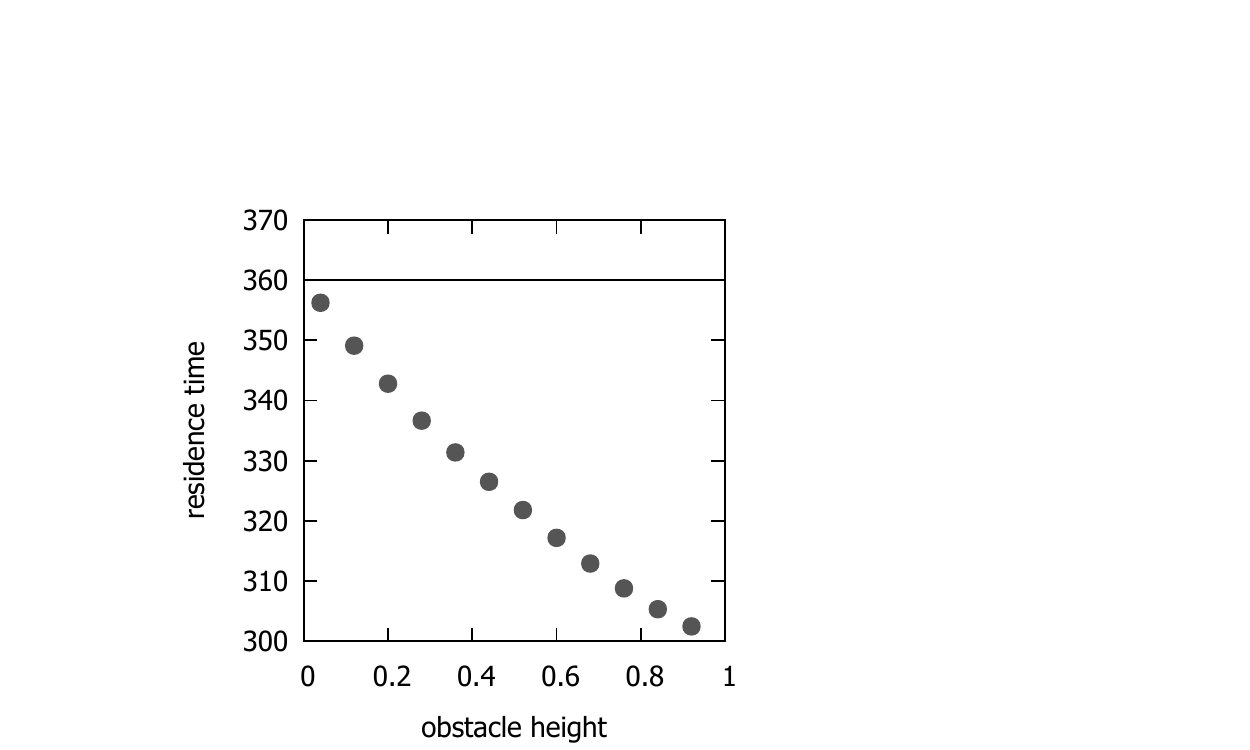}
}
\put(25,-5){
\definecolor{qqqqff}{rgb}{0.,0.,1.}
\definecolor{xdxdff}{rgb}{0.49019607843137253,0.49019607843137253,1.}
\definecolor{uququq}{rgb}{0.25098039215686274,0.25098039215686274,0.25098039215686274}
\definecolor{eqeqeq}{rgb}{0.8784313725490196,0.8784313725490196,0.8784313725490196}
\begin{tikzpicture}[line cap=round,line join=round,>=triangle 45,x=0.5cm,y=0.5cm]
\clip(-0.5257992564018515,-1.0002443737933087) rectangle (9.533059713213223,3.990645421475449);
\fill[color=eqeqeq,fill=eqeqeq,fill opacity=0.10000000149011612] (0.,0.) -- (0.,3.) -- (9.,3.) -- (9.,0.) -- cycle;
\fill[fill=black,fill opacity=0.10000000149011612] (9.,3.) -- (9.5,3.) -- (9.5,0.) -- (9.,0.) -- cycle;
\fill[fill=black,fill opacity=0.10000000149011612] (0.,3.) -- (-0.5,3.) -- (-0.5,0.) -- (0.,0.) -- cycle;
\fill[color=uququq,fill=uququq,fill opacity=0.10000000149011612] (3.6,2.) -- (5.4,2.) -- (5.4,1.) -- (3.6,1.) -- cycle;
\draw [color=uququq] (0.,0.)-- (0.,3.);
\draw [color=uququq] (0.,3.)-- (9.,3.);
\draw [color=uququq] (9.,3.)-- (9.,0.);
\draw [color=uququq] (9.,0.)-- (0.,0.);
\draw (9.,3.)-- (9.5,3.);
\draw (9.5,3.)-- (9.5,0.);
\draw (9.5,0.)-- (9.,0.);
\draw (9.,0.)-- (9.,3.);
\draw (0.,3.)-- (-0.5,3.);
\draw (-0.5,3.)-- (-0.5,0.);
\draw (-0.5,0.)-- (0.,0.);
\draw (0.,0.)-- (0.,3.);
\draw [color=uququq] (3.6,2.)-- (5.4,2.);
\draw [color=uququq] (5.4,2.)-- (5.4,1.);
\draw [color=uququq] (5.4,1.)-- (3.6,1.);
\draw [color=uququq] (3.6,1.)-- (3.6,2.);
\draw [->] (4.5,2.) -- (4.5,2.85);
\draw [->] (4.5,1.) -- (4.5,0.15);
\begin{scriptsize}
\draw[color=black] (2.546686424710366,2.288597693839765) node {$\Omega$};
\end{scriptsize}
\end{tikzpicture}

}
\put(273,-5){
\definecolor{qqqqff}{rgb}{0.,0.,1.}
\definecolor{xdxdff}{rgb}{0.49019607843137253,0.49019607843137253,1.}
\definecolor{uququq}{rgb}{0.25098039215686274,0.25098039215686274,0.25098039215686274}
\definecolor{eqeqeq}{rgb}{0.8784313725490196,0.8784313725490196,0.8784313725490196}
\begin{tikzpicture}[line cap=round,line join=round,>=triangle 45,x=0.5cm,y=0.5cm]
\clip(-0.5257992564018515,-1.0002443737933087) rectangle (9.533059713213223,3.990645421475449);
\fill[color=eqeqeq,fill=eqeqeq,fill opacity=0.10000000149011612] (0.,0.) -- (0.,3.) -- (9.,3.) -- (9.,0.) -- cycle;
\fill[fill=black,fill opacity=0.10000000149011612] (9.,3.) -- (9.5,3.) -- (9.5,0.) -- (9.,0.) -- cycle;
\fill[fill=black,fill opacity=0.10000000149011612] (0.,3.) -- (-0.5,3.) -- (-0.5,0.) -- (0.,0.) -- cycle;
\fill[color=uququq,fill=uququq,fill opacity=0.10000000149011612] (3.2,2.) -- (5.8,2.) -- (5.8,1.) -- (3.2,1.) -- cycle;
\draw [color=uququq] (0.,0.)-- (0.,3.);
\draw [color=uququq] (0.,3.)-- (9.,3.);
\draw [color=uququq] (9.,3.)-- (9.,0.);
\draw [color=uququq] (9.,0.)-- (0.,0.);
\draw (9.,3.)-- (9.5,3.);
\draw (9.5,3.)-- (9.5,0.);
\draw (9.5,0.)-- (9.,0.);
\draw (9.,0.)-- (9.,3.);
\draw (0.,3.)-- (-0.5,3.);
\draw (-0.5,3.)-- (-0.5,0.);
\draw (-0.5,0.)-- (0.,0.);
\draw (0.,0.)-- (0.,3.);
\draw [color=uququq] (3.2,2.)-- (5.8,2.);
\draw [color=uququq] (5.8,2.)-- (5.8,1.);
\draw [color=uququq] (5.8,1.)-- (3.2,1.);
\draw [color=uququq] (3.2,1.)-- (3.2,2.);
\draw [->] (4.5,2.) -- (4.5,2.85);
\draw [->] (4.5,1.) -- (4.5,0.15);
\begin{scriptsize}
\draw[color=black] (2.546686424710366,2.288597693839765) node {$\Omega$};
\end{scriptsize}
\end{tikzpicture}

}
\end{picture}
\caption{Residence time
vs.\ height of a centered rectangular obstacle with 
fixed width $8\cdot 10^{-1}$ (on the left) and 
$12 \cdot 10^{-1}$ (on the right). 
Simulation parameters: 
$L_1=4$, 
$L_2=1$,
$t_m=2\cdot 10^{-2}$,
total number of inserted particles $10^8$,
the total number of particles exiting through the right boundary varies 
from $5.2\cdot 10^{5}$ to $1.4\cdot 10^{5}$  (on the left)
and
from $5.2 \cdot 10^{5}$ to $1.1 \cdot 10^{5}$ (on the right)
depending on the obstacle height.
The solid lines represent the value of the residence time measured 
for the empty strip (no obstacle).
}
\label{f:ct2}
\end{figure}

In the left panel of Figure~\ref{f:ct3} we report the residence time as a 
function of the obstacle width. 
The obstacle is placed at the center of the strip
and its height is fixed to $0.8$. 
When the barrier is thin the residence time is larger than 
the one measured in the empty strip case. But, when the width 
is increased, the residence time decreases and at about $0.7$ it 
becomes smaller than the empty case value. 
The minimum is reached at about $2.3$, then the residence time starts 
to increase and when the width of the obstacles equals that of the strip 
the residence time becomes equal to the empty strip value. 
This last fact is rather obvious, indeed, in this case the strip reduces 
to two independent channels having the same width of the original strip. 

In the right panel of Figure~\ref{f:ct3} a centered square obstacle is 
considered. We note that the residence time happens to be a 
monotonic decreasing function of the obstacle side length. 

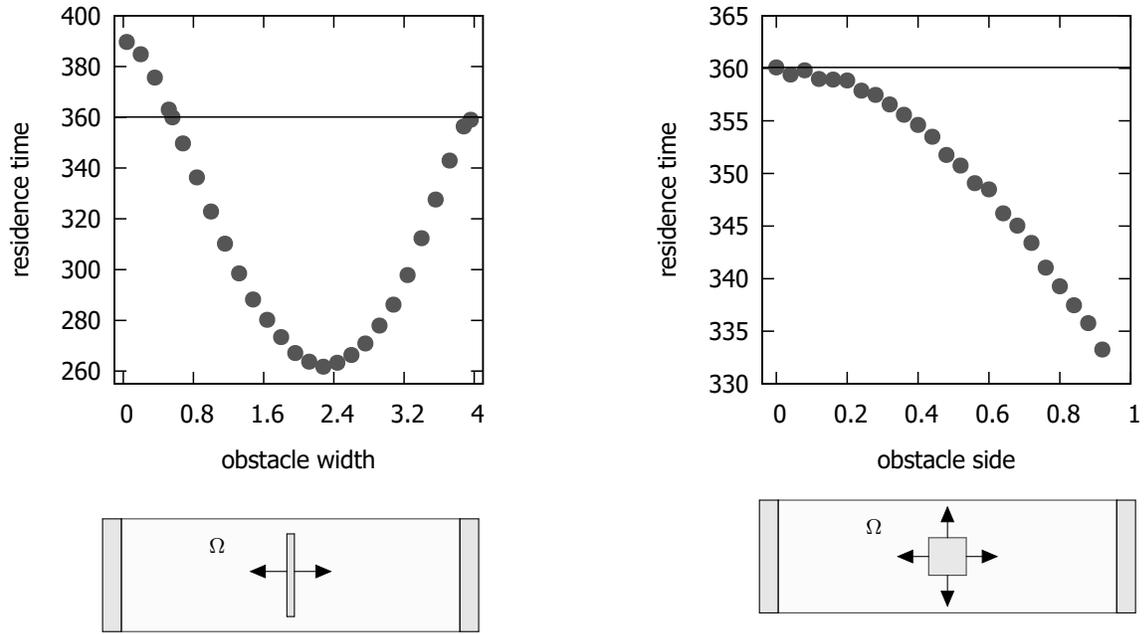
\begin{figure}[t!]
\begin{picture}(200,237)(0,0)
\put(-55,60){
  \includegraphics[width=0.88\textwidth]{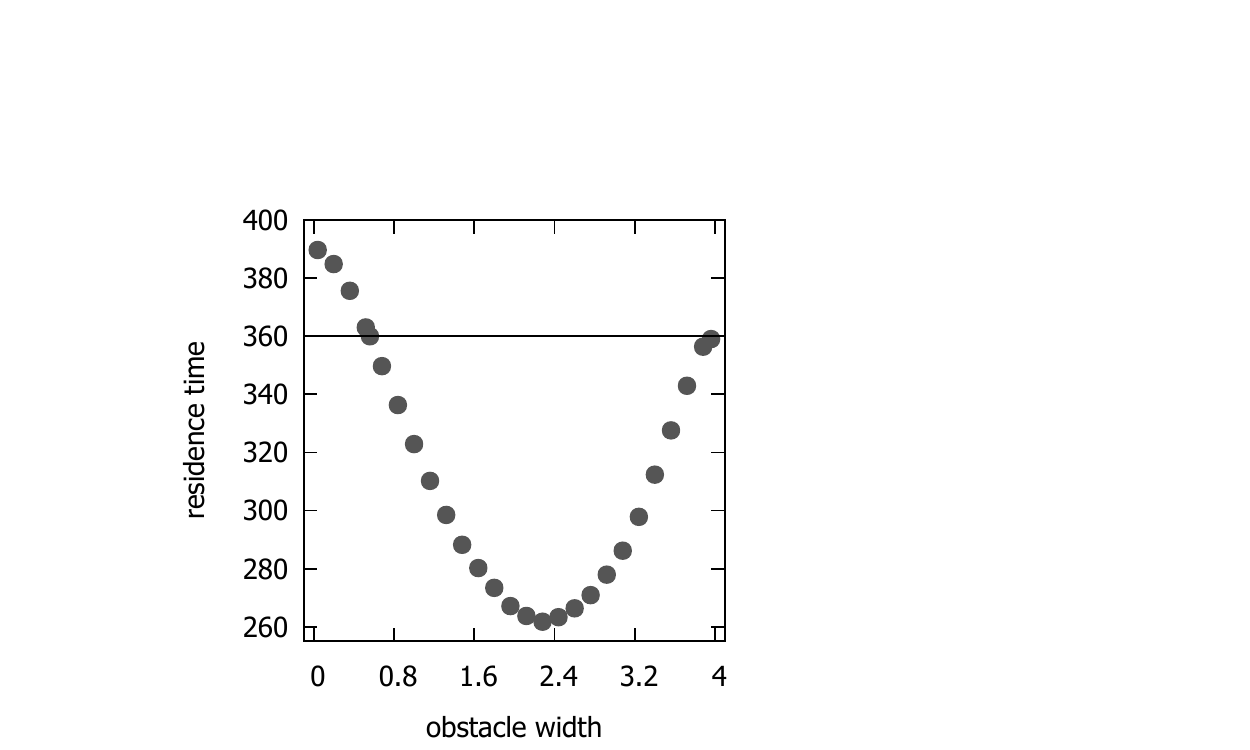}
}
\put(190,60){
  \includegraphics[width=0.88\textwidth]{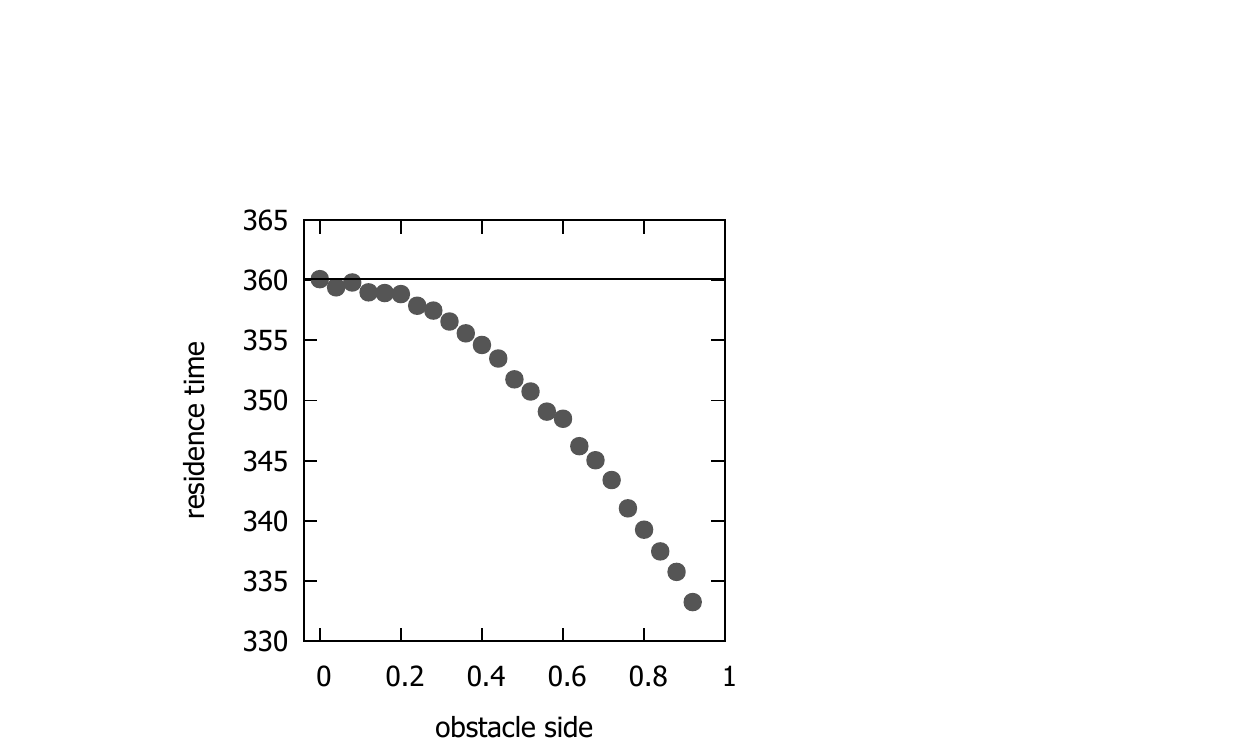}
}
\put(25,-5){
\definecolor{qqqqff}{rgb}{0.,0.,1.}
\definecolor{xdxdff}{rgb}{0.49019607843137253,0.49019607843137253,1.}
\definecolor{uququq}{rgb}{0.25098039215686274,0.25098039215686274,0.25098039215686274}
\definecolor{eqeqeq}{rgb}{0.8784313725490196,0.8784313725490196,0.8784313725490196}
\begin{tikzpicture}[line cap=round,line join=round,>=triangle 45,x=0.5cm,y=0.5cm]
\clip(-0.5176496608172668,-0.5023924992404349) rectangle (9.515351336167397,3.501237469349578);
\fill[color=eqeqeq,fill=eqeqeq,fill opacity=0.10000000149011612] (0.,0.) -- (0.,3.) -- (9.,3.) -- (9.,0.) -- cycle;
\fill[fill=black,fill opacity=0.10000000149011612] (9.,3.) -- (9.5,3.) -- (9.5,0.) -- (9.,0.) -- cycle;
\fill[fill=black,fill opacity=0.10000000149011612] (0.,3.) -- (-0.5,3.) -- (-0.5,0.) -- (0.,0.) -- cycle;
\fill[fill=black,fill opacity=0.10000000149011612] (4.4,2.6) -- (4.6,2.6) -- (4.6,0.4) -- (4.4,0.4) -- cycle;
\draw [color=uququq] (0.,0.)-- (0.,3.);
\draw [color=uququq] (0.,3.)-- (9.,3.);
\draw [color=uququq] (9.,3.)-- (9.,0.);
\draw [color=uququq] (9.,0.)-- (0.,0.);
\draw (9.,3.)-- (9.5,3.);
\draw (9.5,3.)-- (9.5,0.);
\draw (9.5,0.)-- (9.,0.);
\draw (9.,0.)-- (9.,3.);
\draw (0.,3.)-- (-0.5,3.);
\draw (-0.5,3.)-- (-0.5,0.);
\draw (-0.5,0.)-- (0.,0.);
\draw (0.,0.)-- (0.,3.);
\draw (4.4,2.6)-- (4.6,2.6);
\draw (4.6,2.6)-- (4.6,0.4);
\draw (4.6,0.4)-- (4.4,0.4);
\draw (4.4,0.4)-- (4.4,2.6);
\draw [->] (4.4,1.6) -- (3.4,1.6);
\draw [->] (4.6,1.6) -- (5.6,1.6);
\begin{scriptsize}
\draw[color=black] (2.546686424710366,2.288597693839765) node {$\Omega$};
\end{scriptsize}
\end{tikzpicture}

}
\put(273,-5){
\definecolor{qqqqff}{rgb}{0.,0.,1.}
\definecolor{xdxdff}{rgb}{0.49019607843137253,0.49019607843137253,1.}
\definecolor{uququq}{rgb}{0.25098039215686274,0.25098039215686274,0.25098039215686274}
\definecolor{eqeqeq}{rgb}{0.8784313725490196,0.8784313725490196,0.8784313725490196}
\begin{tikzpicture}[line cap=round,line join=round,>=triangle 45,x=0.5cm,y=0.5cm]
\clip(-0.55,-1.) rectangle (9.55,4.);
\fill[color=eqeqeq,fill=eqeqeq,fill opacity=0.10000000149011612] (0.,0.) -- (0.,3.) -- (9.,3.) -- (9.,0.) -- cycle;
\fill[fill=black,fill opacity=0.10000000149011612] (9.,3.) -- (9.5,3.) -- (9.5,0.) -- (9.,0.) -- cycle;
\fill[fill=black,fill opacity=0.10000000149011612] (0.,3.) -- (-0.5,3.) -- (-0.5,0.) -- (0.,0.) -- cycle;
\fill[color=uququq,fill=uququq,fill opacity=0.10000000149011612] (4.,2.) -- (5.,2.) -- (5.,1.) -- (4.,1.) -- cycle;
\draw [color=uququq] (0.,0.)-- (0.,3.);
\draw [color=uququq] (0.,3.)-- (9.,3.);
\draw [color=uququq] (9.,3.)-- (9.,0.);
\draw [color=uququq] (9.,0.)-- (0.,0.);
\draw (9.,3.)-- (9.5,3.);
\draw (9.5,3.)-- (9.5,0.);
\draw (9.5,0.)-- (9.,0.);
\draw (9.,0.)-- (9.,3.);
\draw (0.,3.)-- (-0.5,3.);
\draw (-0.5,3.)-- (-0.5,0.);
\draw (-0.5,0.)-- (0.,0.);
\draw (0.,0.)-- (0.,3.);
\draw [color=uququq] (4.,2.)-- (5.,2.);
\draw [color=uququq] (5.,2.)-- (5.,1.);
\draw [color=uququq] (5.,1.)-- (4.,1.);
\draw [color=uququq] (4.,1.)-- (4.,2.);
\draw [->] (4.,1.5) -- (3.15,1.5);
\draw [->] (4.5,2.) -- (4.5,2.85);
\draw [->] (5.,1.5) -- (5.85,1.5);
\draw [->] (4.5,1.) -- (4.5,0.15);
\begin{scriptsize}
\draw[color=black] (2.546686424710366,2.288597693839765) node {$\Omega$};
\end{scriptsize}
\end{tikzpicture}

}
\end{picture}
\caption{Residence time 
vs.\ width of a centered rectangular obstacle with 
fixed height $0.8$ (on the left) and 
vs.\ the side length of a centered squared obstacle 
(on the right). 
Simulation parameters: 
$L_1=4$, 
$L_2=1$,
$t_m=2\cdot 10^{-2} $,
total number of inserted particles $10^{8}$,
the total number of particles exiting through the right boundary varies 
from $4.2 \cdot 10^{5}$ to $1.1\cdot 10^{5}$ (on the left)
and 
from $5.3\cdot 10^5 $ to $1.3\cdot 10^5$ (on the right)
depending on the obstacle width.
The solid lines represent the value of the residence time measured 
for the empty strip (no obstacle).
}
\label{f:ct3}
\end{figure}

In Figure~\ref{f:ct4} we show that, and this is really surprising,
the residence time is not monotonic even
as a function of the position of the center of the obstacle. 
In the left panel a squared obstacle of side length $0.8$ is considered, 
whereas in the right panel a thin rectangular obstacle $0.04\times0.8$ 
is placed in the strip. 
In both cases the residence time is not monotonic and attains its 
minimum value when the obstacle is placed in the center of the strip. 
It is worth noting, that in the case on the left when the position 
of the center lays between $1.5$ and $2.5$ the residence time in 
presence of the obstacles is smaller than the corresponding value for 
the empty strip. 

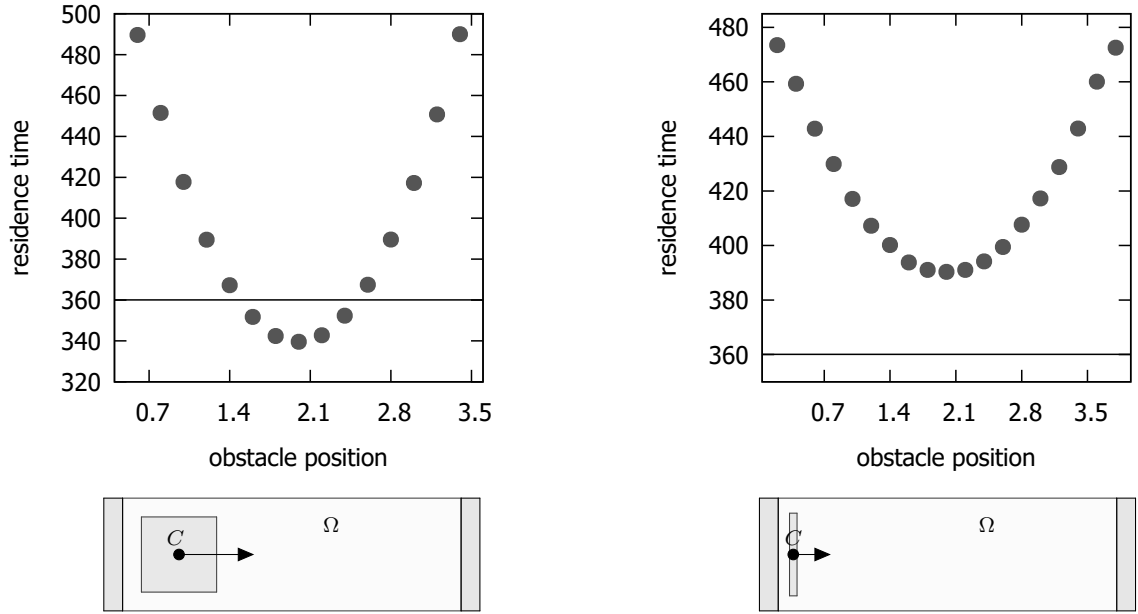
\begin{figure}[t!]
\begin{picture}(200,237)(0,0)
\put(-55,60){
  \includegraphics[width=0.88\textwidth]{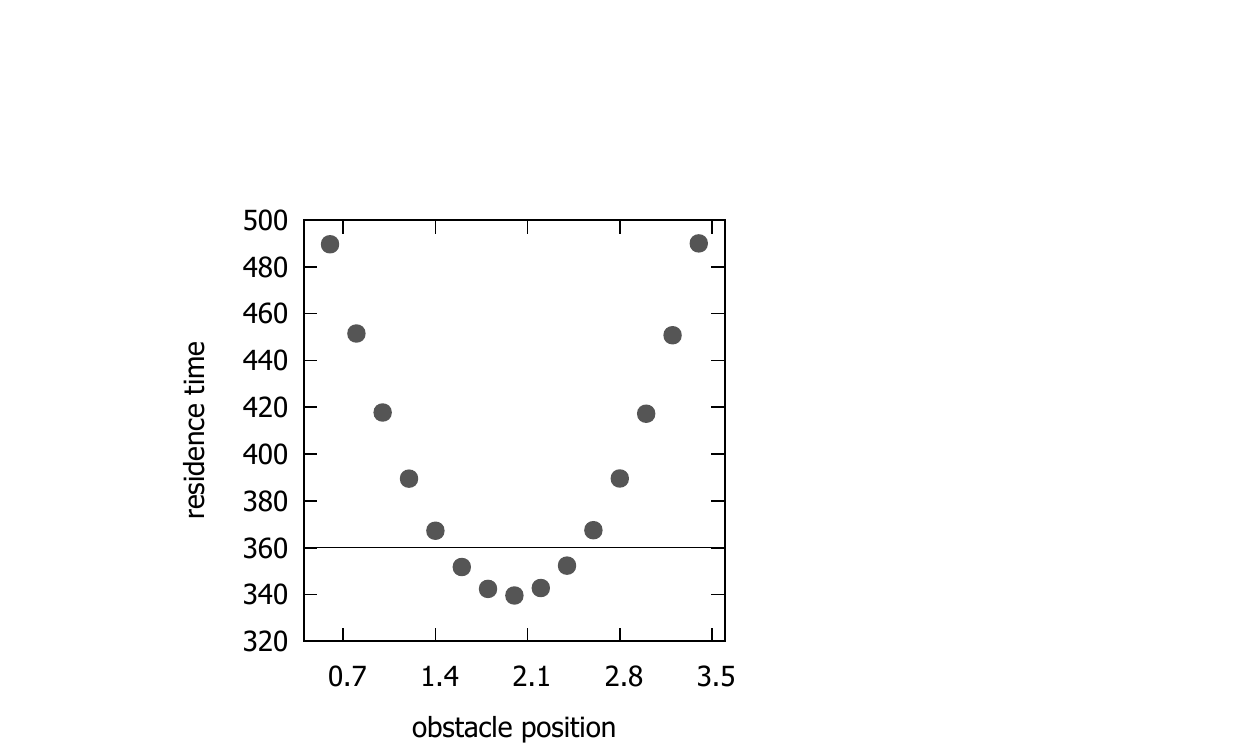}
}
\put(190,60){
  \includegraphics[width=0.88\textwidth]{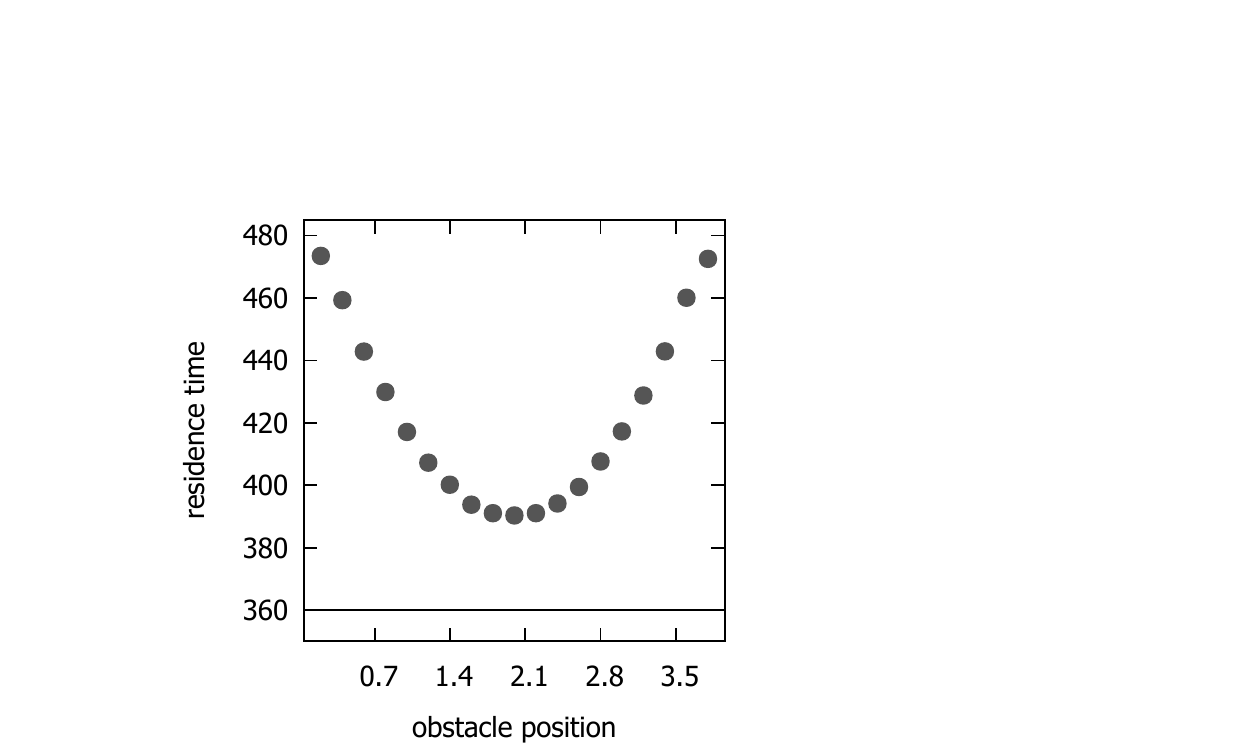}
}
\put(25,-5){
\definecolor{qqqqff}{rgb}{0.,0.,1.}
\definecolor{xdxdff}{rgb}{0.49019607843137253,0.49019607843137253,1.}
\definecolor{uququq}{rgb}{0.25098039215686274,0.25098039215686274,0.25098039215686274}
\definecolor{eqeqeq}{rgb}{0.8784313725490196,0.8784313725490196,0.8784313725490196}
\begin{tikzpicture}[line cap=round,line join=round,>=triangle 45,x=0.5cm,y=0.5cm]
\clip(-0.55,-1.) rectangle (9.55,4.);
\fill[color=eqeqeq,fill=eqeqeq,fill opacity=0.10000000149011612] (0.,0.) -- (0.,3.) -- (9.,3.) -- (9.,0.) -- cycle;
\fill[fill=black,fill opacity=0.10000000149011612] (9.,3.) -- (9.5,3.) -- (9.5,0.) -- (9.,0.) -- cycle;
\fill[fill=black,fill opacity=0.10000000149011612] (0.,3.) -- (-0.5,3.) -- (-0.5,0.) -- (0.,0.) -- cycle;
\fill[color=uququq,fill=uququq,fill opacity=0.10000000149011612] (0.5,2.5) -- (2.5,2.5) -- (2.5,0.5) -- (0.5,0.5) -- cycle;
\draw [color=uququq] (0.,0.)-- (0.,3.);
\draw [color=uququq] (0.,3.)-- (9.,3.);
\draw [color=uququq] (9.,3.)-- (9.,0.);
\draw [color=uququq] (9.,0.)-- (0.,0.);
\draw (9.,3.)-- (9.5,3.);
\draw (9.5,3.)-- (9.5,0.);
\draw (9.5,0.)-- (9.,0.);
\draw (9.,0.)-- (9.,3.);
\draw (0.,3.)-- (-0.5,3.);
\draw (-0.5,3.)-- (-0.5,0.);
\draw (-0.5,0.)-- (0.,0.);
\draw (0.,0.)-- (0.,3.);
\draw [color=uququq] (0.5,2.5)-- (2.5,2.5);
\draw [color=uququq] (2.5,2.5)-- (2.5,0.5);
\draw [color=uququq] (2.5,0.5)-- (0.5,0.5);
\draw [color=uququq] (0.5,0.5)-- (0.5,2.5);
\draw [->] (1.5,1.5) -- (3.5,1.5);
\begin{scriptsize}
\draw [fill=black] (1.5,1.5) circle (2.pt);
\draw[color=black] (1.4090575422785456,1.93586434957070753) node {$C$};
\draw[color=black] (5.546686424710366,2.288597693839765) node {$\Omega$};
\end{scriptsize}
\end{tikzpicture}

}
\put(273,-5){

\definecolor{qqqqff}{rgb}{0.,0.,1.}
\definecolor{xdxdff}{rgb}{0.49019607843137253,0.49019607843137253,1.}
\definecolor{uququq}{rgb}{0.25098039215686274,0.25098039215686274,0.25098039215686274}
\definecolor{eqeqeq}{rgb}{0.8784313725490196,0.8784313725490196,0.8784313725490196}
\begin{tikzpicture}[line cap=round,line join=round,>=triangle 45,x=0.5cm,y=0.5cm]
\clip(-0.55,-1.) rectangle (9.55,4.);
\fill[color=eqeqeq,fill=eqeqeq,fill opacity=0.10000000149011612] (0.,0.) -- (0.,3.) -- (9.,3.) -- (9.,0.) -- cycle;
\fill[fill=black,fill opacity=0.10000000149011612] (9.,3.) -- (9.5,3.) -- (9.5,0.) -- (9.,0.) -- cycle;
\fill[fill=black,fill opacity=0.10000000149011612] (0.,3.) -- (-0.5,3.) -- (-0.5,0.) -- (0.,0.) -- cycle;
\fill[color=uququq,fill=uququq,fill opacity=0.10000000149011612] (0.3,2.6) -- (0.5,2.6) -- (0.5,0.4) -- (0.3,0.4) -- cycle;
\draw [color=uququq] (0.,0.)-- (0.,3.);
\draw [color=uququq] (0.,3.)-- (9.,3.);
\draw [color=uququq] (9.,3.)-- (9.,0.);
\draw [color=uququq] (9.,0.)-- (0.,0.);
\draw (9.,3.)-- (9.5,3.);
\draw (9.5,3.)-- (9.5,0.);
\draw (9.5,0.)-- (9.,0.);
\draw (9.,0.)-- (9.,3.);
\draw (0.,3.)-- (-0.5,3.);
\draw (-0.5,3.)-- (-0.5,0.);
\draw (-0.5,0.)-- (0.,0.);
\draw (0.,0.)-- (0.,3.);
\draw [color=uququq] (0.3,2.6)-- (0.5,2.6);
\draw [color=uququq] (0.5,2.6)-- (0.5,0.4);
\draw [color=uququq] (0.5,0.4)-- (0.3,0.4);
\draw [color=uququq] (0.3,0.4)-- (0.3,2.6);
\draw [->] (0.4,1.5) -- (1.4,1.5);
\begin{scriptsize}
\draw [fill=black] (0.4,1.5) circle (2.pt);
\draw[color=black] (0.4090575422785456,1.93586434957070753) node {$C$};
\draw[color=black] (5.546686424710366,2.288597693839765) node {$\Omega$};
\end{scriptsize}
\end{tikzpicture}

}
\end{picture}
\caption{Residence time 
vs.\ position of the center of the obstacle.
The obstacle is a square of side length $0.8$ on the left and a
rectangle of side lengths $0.04$ and $0.8$ on the right. 
Simulation parameters: 
$L_1=4$, 
$L_2=1$,
$t_m=2\cdot 10^{-2}$,
total number of inserted particles $10^8$,
the total number of particles exiting through the right boundary is stable at the order of $2.6\cdot 10^5$ (on the left)
and 
of $4\cdot 10^5$ (on the right)
not depending on the obstacle position.
The solid lines represent the value of the residence time measured 
for the empty strip (no obstacle).
}
\label{f:ct4}
\end{figure}

Summarizing, the numerical experiments reported in 
Figures~\ref{f:ct1}--\ref{f:ct4} show that the residence time 
strongly depends on the obstacle geometry and position. 
In particular it is seen that large centered obstacles favor the selection 
of particles crossing the strip faster than in the empty strip 
case. 

A possible interpretation of these results can be given.
The strip $(0,L_1)\times(0,L_2)$ is partitioned in the three 
rectangles $L$ (the part on the left of the obstacles), 
$R$ (the part on the right of the obstacles), 
and $C=(0,L_1)\times(0,L_2)\setminus(L\cup R)$.
The phenomenon we reported above can be explained as a consequence 
of two competing effects: the total time spent by a particle in 
the channels
between the obstacle and the horizontal boundaries 
is smaller with respect to the time typically spent in $C$  
in the empty strip case because of the volume reduction 
due to the presence of the obstacle. 
On the other hand the times spent 
in L and in R 
are larger if compared to the times spent there by a particle in the 
empty strip case, due to the fact that it is more difficult to leave 
these regions and enter in the channels flankig the obstacle. 
The increase or the decrease of the residence 
time compared to the empty strip case depends on which of the two effects 
dominates the particle dynamics. 

\begin{figure}[!ht]
\begin{picture}(200,140)(0,0)
\put(-26,-30){
  \includegraphics[width=0.72\textwidth]{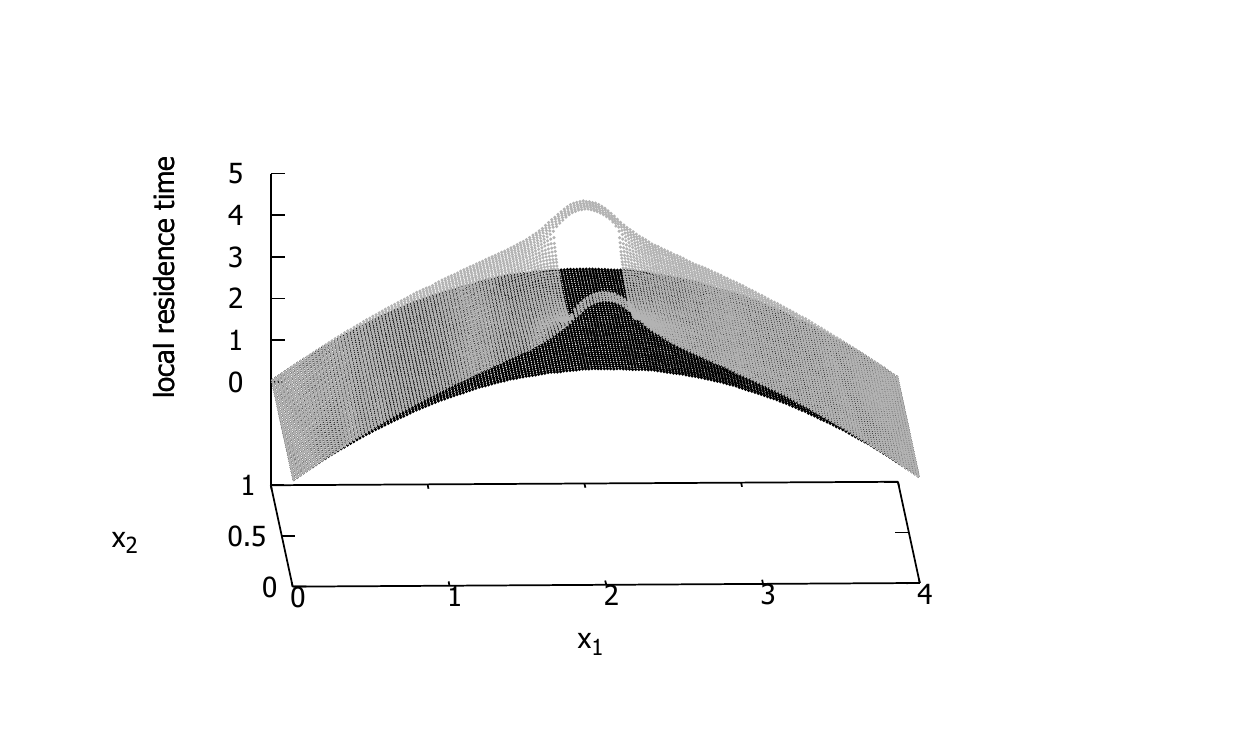}
}
\put(240,-10){
  \includegraphics[width=0.72\textwidth]{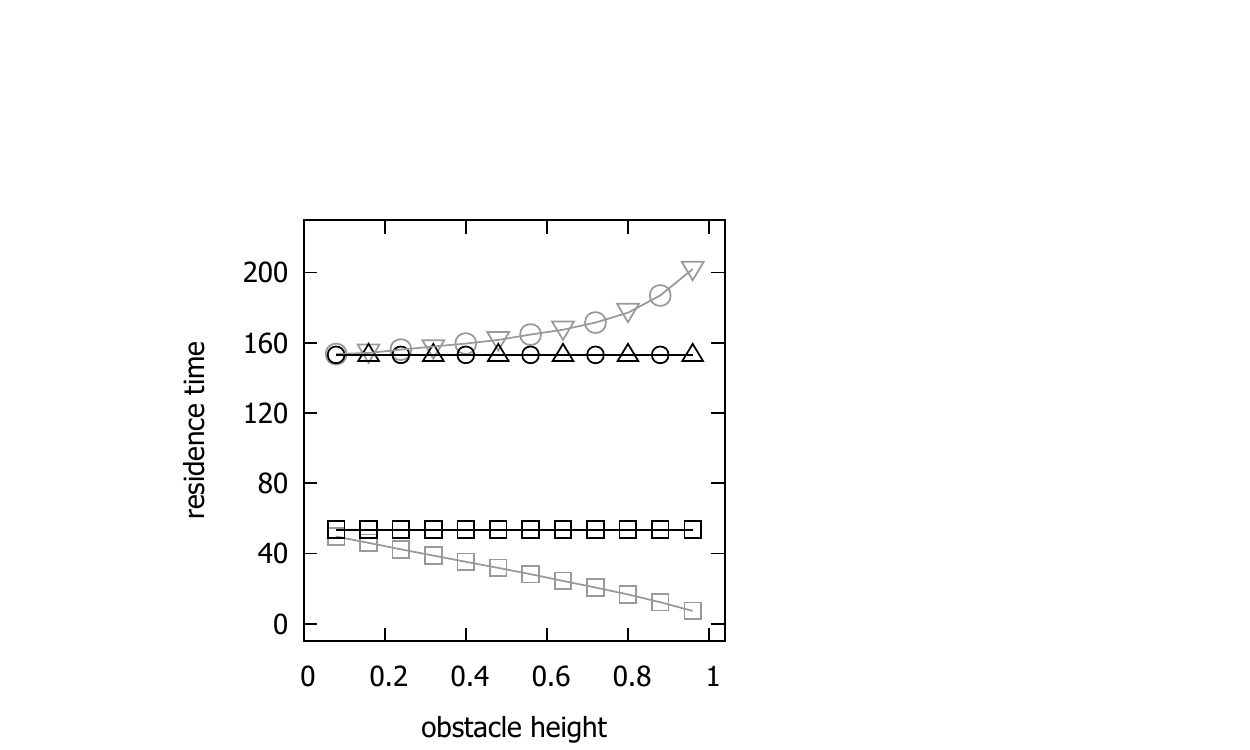}
}
\end{picture}
\caption{As in the right panel in Figure~\ref{f:ct1}. In the left 
panel the height of the obstacle is equal to $0.8$. 
Left panel: the mean time spent by particles crossing the strip 
in each point of the strip 
($0.02\times0.02$ cells have been considered) for the empty strip case 
(black) and in presence of the obstacle (gray).
Right panel: 
residence time in regions L (circles), C (squares), and R (triangles)
in presence of the obstacle (gray) and 
for the empty strip case (black).
}
\label{f:occ2}
\end{figure}

In Figure~\ref{f:occ2} we consider the geometry 
in the right panel of Figure~\ref{f:ct1}. 
We compute the average time spent by particles in small squared cells 
($0.02\times0.02$). This local residence time in presence of the obstacle 
is larger than the one measured in the empty strip, indeed the 
gray surface in the picture is always above the black one. 
But, if the total residence time spent in the regions 
$L,C$, and $R$ is computed, one discovers that the time spent in the 
region $C$ decreases in presence of the obstacle, whereas the time spent in 
$L$ and $R$ increases. 
Note that the local residence time in the cells belonging to the 
channels in $C$ is larger with respect to the empty strip case, 
but the total time in $C$ is smaller due to the fact that the available
volume in $C$ is decreased by the presence of the obstacle.
Hence, the result in the right panel in Figure~\ref{f:ct1}
can be explained as follows: if the height of the obstacle 
is smaller than $0.6$ the effect in $C$ dominates the one in $L$ and $R$
so that the total residence time decreases. On the other hand, 
when the height is larger than $0.6$ the increase of the 
residence time in $L$ and $R$ dominates its decrease 
in $C$, so that the total residence time increases.

\begin{figure}[!ht]
\begin{picture}(200,140)(0,0)
\put(-26,-30){
  \includegraphics[width=0.72\textwidth]{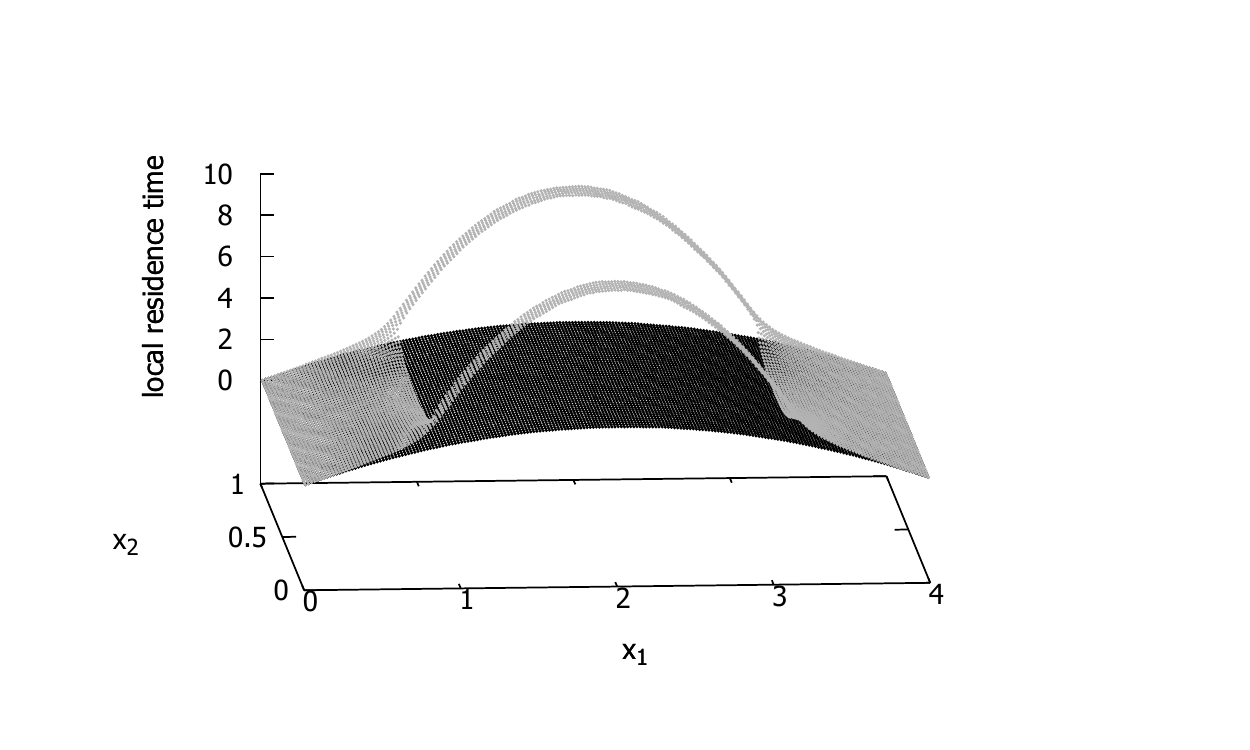}
}
\put(240,-10){
  \includegraphics[width=0.72\textwidth]{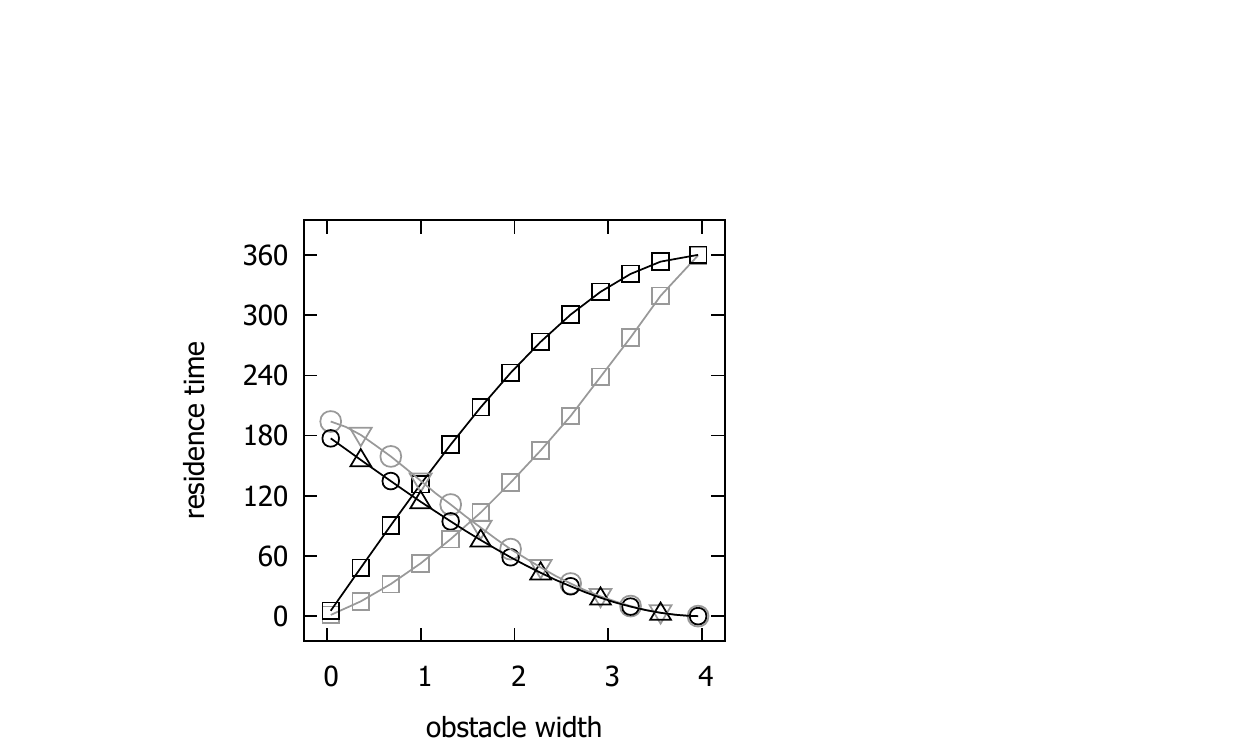}
}
\end{picture}
\caption{As in Figure~\ref{f:occ2} for the geometry 
in the left panel in Figure~\ref{f:ct3}.
In the left panel the width of the obstacle is $2.28$.
}
\label{f:occ5}
\end{figure}

\begin{figure}[!ht]
\begin{picture}(200,140)(0,0)
\put(-26,-30){
  \includegraphics[width=0.72\textwidth]{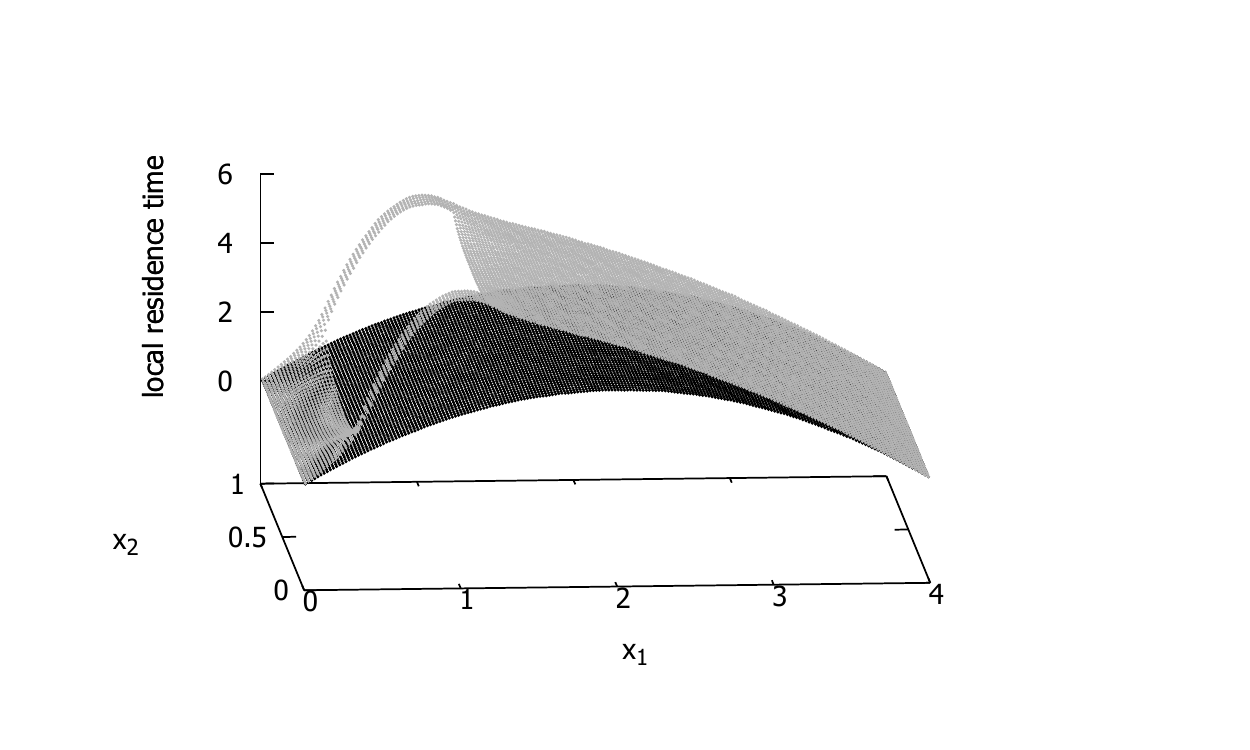}
}
\put(240,-10){
  \includegraphics[width=0.72\textwidth]{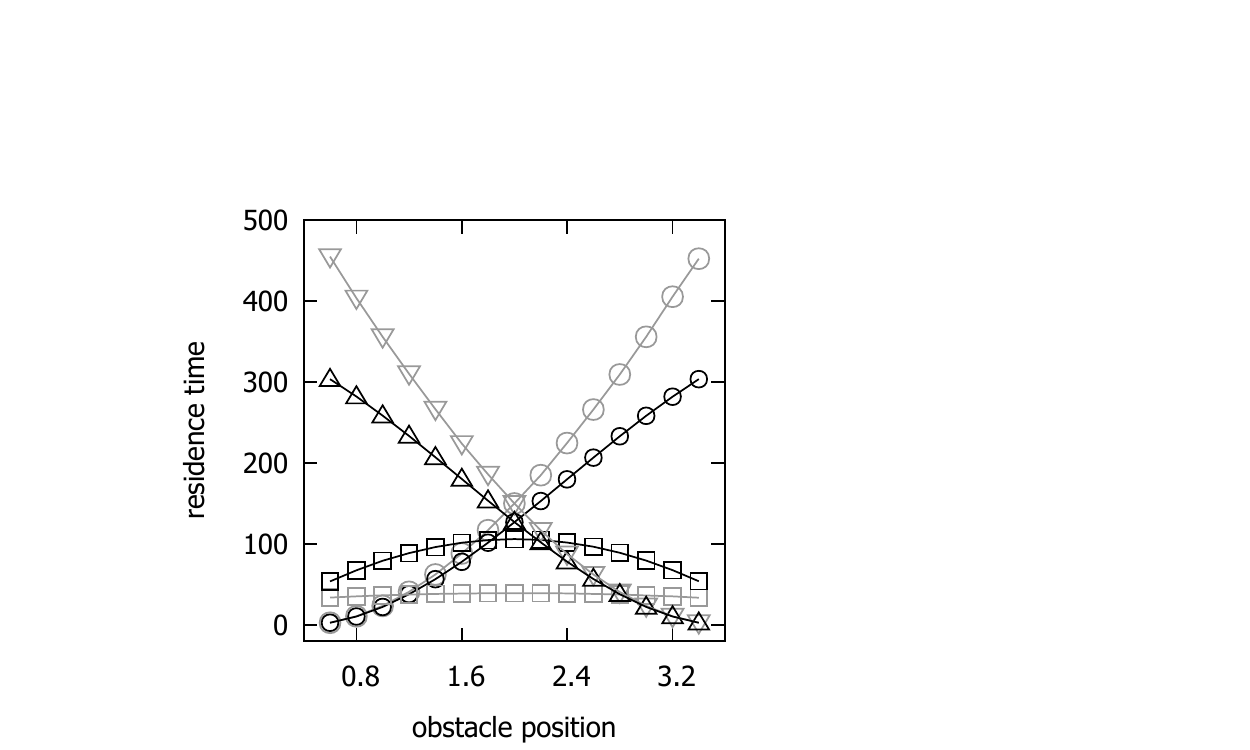}
}
\end{picture}
\caption{As in Figure~\ref{f:occ2} for the geometry 
in the left panel in Figure~\ref{f:ct4}.
In the left panel the position of the center of the obstacle is $0.8$.
}
\label{f:occ7}
\end{figure}

The Figure~\ref{f:occ5}, referring to the geometry 
in the left panel in Figure~\ref{f:ct3}, and 
the Figure~\ref{f:occ7}, referring to the geometry  
in the left panel in Figure~\ref{f:ct4},
can be discussed similarly. 
We just note that in Figures~\ref{f:occ2} and \ref{f:occ5}
the circles and triangles, which correspond to the residence 
time in $L$ and $R$, coincide due to the symmetry of the 
system. Indeed, in both cases the center of the obstacle is 
at the center of the strip.

\vspace{0.4cm}
\section{Proof of results}
\label{s:dimo} 
We prove Theorems \ref{th:ex!g_stat} and \ref{th:g_to_rho}.
We firstly construct the solution of the linear Boltzmann problem in form of a Dyson series.
Then we are able to prove the existence and uniqueness of the associated stationary problem. 
To do this we exploit the diffusive limit of the linear Boltzmann equation in a $L^\infty$ setting and in a bigger domain containing $\Omega$, by means of the Hilbert expansion method (see \cite{BNPP, CIP, EP}). 
The stationary solution is constructed in the form of a Neumann series to avoid the exchange of the limits $t\to \infty$, $\varepsilon\to 0$, following the idea of \cite{BNPP}.
Eventually we prove the convergence of the stationary state to the solution of the mixed Laplace problem. This also requires the Hilbert expansion method.
The auxiliary results stated are proved after the main theorems.

Let us consider the  problem  (\ref{eq:boltz_lin})-(\ref{eq:bound_cond_h_eps}) with the  datum  $g_\varepsilon(x,v,0)=f_0(x,v)\in L^\infty(\Omega\times S^1)$. 
We can express the operator defined  in (\ref{eq:op_L}) as $\cc{L} f(v) = 2 \lambda (\cc{K}-\cc{I})f(v)$ , where 
\begin{equation}\label{eq:def_K}
(\cc{K}f)(v)=\frac{1}{2} \int_{-1}^1 \ud \delta \, f(v')
\end{equation}
and $\cc{I}$ is the identity. Therefore the equation \eqref{eq:boltz_lin} can be written as 
\begin{equation}\label{eq:4.14}
\partial_t f + (v\cdot \nabla_x + 2 \eta_\varepsilon \lambda \cc{I}  )f= 2 \eta_\varepsilon \lambda \cc{K} f. 
\end{equation} 

We want to exploit the Duhamel's principle and express the solution as a series expansion.
We consider the semigroup  generated by $\cc{A}=(v\cdot \nabla_x + 2 \eta_\varepsilon \lambda \cc{I}  )$. We recall that in the whole plane $\bb{R}^2$ this semigroup acts as  $e^{-t\cc{A}} f(x,v)= e^{- 2\lambda\eta_\varepsilon t} f(x-vt,v) $, while the semigroup generated only by the transport term $v\cdot\nabla_x $ would be $e^{-t \,v\cdot \nabla}f(x,v)=f(x-vt,v)$.

We want to consider the semigroup generated by $\cc{A}$ on our domain $\Omega$ initial datum $f_0$ and  boundary conditions \eqref{eq:bound_cond_h_eps}.
Recall that $\partial \Omega_E$ is a specular reflective boundary while on $\partial \Omega_L \cup \partial \Omega_R$ the system is in contact with reservoirs with particle densities $f_B(x,v)$.
Since the equation describes the evolution of a particle moving in the space with velocity of modulus one, having random collisions with impact parameter $\delta$,
for any sequence of collision  times and impact parameters $t_i$, $\delta_i$ we can construct the backward trajectory of a particle as long as it stays in $\Omega$. 
Indeed the backward trajectory for a particle in $(x,v)$ at time $t$ starts by letting the particle move with velocity $-v$. For a time $t-t_1$ it does not hit any scatterers, but if the particle reaches the elastic boundary $\partial\Omega_E$ during this time, the velocity $-v$ becomes $-v'$ following the elastic collision rule $-v'=-(v-2(n\cdot v) n) $, where $n$ is the inward pointing normal to $\Omega$. 
After a time $t-t_1$ the particle performs a collision with impact parameter $\delta_1$ that produces the velocity $-v_1$. Then again the particles travels for a time $t_1-t_2$ elastically colliding if touching the boundary $\partial\Omega_E$ and so on until it reaches a reservoir or it has traveled for a total time $t$. 

In the same way, given the sequence $x$, $v$, $t_1$, $\ldots$, $t_m$, $\delta_1$, $\ldots$, $\delta_m$, we define the flow
$\Phi^{-t}_m (x,v, t_1,\ldots,t_m,\delta_1,\ldots,\delta_m)$ as the backward trajectory starting from $x$ with velocity $v$ and having $m$ transition in velocity obtained after a time $t-t_1$, $\ldots$, $t_i-t_{i+1}$, $\ldots$, $t_m$  ($i=1,\ldots,m-1$) by a scattering with an hard disk with impact parameter respectively $\delta_i$ ($i=1,\ldots,m$). 
We impose that the trajectories described by this flow $\Phi^{-t}(x,v)$ make a change of velocity from $v$ to $v'=v-2(n\cdot v)n$ any time the elastic boundary $\partial{\Omega}_E$ is reached.

We define the function $\tau=\tau(x,v,t,t_1,\ldots,t_m,\delta_1,\ldots,\delta_m)$ that represents the time when the particle that is in $(x,v)$ at time $t$ leaves a reservoirs and it enters into the strip.
So if the backward trajectory having collision times and parameters $t_1,\ldots,t_m,\delta_1,\ldots,\delta_m$ reaches the boundary $\partial\Omega_L \cup \partial\Omega_R$  in the time interval $[0,t]$, then  it happens after a backward time 
$t-\tau$. 
If the trajectory $\Phi^{-s} (x,v)$ never hits the boundary $\partial\Omega_L \cup \partial\Omega_R$ for any time $s\in [0,t]$ we set $\tau=0$.

We are now able to write  the solution $g_\varepsilon(x,v,t)$  using the Duhamel's principle. The semigroup generated by $\cc{A}$ on our domain has a transport term that we can express thanks to the flow $\Phi^{-t}(x,v)$, and the transported datum is $f_B$ or $f_0$ depending on the case the backward trajectory touches a reservoir in the time interval $[0,t]$ on not.
We use  the function $\tau$ to distinguish these two cases. We consider the collision operator $2 \lambda \eta_\varepsilon \cc{K}$ as the source term for the linear problem \eqref{eq:4.14}. So we construct the following expression for $g_\varepsilon$
\begin{equation}
\label{eq:duhamel_prima_iter}
\begin{split}
g_\varepsilon(x,v,t)= &e^{-2\lambda \eta_\varepsilon t} f_0(\Phi_0^{-t}(x,v))\chi(\tau=0) +
	e^{-2\lambda \eta_\varepsilon (t-\tau)} f_B(\Phi_0^{-(t-\tau)}(x,v)) \chi(\tau>0)	+
\\&	 +  \lambda\eta_\varepsilon \int_0^t e^{-2\lambda\eta_\varepsilon(t-t_1)} 2\cc{K} g_\varepsilon(\Phi_0^{-(t-t_1)}(x,v),t_1)  \chi(\tau < t_1) \,\mathrm{d} t_1.
\end{split}
\end{equation}  
The notation $\chi$ represents the characteristic function.

The meaning of \eqref{eq:duhamel_prima_iter} is clear: we separate the contribution given to $g_\varepsilon$ from trajectories transporting the initial datum $f_0$, having no collisions with  scatterers and never hitting a reservoirs in the time interval $[0,t]$; the contribution from trajectories transporting the initial datum $f_B$ exiting from a reservoirs at time $\tau$ and than moving in $\Omega$ until the time $t$ without colliding any scatterers; finally the last term is the contribution due to trajectories having the last collision with a scatterers at time $t_1$ and never touching the reservoirs in the time interval $[t_1,t]$.

We iterate the procedure by using \eqref{eq:duhamel_prima_iter} again for the $g_\varepsilon$ in the last integral and from \eqref{eq:def_K} we find:
\begin{equation}
\label{eq:duhamel_second_iter}
\begin{split}
g_\varepsilon(x,v,t)= &\;e^{-2\lambda \eta_\varepsilon t} f_0(\Phi^{-t}_0(x,v))\chi(\tau=0)+
	e^{-2\lambda \eta_\varepsilon (t-\tau)} f_B(\Phi^{-(t-\tau)}_0(x,v)) \chi(\tau>0)	
\\	& +  \lambda\eta_\varepsilon \int_0^t \,\mathrm{d} t_1 e^{-2\lambda\eta_\varepsilon(t-t_1)} \int_{-1}^1 \,\mathrm{d}\delta_1 
 \bigg[e^{-2\lambda \eta_\varepsilon t_1} f_0(\Phi_1^{-t}(x,v,t_1,\delta_1))\chi(\tau=0)
 \\&
	+e^{-2\lambda \eta_\varepsilon (t_1-\tau)} f_B(\Phi_1^{-(t-\tau)}(x,v,t_1,\delta_1)) \chi(\tau>0)\chi(\tau<t_1)	 
 \\&	 +  \lambda\eta_\varepsilon \int_0^{t_1}  \,\mathrm{d} t_2 \,  e^{-2\lambda\eta_\varepsilon(t_1-t_2)} \int_{-1}^1 \,\mathrm{d}\delta_2
 g_\varepsilon(\Phi_1^{-(t-t_2)}(x,v,t_1,\delta_1),t_2)  \chi(\tau < t_2)
\bigg]
.
\end{split}
\end{equation}

\noindent
By successive iterations we write the series expansion for the density of particles $g_\varepsilon(x,v,t)$ as
\begin{equation}
\label{eq:sol_h_eps_series}
\begin{split}
g_\varepsilon(x,v,t)=& e^{-2\lambda \eta_\varepsilon t} f_0(\Phi^{-t}_0(x,v))\chi(\tau=0)	+  \sum_{m\geq 1}e^{-2\lambda \eta_\varepsilon t}  (\lambda \eta_\varepsilon)^m \int_0^{t} \!\!\mathrm{d}t_1
\ldots \int_0^{t_{m-1}} \!\! \mathrm{d} t_m 
\\& 
 \int_{-1}^1 \!\!\mathrm{d} \delta_1 \ldots \int_{-1}^1 \!\! \mathrm{d} \delta_m \chi(\tau =0)  f_0(\Phi_m^{-t}(x,v,t_1,\ldots,t_m,\delta_1,\ldots,\delta_m)) 
\\& \!
+ e^{-2\lambda \eta_\varepsilon (t-\tau)} f_B(\Phi^{-(t-\tau)}_0(x,v)) \chi(\tau>0) +  \sum_{m\geq 1}\! (\lambda \eta_\varepsilon)^m \!\!\int_0^{t}\! \!\!\mathrm{d}t_1
\ldots \int_0^{t_{m-1}}\! \!\! \mathrm{d} t_m \!\int_{-1}^1 \!\!\mathrm{d} \delta_1 \ldots 
\\&
\int_{-1}^1 \!\! \mathrm{d} \delta_m \chi(\tau <t_m) \chi(\tau>0) e^{-2\lambda \eta_\varepsilon(t-\tau)} f_B(\Phi_m^{-(t-\tau)}(x,v,t_1,\ldots,t_m,\delta_1,\ldots,\delta_m))=
\\ =& g_\varepsilon^{in} (x,v,t) + g_\varepsilon^{out}(x,v,t).
\end{split}
\end{equation}
Note that the series are clearly convergent in $L^{+\infty}$.


In (\ref{eq:sol_h_eps_series}) the terms with $\chi(\tau=0)$ define $g_\varepsilon^{in}$ that represents the contributions to $g_\varepsilon$ due to trajectories that stay in $\Omega$ for every time in $[0,t]$ while the terms with $\chi(\tau>0)$ define $g_\varepsilon^{out}$ that collects the contributions due to trajectories leaving a mass reservoir at time $\tau>0$.

Note that $g_\varepsilon^{out}$ solves the problem (\ref{eq:boltz_lin})-(\ref{eq:bound_cond_h_eps}) with initial datum $f_0=0$.

We will use  the shorthand notation $\Phi^{-s}(x,v)$  instead of $\Phi_m^{-s}(x,v,t_1,\ldots,t_m,\delta_1,\ldots,\delta_m)$  where it is clear by the context to which sequence of collisions we refer.
Moreover, the terms with zero collision will be included in the series as the $m=0$ terms.

We denote by $S_\varepsilon$ acting on any  $h\in L^\infty (\Omega\times S^1)$ the Markov semigroup  associated to the $g_\varepsilon^{in}$ term in (\ref{eq:sol_h_eps_series}) for an initial datum $h$, namely
\begin{equation}\label{eq:def_semigr_S}
(S_\varepsilon(t)h)(x,v)=\! \sum_{m\geq 0}e^{-2\lambda \eta_\varepsilon t}  (\lambda \eta_\varepsilon)^m \! \int_0^{t} \!\!\mathrm{d}t_1
\ldots \int_0^{t_{m-1}} \!\! \mathrm{d} t_m  \!
 \int_{-1}^1 \!\!\mathrm{d} \delta_1 \ldots \int_{-1}^1 \!\! \mathrm{d} \delta_m \chi(\tau =0)  h(\Phi^{-t}(x,v)),
\end{equation}
so that in (\ref{eq:sol_h_eps_series}) $g_\varepsilon^{in}(t)=S_\varepsilon(t) f_0$.

\begin{proposition}
\label{pr:est_op_S_eps}
There exists $\varepsilon_0>0$ such that for any $\varepsilon<\varepsilon_0$ and for any $h \in L^\infty(\Omega \times S^1)$ it holds
\begin{equation}\label{eq:S_eps_est}
\|S_\varepsilon (\eta_\varepsilon) h \|_\infty \leq \alpha \| h \|_\infty,\qquad \alpha < 1.
\end{equation}
\end{proposition}
 
 Note that in the estimate (\ref{eq:S_eps_est}) we are considering $t=\eta_\varepsilon$ and the estimate is saying that there is a strictly positive probability for a backward trajectory to exit from $\Omega$ in a time of the order of $\eta_\varepsilon$.

\begin{proof}[Proof of Theorem \ref{th:ex!g_stat}]
From \eqref{eq:sol_h_eps_series}
the stationary solution $g_\varepsilon^S$ of the problem (\ref{eq:stationary problem}) verifies
\begin{equation*}
g^S_\varepsilon= g_\varepsilon^{out}(t_0) + S_\varepsilon(t_0)g_\varepsilon^S,
\end{equation*}
for every  $t_0>0$. We can formally write it by iterating the previous one in the form of the Neumann series 
\begin{equation}
\label{eq:neumann_ser}
g_\varepsilon^S= \sum_{N\geq 0} (S_\varepsilon(t_0))^N g_\varepsilon^{out}(t_0).
\end{equation}

In order to verify the existence and uniqueness of $g^S_\varepsilon$ we  show that the \eqref{eq:neumann_ser} converges. 
Indeed from Proposition \ref{pr:est_op_S_eps} 
 and \eqref{eq:neumann_ser}, chosen $t_0=\eta_\varepsilon$
\begin{displaymath}
\| g_\varepsilon^S \|_{\infty} \leq \sum_{N\geq 0} \| (S_\varepsilon(\eta_\varepsilon))^N g_\varepsilon^{out} \|_\infty
\leq \sum_{N>0} \alpha^{N}\|g_\varepsilon^{out}(\eta_\varepsilon)\|_\infty
 \leq \frac{1}{1-\alpha} \| g_\varepsilon^{out} \|_\infty
\leq \frac{1}{1-\alpha} \mathrm{max}\{\rho_L,\rho_R\}. 
\end{displaymath}

As a consequence the Neumann series \eqref{eq:neumann_ser} converges in $L^\infty$ and identifies a single element in $L^\infty$. Choosing an arbitrary $t_0$ bigger than $\eta_\varepsilon$ of the same order of $\eta_\varepsilon$ and thanks to the semigroup property of $S_\varepsilon$ it follows that $g^S_\varepsilon$ does not depend on the time $t_0$.
So there exists a unique stationary solution $g^S_\varepsilon \in L^\infty ({\Omega} \times S^1)$ satisfying (\ref{eq:stationary problem}).
\end{proof}

In order to prove Theorem \ref{th:g_to_rho} we need some properties of the linear Boltzmann operator $\cc{L}$ defined in \eqref{eq:op_L}. We summarize them in the next lemma.
 
\begin{lemma}
\label{le:L^-1}
Let $\mathcal{L}$ be the  operator  defined in \eqref{eq:op_L}, then $\cc{L}$ is a selfadjoint operator on  $L^2(S^1)$
 and has the form $\cc{L}=2\lambda (\cc{K} - \cc{I})$ where $\cc{K}$ is a selfadjoint and compact operator (in $L^2(S^1)$). Moreover, $\cc{K}$ is positive and its spectrum is contained in $[0,1]$. The value $0$ is the only accumulation point for the spectrum and $1$ is a simple eigenvalue.  So it holds that
 $\{\mathrm{Ker} \cc{L}\}^{\perp}=\{ h\in L^2 (S^1) :\int_{S^1} \mathrm{d} v\,\, h(v)=0 \}$ and there exists  $C>0$ such that for any $h \in L^\infty(S^1)$ that verifies $\int_{S^1} \mathrm{d} v\,\, h(v)=0$ we have
\begin{equation}
\|\mathcal{L}^{-1} h \|_\infty \leq C \| h \|_\infty.
\end{equation}
\end{lemma}

\begin{proof}[Proof of Lemma \ref{le:L^-1}]
The existence and the estimate of norm of $\cc{L}^{-1}$ are discussed in Lemma 4.1 from {Section 4.1} of \cite{BNPP}.
The compactness of the operator $\cc{K}$ and the spectral property of $\cc{L}$ are discussed in \cite{EP}.
\end{proof}

\begin{proof}[Proof of Theorem \ref{th:g_to_rho}]
The proof makes use of the Hilbert expansion (see e.g. \cite{BNPP,CIP,EP}). Assume that $g^S_\varepsilon$ has the following form
\begin{equation*}
g^S_\varepsilon(x,v)= g^{(0)}(x) + \sum_{k=1}^{+\infty} \bigg( \frac{1}{\eta_\varepsilon}  \bigg)^k  g^{(k)} (x,v),
\end{equation*}
where $g^{(k)}$ are not depending on $\eta_\varepsilon$.
We require $g^{(0)}$ to satisfy the same Dirichlet boundary conditions as the whole solution $g^S_\varepsilon$ on $\partial\Omega_L\cup\partial {\Omega}_R$:
\begin{equation}\label{eq:BD_g^0}
\begin{cases}
g^{(0)}(x)  = \rho_L  \qquad \quad & x\in \partial\Omega_L  \\
g^{(0)}(x)=  \rho_R  \qquad &  x\in \partial\Omega_R . 
\end{cases}
\end{equation}

By imposing that $g^S_\varepsilon$ solves \eqref{eq:stationary problem} and by comparing terms of the same order we get the following chain of equations:
\begin{equation*}
v\cdot \nabla_x g^{(k)} = \mathcal{L} g^{(k+1)},\qquad k\geq 0,
\end{equation*}
where we used that $\mathcal{L} g^{(0)}(x)=0$ since $g^{(0)}$ is independent of $v$.
The first two equations read
\begin{itemize}
\item[(i)]$v \cdot \nabla_x g^{(0)} (x) = \mathcal{L} g^{(1)}(x,v)$,
\item[(ii)] $v \cdot \nabla_x g^{(1)} (x,v)= \mathcal{L} g^{(2)}(x,v)$.
\end{itemize}
Let us consider the first one. By the Fredholm alternative, this equation has a solution if and only if the left hand side belongs to $(\rr{Ker}\cc{L})^\perp$. We recall that the null space of $\cc{L}$ is constituted by the constant functions (with respect to $v$), so we can solve equation $(i)$
 if and only if the left hand side belongs to 
 $(\rr{Ker} \mathcal{L})^{\perp} = \{ h\in L^2(S^1)\text{ such that }\int_{S^1} \mathrm{d} v \,\, h(v) = 0 \}$ (see Lemma \ref{le:L^-1}). 
Since $v\cdot \nabla_x g^{(0)}(x)$ is an odd function of $v$, it belongs to $(\rr{Ker} \mathcal{L})^{\perp}$. So we can invert the operator $\mathcal{L}$ and set
\begin{equation}
\label{eq:invert_L_g1}
 g^{(1)} (x,v) = \mathcal{L}^{-1} ( v \cdot \nabla_x g^{(0)}(x)) + \zeta^{(1)}(x), 
\end{equation}
where $\zeta^{(1)}(x)\in \rr{Ker}\, \mathcal{L}$ and $\cc{L}^{-1} (v\cdot \nabla_x g^0)$ is an odd function of $v$ since $\cc{L}^{-1}$ preserves the parity, namely it maps odd (even) function of $v$ in odd (even) functions (see \cite{EP}).

We integrate equation $(ii)$ with respect to the uniform measure on $S^1$. 
We can notice that $\int_{S^1} \mathrm{d} v \, v \cdot \nabla_x \zeta^{(1)}(x)=0$ ($\zeta^{(1)}$ depends only on $x$, so the function in the integral is odd in the velocity) and $\int_{S^1} \mathrm{d} v \, \mathcal{L} g^{(2)}=0$ (since operator $\cc{L}$ preserves mass), so by \eqref{eq:invert_L_g1} we obtain
\begin{equation}\label{eq:integr_ii}
\frac{1}{2\pi} \bigg(  \int_{S^1} 
\ud v \,\, v \cdot \nabla_x 
(\mathcal{L}^{-1} ( v \cdot \nabla_x g^{(0)}(x))) \bigg)  = 0.
\end{equation}
By expanding the scalar product and using the linearity of $\cc{L}^{-1}$ we get
\begin{equation}
 -\sum_{i,j=1}^2 D_{i,j}\partial_{x_i} \partial_{x_j} g^{(0)} (x)=0.
\end{equation}
We define the $2\times 2$ matrix $D_{i,j}=\frac{1}{2\pi} \int_{S^1} \rr{d} v \,v_i (-\cc{L}^{-1})v_j$  and we observe that $D_{ij}=0$ if $i\neq j$ as follows by the change $v_i \to -v_i$ while $D_{11}=D_{22}=D>0$ thanks to the isotropy and the spectral property of the operator (see \cite{EP}). Hence $D$ is given by the formula \eqref{eq:greenKubo}
\begin{equation*} 
D=\frac{1}{4\pi}\int_{S^1} \mathrm{d} v \,\, v \cdot (- \mathcal{L})^{-1} v,
\end{equation*} 
 and the integrated equation $(ii)$ becomes
\begin{equation}  \label{eq:Laplace_g^0}
\frac{1}{2\pi} \bigg(  \int_{S^1} 
\ud v \,\, v \cdot \nabla_x 
(\mathcal{L}^{-1} ( v \cdot \nabla_x g^{(0)}(x))) \bigg)  =0
\Leftrightarrow - D \Delta_x g^{(0)}(x)=0.
\end{equation}

We require $g^S_\varepsilon(x,v)$ to satisfy the reflective boundary condition $g^S_\varepsilon(x,v') = g^S_\varepsilon(x,v)$ on $\partial {\Omega}_E$. By imposing it on the first term $g^{(1)}(x,v) = g^{(1)}(x,v')$ for every $x\in \partial\Omega_E$, $v\cdot n <0$, from \eqref{eq:invert_L_g1} we obtain
\begin{equation}\label{eq:imp_neum_cond}
\mathcal{L}^{(-1)} (v \cdot \nabla_x g^{(0)}) + \zeta^{(1)}(x) = \mathcal{L}^{(-1)} (v' \cdot \nabla_x g^{(0)}) + \zeta^{(1)}(x).
\end{equation}
By means of the elastic collision rule $v'= v- 2(v\cdot n) n$, the linearity of $\cc{L}^{-1}$ allow us to write
\begin{equation*}
\mathcal{L}^{(-1)} ((v- 2(v\cdot n) n) \cdot \nabla_x g^{(0)}) =\mathcal{L}^{(-1)} (v \cdot \nabla_x g^{(0)}) - 2 (n \cdot \nabla_x g^{(0)}) \mathcal{L}^{(-1)}(v\cdot n)  .
\end{equation*}
Left and  right members in \eqref{eq:imp_neum_cond} are the same if and only if
$(n \cdot \nabla_x g^{(0)}) \mathcal{L}^{(-1)}(v\cdot n)=0$. 
Since $\int_{S^1} \ud v \, v\cdot n = 0$ we get $\mathcal{L}^{(-1)}(v\cdot n)\neq 0$, so the only possibility is  $(n \cdot \nabla_x g^{(0)}) = 0$.
 Therefore $g^{(0)}(x)$ has to satisfy the Neumann boundary conditions $\partial_n g^{(0)}(x)=0$, for all $x \in \partial{\Omega}_E$.

From the previous one, \eqref{eq:Laplace_g^0} and \eqref{eq:BD_g^0} we have shown that  the term $g^{(0)}(x)$ solves the problem
\begin{equation}\label{eq:pb_lap_h0}
\begin{cases}
\Delta_x g^{(0)}(x)=0\qquad \quad & x\in  {\Omega}   \\
g^{(0)}(x)  = \rho_L  \qquad & x\in \partial\Omega_L  \\
g^{(0)}(x)=  \rho_R  \qquad &  x\in \partial\Omega_R
\\ \partial_n g^{(0)}(x)=0\qquad & x \in \partial{\Omega}_E.
\end{cases}
\end{equation}
We can deal with this mixed problem following the method of \cite{Lad}, Chapt. II. Furthermore, regularity results guarantee $g^0\in C^\infty{(\overline{\Omega})}$ (see \cite{Eva}, Chapt. 6).

Since \eqref{eq:integr_ii} shows that $\int_{S^1} \ud v \,\, v\cdot \nabla_x  g^{(1)}=0$, we can invert $\mathcal{L}$ in equation (ii) to obtain
\begin{equation} \label{eq:g^2}
g ^{(2)}(x,v)= \mathcal{L}^{-1} (v \cdot \nabla_x \mathcal{L}^{-1} ( v \cdot \nabla_x g^{(0)}(x))) + \mathcal{L}^{-1} (v \cdot \nabla_x \zeta^{(1)}(x))  + \zeta^{(2)}(x),
\end{equation}
where $\zeta^{(2)}$ belongs to the kernel of $\mathcal{L}$.

Now, integrating the third equation  $v\cdot \nabla_x g^{(2)}(x)=\mathcal{L} g^{(3)}(x,v)$  with respect to the uniform measure on $S^1$, we find thanks to \eqref{eq:g^2}
\begin{equation}\label{eq:integr_iii}
\begin{split}
&\int_{S^1} \ud v \,\, v \cdot \nabla_x ( \mathcal{L}^{-1} (v \cdot \nabla_x \mathcal{L}^{-1} ( v \cdot \nabla_x g^{(0)}(x))) ) +\\ &+ \int_{S^1} \ud v \,\, v \cdot \nabla_x( \mathcal{L}^{-1} (v \cdot \nabla_x \zeta^{(1)}(x)))
 + \int_{S^1} \ud v \,\, v \cdot \nabla_x (\zeta^{(2)}(x)) =0.
\end{split}
\end{equation}
The last integral is null because of the independence of $\zeta^{(2)}(x)$ from $v$.
The first integral is null because the function in the integral is an odd function of the velocity thanks to the fact that the operator $\cc{L}^{-1}$ preserves the parity.
The  \eqref{eq:integr_iii} becomes 
\begin{equation}\label{eq:Laplace_zeta1}
 \int_{S^1} \ud v \,\, v \cdot \nabla_x( \mathcal{L}^{-1} (v \cdot \nabla_x \zeta^{(1)}(x)))= 
 -D \Delta_x \zeta^{(1)}(x) =0
\end{equation}
Since there are no restriction on the choice of the boundary condition, we impose the Dirichlet data $\zeta^{(1)}(x)=0$ on the boundary $\partial\Omega_L\cup\partial{\Omega}_R$. So that by the previous and \eqref{eq:Laplace_zeta1} we find
 $\zeta^{(1)}(x)\equiv 0$ and hence $g^{(1)} (x,v) = \mathcal{L}^{-1} ( v \cdot \nabla_x g^{(0)}(x))$.

Because of the \eqref{eq:Laplace_g^0} the first term of the right hand side of equation \eqref{eq:g^2} is null too. So \eqref{eq:g^2} reduces to  $g^{(2)}(x,v)= \zeta^{(2)}(x)$.

Moreover from the third equation we get, by inverting $\mathcal{L}$, 
\begin{center}
$g^{(3)}(x,v)=
 \mathcal{L}^{-1} ( v\cdot \nabla_x g^{(2)}(x,v) ) + \zeta^{(3)}(x) 
= \mathcal{L}^{-1} ( v\cdot \nabla_x \zeta^{(2)}(x) ) + \zeta^{(3)}(x),
$\end{center}
with $\zeta^{(3)}(x)$ belonging to $\rr{Ker} \mathcal{L}$.

By integrating on $S^1$ the fourth equation  $v\cdot \nabla_x g^{(3)} = \mathcal{L} g^{(4)}$  and by exploiting that $\int_{S^1} \ud v \,\, \cc{L}g^{4}(x,v)=0$ and that $\int_{S^1} \ud v \,\, v \cdot \nabla_x \zeta^{(3)}(x)=0$ 
 we find
\begin{equation}
\int_{S^1} \ud v \,\, v\cdot \nabla_x (\mathcal{L}^{-1} ( v\cdot \nabla_x \zeta^{(2)}(x) ) ) =
-D  \Delta_x \zeta^{(2)}(x)  =
0.
\end{equation}
We choose zero boundary condition at the reservoirs, namely $\zeta^{(2)}(x)=0$ on $\partial\Omega_L\cup \partial {\Omega}_R$, so we find $\zeta^{(2)}(x)\equiv 0$. Then $ g^{(2)}(x,v)\equiv 0 $.

We can now write the expansion  for $g^S_\varepsilon$ as
\begin{equation}\label{eq:resto+g^i}
g^S_\varepsilon = g^{(0)} + \frac{1}{\eta_\varepsilon} g^{(1)} + \frac{1}{\eta_\varepsilon} R_{\eta_\varepsilon}.
\end{equation}
The remainder $R_{\eta_\varepsilon}$ satisfies 
\begin{equation} \label{eq:bol_st_resto}
v \cdot \nabla_x R_{\eta_\varepsilon} =  \eta_\varepsilon \mathcal{L} R_{\eta_\varepsilon} .
\end{equation}

We required $g^{(0)}$ to satisfy the same boundary conditions as the whole solution at contact with the reservoirs, namely on $\partial\Omega_L\cup\partial{\Omega}_R$, so the boundary conditions for $R_{\eta_\varepsilon}$ read
\begin{equation}
\begin{cases} \label{eq:resto_bordo}
R_{\eta_\varepsilon}(x,v)= - \mathcal{L}^{-1}(v \cdot \nabla_x g^{(0)}(x))
\qquad\quad & x \in \partial\Omega_L\cup\partial\Omega_R \,,\, v\cdot n(x)>0 ,
\\ R_{\eta_\varepsilon}(x,v') = R_{\eta_\varepsilon}(x,v)\qquad &  x\in \partial \tilde{\Omega}_E \,,\, v\cdot n(x) < 0 .
\end{cases}
\end{equation}

Note that the problem \eqref{eq:bol_st_resto}-\eqref{eq:resto_bordo} has the form of \eqref{eq:stationary problem}. From Theorem \ref{th:ex!g_stat} we know that it admits a unique solution in $L^\infty$.

From the \eqref{eq:resto+g^i}, thanks to the the fact that both $g^{(1)}$ and $R_{\eta_\varepsilon}$ are bounded in $L^\infty$ norm, we conclude that $g^S_\varepsilon\to g^{(0)}$.
\end{proof}

In order to prove Proposition \ref{pr:est_op_S_eps} we follow the strategy of the proof of {Proposition 3.1} in \cite{BNPP}. Here we have the additional difficulty of the specular reflective boundaries of horizontal sides of the strip and the presence of the obstacles in $\Omega$. 
 In the proof are exploited the diffusive limit of the linear Boltzmann equation in a $L^\infty$ setting and in a bigger domain containing $\Omega$ as stated in Proposition \ref{pr:rescal_bolt_to_heat} below and the properties of $\cc{L}$ summarized in Lemma \ref{le:L^-1}. 

We construct the extended domain $\Lambda$ as the infinite strip constructed by removing the left and right sides of $\Omega$ and keeping the upper and lower elastic boundaries at $x_2=0$ and $x_2=L_2$ and the obstacles into $\Omega$ (see Figure \ref{fig:lambda_domain}). We call $\partial\Lambda_E$ the union of upper and lower sides of $\Lambda$ with the obstacles boundaries.

\begin{figure}[h!]
\centering
\definecolor{wqwqwq}{rgb}{0.3764705882352941,0.3764705882352941,0.3764705882352941}
\definecolor{qqqqff}{rgb}{0.,0.,1.}
\definecolor{ffqqqq}{rgb}{1.,0.,0.}
\definecolor{eqeqeq}{rgb}{0.8784313725490196,0.8784313725490196,0.8784313725490196}
\begin{tikzpicture}[line cap=round,line join=round,>=triangle 45,x=1.cm,y=1.cm]
\clip(-1.2,-2.5) rectangle (10.2,3.75);
\fill[line width=1.2pt,color=eqeqeq,fill=eqeqeq,fill opacity=0.10000000149011612] (-2.,0.) -- (-2.,3.) -- (11.,3.) -- (11.,0.) -- cycle;
\draw  (0.,3.)-- (9.,3.);
\draw  (9.,0.)-- (0.,0.);
\draw (9.,3.)-- (11.,3.);
\draw (11.,0.)-- (9.,0.);
\draw (0.,3.)-- (-2.,3.);
\draw (-2.,0.)-- (0.,0.);
\draw [rotate around={73.90918365114764:(7.541710094234184,1.8497761157344577)},line width=1.pt,color=black,fill=black,fill opacity=0.10000000149011612] (7.541710094234184,1.8497761157344577) ellipse (0.5cm and 0.4cm);

\draw [rotate around={15:(3.3,1.2)},line width=1.pt,color=black,fill=black,fill opacity=0.10000000149011612] (3.3,1.2) ellipse (0.3cm and 0.5cm);



\begin{scriptsize}
\draw[color=black] (4.5373540388846925,.358072978331086) node {$\partial \Lambda_E$};
\draw[color=black] (5.810517845286029,2.307804558281278) node {$\Lambda$};
\end{scriptsize}
\end{tikzpicture}
\vspace{-2.3cm}
\caption{Domain $\Lambda$: infinite strip with big fixed obstacles: the whole boundaries of $\Lambda$ is a specular reflective boundary.}\label{fig:lambda_domain}
	\end{figure}
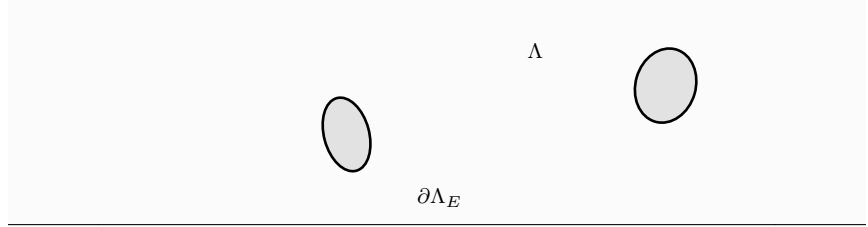

We introduce $h_\varepsilon: \Lambda\times S^1\times[0,T]\to \bb{R}^+$ the solution of the following rescaled linear Boltzmann equation
\begin{equation} \label{eq:rescaled_bolt}
\begin{cases}
(\partial_t + \eta_\varepsilon v\cdot \nabla_x) h_\varepsilon = \eta_\varepsilon^2 \cc{L} h_\varepsilon \qquad\quad& x\in \Lambda
\\ h_\varepsilon(x,v',t)=h_\varepsilon(x,v,t)\qquad& x\in \partial\Lambda_E,\,v\cdot n <0,\, t\geq 0
\\ h_\varepsilon(x,v,0)=\rho_0(x) \qquad\quad &x\in\Lambda ,
\end{cases}
\end{equation}
where $\rho_0(x)$ is a smooth function of the only variable $x$ (local equilibrium).

\begin{proposition}\label{pr:rescal_bolt_to_heat}  
Let  $h_\varepsilon$ be the solution of \eqref{eq:rescaled_bolt}, with an initial datum $\rho_0\in C^\infty(\overline{\Lambda})$ such that there exists $M>0$ with $\rho_0(x)=0$ if $|x|>M$ and $\partial_n \rho_0 (x)=0$ for $x\in\partial \Lambda_E$. Then, as $\varepsilon\to 0$, $h_\varepsilon$ converges to the solution of the heat equation
\begin{equation}\label{eq:heat_rescal_lambda}
\begin{cases}
\partial_t \rho - D\Delta\rho=0 \qquad\quad &x\in \Lambda
\\ \rho(x,0)=\rho_0(x) \qquad\quad &x \in \Lambda
\\ \partial_n \rho(x,t)=0 \qquad\quad& x\in\partial\Lambda_E,\, t\geq 0,
\end{cases}
\end{equation}
where the diffusion coefficient $D$ is given by the formula 
\begin{equation} \label{eq:greenKubo}
D=\frac{1}{4\pi}\int_{S^1} \mathrm{d} v \,\, v \cdot (- \mathcal{L})^{-1} v.
\end{equation}
The convergence is in $L^\infty([0,T];L^\infty(\Lambda\times S^1))$.
\end{proposition}


\begin{proof}[Proof of Proposition \ref{pr:est_op_S_eps}]

The semigroup $S_\varepsilon$ defined in \eqref{eq:def_semigr_S} can be equivalently written as extended to functions belonging to $L^{\infty}(\Lambda\times S^1)$, namely
\begin{equation}\label{eq:def_S_ext}
\begin{split}
(S_\varepsilon(t)f)(x,v)= & \chi_\Omega(x) \sum_{m\geq 0}e^{-2\lambda \eta_\varepsilon t}  (\lambda \eta_\varepsilon)^m \int_0^{t} \!\!\mathrm{d}t_1
\ldots \int_0^{m-1} \!\! \mathrm{d} t_m 
\\& 
 \int_{-1}^1 \!\!\mathrm{d} \delta_1 \ldots \int_{-1}^1 \!\! \mathrm{d} \delta_m \chi(\tau =0)  f(\Phi^{-t}(x,v)) \chi_\Omega( \Phi^{-t}(x)),
 \end{split}
\end{equation}
for any $f\in L^\infty(\Lambda \times S^1)$, where $\chi_\Omega$ is the characteristic function of $\Omega$ and $\Phi^{-t}(x)$ is the first component (the position) of $\Phi^{-t}(x,v)$, the backward flux individuated by $x$, $v$, $t_1$, $\ldots$, $t_m$, $\delta_1$, $\ldots$, $\delta_m$.
The addition of $\chi_\Omega( \Phi^{-t}(x))$ guarantees together with $\chi(\tau=0)$ that the dynamics stay internal to $\Omega$. Moreover, the following  estimate holds
\begin{equation*}
\begin{split}
S_\varepsilon(t) f\leq & \| f \|_\infty \sum_{m\geq 0}e^{-2\lambda \eta_\varepsilon t}  (\lambda \eta_\varepsilon)^m \int_0^{t} \!\!\mathrm{d}t_1
\ldots \int_0^{m-1} \!\! \mathrm{d} t_m 
 \int_{-1}^1 \!\!\mathrm{d} \delta_1 \ldots \int_{-1}^1 \!\! \mathrm{d} \delta_m   \chi_\Omega( \Phi^{-t}(x)).
 \end{split}
\end{equation*}
We construct  $\chi_\Omega^\delta\in C^\infty(\overline{\Lambda})$, a mollified version of $\chi_\Omega$, $\chi_\Omega^\delta\geq\chi_\Omega$, $\chi_\Omega^\delta \leq 1$ and $\Omega\subset supp(\chi_\Omega^\delta)\subset(-\delta,L_1+\delta)\times[0,L_2]$.
So we can write
\begin{equation}\label{eq:S_con_chi^delta}
\begin{split}
S_\varepsilon(t) f\leq & \| f \|_\infty \sum_{m\geq 0}e^{-2\lambda \eta_\varepsilon t}  (\lambda \eta_\varepsilon)^m \int_0^{t} \!\!\mathrm{d}t_1
\ldots \int_0^{m-1} \!\! \mathrm{d} t_m 
 \int_{-1}^1 \!\!\mathrm{d} \delta_1 \ldots \int_{-1}^1 \!\! \mathrm{d} \delta_m   \chi^\delta_\Omega( \Phi^{-t}(x)).
 \end{split}
\end{equation}
Note that the series in \eqref{eq:S_con_chi^delta} defines a function $F$ which solves
\begin{equation} \label{eq:bolt_lambda}
\begin{cases}
(\partial_t + v\cdot \nabla_x) F(x,v,t) = \eta_\varepsilon \cc{L} F(x,v,t) \qquad\quad& x\in \Lambda
\\ F(x,v',t)=F(x,v,t)\qquad& x\in \partial\Lambda_E,\,v\cdot n <0,\, t\geq 0
\\ F(x,v,0)=\chi_\Omega^\delta(x) \qquad\quad &x\in\Lambda .
\end{cases}
\end{equation}
Defining $G_\varepsilon(x,v,t)$ as $F(x,v,\eta_\varepsilon t)$,
 $G_\varepsilon$ solves \eqref{eq:rescaled_bolt} with initial datum $\rho_0=\chi_\Omega^\delta$. 
 Thanks to {Proposition} \ref{pr:rescal_bolt_to_heat} we know that at time $t=1$ 
 \begin{displaymath}
 \|G_\varepsilon(1)- \rho^\delta(1)\|_\infty\leq \omega(\varepsilon)
 \end{displaymath}
 where $\rho^\delta$ solves \eqref{eq:heat_rescal_lambda} with initial datum $\chi_\Omega^\delta$ and $\omega(\varepsilon)$ denotes a positive function vanishing with $\varepsilon$.
Moreover, we can notice that the function $\rho^\delta$ is the solution of  a diffusion equation with initial datum $0\leq \chi_\Omega^\delta \leq 1$ with support in a bounded subset of the infinite strip $\Lambda$.
By the strong maximum principle we know that for the positive time $t=1$, it holds that $\rho^\delta(x,1)<1$.
Therefore for $\varepsilon$ small enough 
\begin{equation}
\begin{split}
\|S_\varepsilon(\eta_\varepsilon) f \|_\infty \leq& \| f \|_\infty \| S_\varepsilon(\eta_\varepsilon) \chi_\Omega^\delta  \|_\infty \leq \| f \|_\infty ( \| G_\varepsilon(1) - \rho^\delta(1)\|_\infty +\| \rho^\delta(1) \|_\infty)
\\ \leq & \|f \|_\infty (\omega(\varepsilon) + \| \rho^\delta (1) \|_\infty) < \alpha \|f\|_\infty, \,\,\, \alpha<1,
\end{split}
\end{equation}
where we have used \eqref{eq:S_con_chi^delta} for $t=\eta_\varepsilon$.
\end{proof}

\begin{proof}[Proof of Proposition \ref{pr:rescal_bolt_to_heat}]
Let $h_\varepsilon:\Lambda\times S^1 \times [0,T]$ the solution of \eqref{eq:rescaled_bolt}. We use the Hilbert expansion technique 
to prove that $h_\varepsilon$ converges to the solution of the heat equation \eqref{eq:heat_rescal_lambda}. We search $h_\varepsilon$ of the form
\begin{equation*}
h_\varepsilon(x,v,t)= h^{(0)}(x,t) + \sum_{k=1}^{+\infty} \left(\frac{1}{\eta_\varepsilon}\right)^k h^{(k)} (x,v,t),
\end{equation*}
with coefficient $h^{(k)}$ not depending on $\eta_\varepsilon$.
By imposing that $h_\varepsilon$ solves \eqref{eq:rescaled_bolt} and comparing terms of the same order we find the identity $\cc{L}h^{(0)}(x,t)=0$ and the chain of equations
\begin{displaymath}
\begin{split}
v\cdot \nabla_x h^{(0)}=&\cc{L} h^{(1)}
\\
\partial_t h^{(k)} + v \cdot h^{(k+1)} =& \cc{L} h^{(k+2)} 
\quad\quad \textrm{ for } k\geq 0.
\end{split}
\end{displaymath}
We impose that $h^{(0)}$  satisfy the same initial condition of the whole solution $h_\varepsilon$, namely
\begin{displaymath}
h^{(0)}(x,0) = \rho_0(x).
\end{displaymath}
Let us start from the first equation $(i)$ $v\cdot\nabla_x h^{(0)}=\cc{L} h^{(1)}$. Thanks to the Fredholm alternative and by proceeding as in the proof of Theorem \ref{th:g_to_rho}, 
we can solve equation $(i)$
 if and only if the left hand side belongs to 
 $(\rr{Ker} \mathcal{L})^{\perp} = \{ h\in L^2(S^1)\text{ such that }\int_{S^1} \mathrm{d} v \,\, h(v) = 0 \}$. 
Since $v\cdot \nabla_x h^{(0)}(x)$ is an odd function of $v$, it belongs to $(\rr{Ker} \mathcal{L})^{\perp}$. So we can invert the operator $\mathcal{L}$ finding
\begin{equation}\label{eq:dim1_inv_i}
h^{(1)}(x,v,t)= \cc{L}^{-1} (v\cdot\nabla_x h^{(0)}(x,t)) + \zeta^{(1)}(x,t).
\end{equation}
where $\zeta^{(1)}(x,t)$ is a function to be determined in the kernel of $\cc{L}$.
Recall that $\cc{L}^{-1}$ preserves the parity.

We integrate the second equation $(ii)$ $\partial_t h^{(0)}+ v\cdot \nabla_x h^{(1)}= \cc{L}h^{(2)}$ with respect to the uniform measure on the sphere $S^1$.
Thanks to the equation \eqref{eq:dim1_inv_i}  and the observations that $\int_{S^1} \ud v\, \cc{L} h^{(2)}=0$ and  $\int_{S^1} \ud v\, v\cdot \nabla_x \zeta^{(1)}(x,t)=0$,
 it holds
\begin{equation}
\frac{1}{2\pi}\int_{S^1} \ud v\,  \partial_t h^{(0)}(x,t) + v\cdot\nabla_x (\cc{L}^{-1} v \cdot \nabla_x h^{(0)}(x,t))=0.
\end{equation}
As in the proof of Theorem \ref{th:g_to_rho} defining $D_{i,j}=\frac{1}{2\pi} \int_{S^1} \rr{d} v \,v_i (-\cc{L}^{-1})v_j$, we find that the diffusion coefficient $D$ is given by the formula \eqref{eq:greenKubo}
\begin{equation*} 
D=\frac{1}{4\pi}\int_{S^1} \mathrm{d} v \,\, v \cdot (- \mathcal{L})^{-1} v
\end{equation*} 
so that
the heat equation for $h^{(0)}$ is
\begin{equation}\label{eq:heat_eq_h0}
\partial_t h^{(0)} - D \Delta_x h^{(0)}=0.
\end{equation}

$h_\varepsilon(x,v)$ has to satisfy the reflective boundary condition $h_\varepsilon(x,v',t) = h_\varepsilon(x,v,t)$ on $\partial {\Lambda}_E$. By imposing it on the first term $h^{(1)}(x,v,t) = h^{(1)}(x,v',t)$ for every $x\in \partial\Omega_E$, $v\cdot n <0$, we obtain
proceeding in the same way of the proof of Theorem \ref{th:g_to_rho} that $h^{(0)}(x,t)$ has to satisfy the Neumann boundary conditions $\partial_n h^{(0)}(x,t)=0$, for all $x \in \partial{\Lambda}_E$.

We have so shown that  the term $h^{(0)}(x,t)$ solves the problem
\begin{equation}\label{eq:pb_heat_h0}
\begin{cases}
\partial_t h^{(0)} - \Delta_x h^{(0)}=0\qquad \quad & x\in  {\Lambda}   \\
h^{(0)}(x,0)  = \rho_0(x)  \qquad & x\in \Lambda  
\\ \partial_n h^{(0)}(x,t)=0\qquad & x \in \partial{\Lambda}_E.
\end{cases}
\end{equation}
In particular $h^{(0)}(t)\in L^{\infty}(\Lambda\times S^1)$ for any $t\geq0$.

The equation \eqref{eq:heat_eq_h0} allow us to verify that when integrating the equation $(ii)$ the left hand side vanishes. It implies that we can invert operator $\cc{L}$ finding
\begin{equation}\label{eq:h2_invert_ii}
h^{(2)}(x,v,t)=\cc{L}^{-1} (\partial_t h^{(0)}(x,t) + v\cdot \nabla_x (\cc{L}^{-1} (v\cdot\nabla_x h^{(0)}(x,t)) )  + v\cdot\nabla_x \zeta^{(1)}(x,t) ) + \zeta^{(2)}(x,t),
\end{equation}
where $\zeta^{(2)}(x,t)$ is a function in $\rr{Ker} \,\cc{L}$.

Next equation is $(iii)$ $\partial_t h^{(1)}+v \cdot\nabla_x h^{(2)}=\cc{L}h^{(3)}$. When integrating it with respect to the uniform measure on $S^1$, we exploit the fact that the operator $\cc{L}^{-1}$ preserves the parity. So, substituting $h^{(1)}$ and $h^{(2)}$  with their expressions given by \eqref{eq:dim1_inv_i} and \eqref{eq:h2_invert_ii}, the only terms surviving give the equation for $\zeta^{(1)}$
\begin{equation}\label{eq:zeta^1_heat}
\partial_t \zeta^{(1)}(x,t) - D \Delta_x \zeta^{(1)}(x,t)=0 . 
\end{equation}
Since there are no restrictions on the choice of the initial condition for $\zeta^{(1)}$, we fix $\zeta^{(1)}(x,0)=0$. So $\zeta^{(1)}(x,t)\equiv 0$ for any $(x,t)$ and the expression for $h^{(1)}$
 reduces to \begin{displaymath}
 h^{(1)}(x,v,t)= \cc{L}^{-1}(v\cdot \nabla_x h^{(0)} (x,t)).
 \end{displaymath}
By the {Lemma \ref{le:L^-1}}  and the smoothness of $h^{(0)}$ we have
\begin{displaymath}
\sup_{t\in[0,T]} \| h^{(1)}(t) \|_\infty \leq C \sup_{t\in[0,T]} \| \nabla_x h^{(0)} (t)\|_\infty <+\infty.
\end{displaymath}

In the same way, by Lemma \ref{le:L^-1} and smoothness of $h^{(0)}$ it follows that the first term in the expression of $h^{(2)}$, i.e. $h_1^{(2)}=\cc{L}^{-1} (\partial_t h^{(0)}(x,t) + v\cdot \nabla_x (\cc{L}^{-1} (v\cdot\nabla_x h^{(0)}(x,t)) ))$, is in $L^\infty([0,T];L^\infty(\Lambda\times S^1))$, as well as its spatial derivatives. 

Observe now that the left hand side of equation $(iii)$ has null integral on $S^1$ due to \eqref{eq:zeta^1_heat}. By inverting $\cc{L}$ we obtain the formula for $h^{(3)}$ 
\begin{equation}
\begin{split}
h^{(3)}(x,v,t)=& \cc{L}^{-1}(\partial_t h^{(1)} + v\cdot \nabla_x h^{(2)}(x,v,t)) + \zeta^{(3)}(x,t)
\\=& \cc{L}^{-1} (\partial_t \cc{L}^{-1} (v\cdot\nabla_x h^{(0)}(x,t))   + v\cdot\nabla_x (h_1^{(2)} (x,v,t) + \zeta^{(2)}(x,t))) + \zeta^{(3)}(x,t),
\end{split}
\end{equation}
where $\zeta^{(3)} \in \rr{Ker}\,\cc{L}$.
We integrate now the equation $(iv)$  $\partial_t h^{(2)}+v\cdot\nabla_x h^{(3)}=\cc{L}h^{(4)}$ with respect to the uniform measure on $S^1$. We find the equation for $\zeta^{(2)}(x,t)$
\begin{equation}
\partial_t \zeta^{(2)} - D\Delta_x \zeta^{(2)}= S(x,t),
\end{equation}
where the source $S(x,t)$ is given by
\begin{equation*}
\begin{split}
S(x,t) =&-\frac{1}{2\pi} \int_{S^1} \ud v\, v\cdot \nabla_x \cc{L}^{-1}(\partial_t \cc{L}^{-1}(v\cdot\nabla_x h^{(0)}(x,t)))
\\ & -\frac{1}{2\pi}\int_{S^1} \ud v \, v\cdot\nabla_x\cc{L}^{-1}(v\cdot \nabla_x h_1^{(2)}(x,v,t))).
\end{split}
\end{equation*}
We consider as initial datum $\zeta^{(2)}(x,0)=0$, so we have $\zeta^{(2)}\in L^{\infty}([0,T];L^{\infty}(\Lambda))$ and its spatial derivative as well, since $S\in L^\infty([0,T];L^\infty(\Lambda))$.

We write the the expansion truncated at order $\eta_\varepsilon^{-2}$ for the solution:
\begin{equation}\label{eq:tronc_hilb_exp_heat}
h_\varepsilon(x,v,t)= h^{(0)}(x,t) + \frac{1}{\eta_\varepsilon}h^{(1)}(x,v,t) + \frac{1}{\eta_\varepsilon^2 } h^{(2)}(x,v,t) +\frac{1}{\eta_\varepsilon} R_{\eta_\varepsilon}(x,v,t).
\end{equation}
We have shown that $h^{(i)}(t)\in L^\infty(\Lambda\times S^1)$ for $i=0,1,2$. Now we have to prove that even,the remainder $R_{\eta_\varepsilon}$ is in $L^\infty$.

The remainder  $R_{\eta_\varepsilon}$ satisfies the equation
\begin{equation}\label{eq:eq_remainder_heat}
(\partial_t + \eta_\varepsilon v\cdot\nabla_x) R_{\eta_\varepsilon} = \eta_\varepsilon^2  \cc{L} R_{\eta_\varepsilon} - T_{\eta_\varepsilon}
\end{equation}
with initial condition
\begin{equation*}
R_{\eta_\varepsilon}(x,v,0)= -h^{(1)}(x,v,0) - \frac{1}{\eta_\varepsilon} h^{(2)}(x,v,0)
\end{equation*}
and boundary conditions
\begin{equation*}
R_{\eta_\varepsilon} (x,v',t)=R_{\eta_\varepsilon}(x,v,t)\qquad x\in\partial\Lambda_E , \, v\cdot n <0.
\end{equation*}
The term $T_{\eta_\varepsilon}$ on the left hand side of \eqref{eq:eq_remainder_heat} is $T_{\eta_\varepsilon}  = \partial_t h^{(1)} + \frac{1}{\eta_\varepsilon} \partial_t h^{(2)} + v\cdot\nabla_x h^{(2)}$.
So $T_{\eta_\varepsilon} \in L^\infty ([0,T] ; L^{\infty}(\Lambda\times S^1))$ and thanks to the smoothness hypothesis  on $\rho_0$ also the initial datum $R_{\eta_\varepsilon}(x,v,0)$ belongs to $L^\infty$.

By denoting by $S_{\eta_\varepsilon}(t)$ the semigroup associated to the generator $\eta_\varepsilon (v\cdot \nabla_x - \eta_\varepsilon \cc{L})$ with reflective boundary conditions on $\partial\Lambda_E$, the equation \eqref{eq:eq_remainder_heat} becomes 
\begin{equation*}
R_{\eta_\varepsilon} (t) =  S_{\eta_\varepsilon} (t) R_{\eta_\varepsilon}(0) + \int_0^t \ud s\, S_{\eta_\varepsilon} (t-s) T_{\eta_\varepsilon} (s).
\end{equation*}
By means of the series expansion found in \eqref{eq:S_con_chi^delta},  the solution 
can be written in the following way:
\begin{equation*}
\begin{split}
&R_{\eta_\varepsilon}(x,v,t) =  \sum_{m\geq 0} e^{-2\lambda {\eta_\varepsilon^2} t}  (\lambda \eta_\varepsilon)^m \int_0^{{\eta_\varepsilon}t} \!\!\mathrm{d}t_1
\ldots \int_0^{m-1} \!\! \mathrm{d} t_m 
 \int_{-1}^1 \!\!\mathrm{d} \delta_1 \ldots \int_{-1}^1 \!\! \mathrm{d} \delta_m   R_{{\eta_\varepsilon}}(0) ( \Phi^{-{\eta_\varepsilon}t}(x,v))
\\ +& \int_0^t \!\ud s  \sum_{m\geq 0} e^{-2\lambda {\eta_\varepsilon^2} (t-s)}  (\lambda \eta_\varepsilon)^m \!\int_0^{{\eta_\varepsilon}(t-s)}\!\!\! \!\!\mathrm{d}t_1
\ldots \int_0^{m-1} \!\! \mathrm{d} t_m 
 \int_{-1}^1 \!\!\mathrm{d} \delta_1 \ldots \int_{-1}^1 \!\! \mathrm{d} \delta_m   T_{{\eta_\varepsilon}}(0) ( \Phi^{-{\eta_\varepsilon}(t-s)}(x,v),s).
 \end{split}
\end{equation*}
Therefore we can estimate
\begin{equation*}
\sup_{t\in[0,T]} \| R_{{\eta_\varepsilon}}(t) \|_\infty \leq \| R_{\eta_\varepsilon}(0) \|_\infty + T \sup_{t\in[0,T]} \| T_{\eta_\varepsilon}(t) \|_\infty \leq C <+\infty.
\end{equation*}
So the remainder is uniformly bounded too. Hence from the estimates and \eqref{eq:tronc_hilb_exp_heat} it follows that $h_\varepsilon$ converges to $h^{(0)}$ in $L^\infty$ for $\eta_\varepsilon\to \infty$.

\end{proof}


\begin{acknowledgments}
We thank Daniele Andreucci for many useful discussions
on mixed boundary problems and for having suggested the reference \cite{Lad}. We are grateful to
Mario Pulvirenti for many illuminating discussions on the topic
of the paper and for having suggested the idea of studying the
residence time problem in the framework of Kinetic Theory.
\end{acknowledgments}



\end{document}